\documentclass[11pt]{article}

\usepackage[margin = 1in]{geometry}

\usepackage{hyperref}
\usepackage{amsthm}
\usepackage{caption}
\usepackage{subcaption}
\usepackage{graphicx}
\usepackage{enumitem}
\usepackage{amsmath,amssymb,color}
\usepackage{multirow}
\hypersetup{colorlinks=true,citecolor=blue, linkcolor=blue, urlcolor=blue}

\newtheorem{theorem}{Theorem}
\newtheorem{lemma}[theorem]{Lemma}

\newtheorem{corollary}[theorem]{Corollary}
\newtheorem{remark}{Remark}

\newcommand{\ra}{\rightarrow}

\def\Bu{{\mathfrak{B}}_u}
\def\Bp{{\mathfrak{B}}_o}
\def\Bh{{\mathfrak{B}}_{rc}}
\def\betau{{{\beta}}_{\mathrm u}}
\def\betarc{{{\beta}}_{\mathrm rc}}
\def\Zrc{{{Z}}_{\mathrm rc}}
\newcommand{\Tmix}{T_{\mathrm{mix}}}

\def\eps{\varepsilon}
\def\epsilon{\varepsilon}

\newcommand{\TreeD}{\mathbb{T}_{\Delta}}

\def\vv{\ensuremath{\mathbf{v}}}
\def\u{\ensuremath{\mathbf{u}}}
\def\m{\ensuremath{\mathbf{m}}}
\def\w{\ensuremath{\mathbf{w}}}
\def\alphab{\ensuremath{\boldsymbol{\alpha}}}
\newcommand{\norm}[1]{\left\|#1\right\|}

\title{Swendsen-Wang Algorithm on the Mean-Field Potts Model\thanks{An extended abstract of this paper appeared in the proceedings of RANDOM/APPROX 2015.}}

\author{Andreas Galanis\thanks{University of Oxford,
  Wolfson Building, Parks Road, Oxford, OX1~3QD, UK.
  \texttt{andreas.galanis@cs.ox.ac.uk}.
The research leading to these results has received funding from the European Research Council under
the European Union's Seventh Framework Programme (FP7/2007-2013) ERC grant agreement no. 334828. The paper
reflects only the authors' views and not the views of the ERC or the European Commission. The European Union is not liable for any use that may be made of the information contained therein.
}
\and
 Daniel \v{S}tefankovi\v{c}\thanks{Department of Computer Science, University of Rochester,
Rochester, NY 14627.   \texttt{stefanko@cs.rochester.edu}.
Research supported in part by NSF grant CCF-1318374.}
 \and Eric Vigoda\thanks{School of Computer Science, Georgia
Institute of Technology, Atlanta GA 30332.
 \texttt{vigoda@cc.gatech.edu}.
Research supported in part by NSF grant CCF-1217458.
}
}

\begin{document}

\maketitle

\begin{abstract}
We study the $q$-state ferromagnetic Potts model  on the $n$-vertex complete
graph known as the mean-field (Curie-Weiss) model.  We analyze
 the Swendsen-Wang algorithm which is a Markov chain that utilizes the random cluster
representation for the ferromagnetic Potts model to recolor large
sets of vertices in one step and potentially overcomes obstacles
that inhibit single-site Glauber dynamics. 
Long et al. studied the case $q=2$, the Swendsen-Wang algorithm for the mean-field ferromagnetic Ising model, and showed that the mixing time satisfies: (i) $\Theta(1)$ for
$\beta<\beta_c$, (ii) $\Theta(n^{1/4})$ for $\beta=\beta_c$, (iii) $\Theta(\log n)$ for $\beta>\beta_c$, where $\beta_c$ is the critical temperature for the ordered/disordered phase transition.
In contrast, for $q\geq 3$ there are two
critical temperatures $0<\betau<\betarc$ that are relevant.
We prove that the mixing time of the
Swendsen-Wang algorithm for the ferromagnetic Potts model on the
$n$-vertex complete graph satisfies: (i) $\Theta(1)$ for
$\beta<\betau$, (ii) $\Theta(n^{1/3})$ for $\beta=\betau$, (iii) $\exp(n^{\Omega(1)})$ for $\betau<\beta<\betarc$,
and (iv) $\Theta(\log{n})$ for $\beta\geq\betarc$.
These results complement refined results of Cuff et al. on the mixing time of
the Glauber dynamics for the ferromagnetic Potts model.
\end{abstract}


\thispagestyle{empty}

\newpage

\setcounter{page}{1}

\section{Introduction}

The mixing time of Markov chains is of critical importance for
simulations of statistical physics models.  It is especially interesting
to understand how phase transitions in these models manifest in
the behavior of the mixing time; these connections are the topic of this paper.

We study the $q$-state ferromagnetic Potts model.
In the following definition the case $q=2$ corresponds to the Ising model and $q\geq 3$ is the
Potts model.  For a graph $G=(V,E)$ the configurations of the model
are assignments $\sigma:V\rightarrow [q]$ of spins to vertices; denote by $\Omega$ the set of all configurations.
The model is parameterized by $\beta>0$, known as the (inverse) temperature.
For a configuration $\sigma\in\Omega$ let $m(\sigma)$ be the
number of edges in $E$ that are monochromatic under $\sigma$ and
let its weight be $w(\sigma) = \exp(\beta m(\sigma))$.
Then the Gibbs distribution $\mu$ is defined as follows: for $\sigma\in\Omega$, $\mu(\sigma)=w(\sigma)/Z(\beta)$,
where $Z(\beta)= \sum_{\sigma\in\Omega} w(\sigma)$ is the normalizing constant, known as the
partition function.

A useful feature for studying the ferromagnetic Potts model is its alternative formulation known
as the random-cluster model.  Here configurations are subsets of edges and the
weight of such a configuration $S\subseteq E$ is
\[
w(S) = p^{|S|}(1-p)^{|E\setminus S|}q^{k(S)},
\]
where $p=1-\exp(-\beta)$ and $k(S)$ is the number of connected components in the graph $G'=(V,S)$
(isolated vertices do count).  The corresponding partition function $\Zrc = \sum_{S\subseteq E} w(S)$
satisfies $\Zrc = (1-p)^{|E|} Z$.

The focus of this paper is the Curie-Weiss model which in computer
science terminology is the $n$-vertex complete graph $G=(V,E)$.  The interest in this model
is that it allows more detailed results and these results are believed to extend to
other graphs of particular interest such as random regular graphs.
For convenience we parameterize the model in terms of a constant $B>0$ such that 
 the Gibbs distribution is as follows:
\begin{equation}\label{eq:Gibbs}
\mu(\sigma)=\frac{1}{Z(\beta)} (1-B/n)^{-m(\sigma)}.
\end{equation}
(Note that $\beta=-\ln(1-B/n)\sim B/n$ for large $n$.)  The following critical points $\Bu<\Bp<\Bh$ for the parameter $B$ are well-studied\footnote{$\Bp$ is $\beta_c$ in \cite[Equation (3.1)]{HET} and
$\Bu$ is equivalent to $\beta_s$ in \cite[Equation (1.1)]{CDLLPS}
under the parametrization $z=B (qx-1)/(q-1)$.  We
follow the convention of counting monochromatic edges~\cite{HET} as opposed to counting monochromatic
pairs of vertices~\cite{CDLLPS}; hence our thresholds are larger than those in~\cite{CDLLPS} by a factor
of~$2$.}    and relevant to our study of the Potts model on the complete graph:
\begin{gather}
\label{ehrehr}
\Bu =  \sup\Big\{ B\geq 0\,\Big|\, \frac{B-z}{B+(q-1)z} \neq {\mathrm e}^{-z} \  \mbox{for all}\ z>0\Big\}
 = \min_{z\geq 0}\left\{ z+\frac{qz}{{\mathrm e}^{z}-1}\right\}, \\
\label{eq:BpBh}
\Bp =  \frac{2(q-1)\ln(q-1)}{q-2},\qquad \Bh=q.
\end{gather}
These thresholds correspond to the critical points for the infinite $\Delta$-regular tree $\TreeD$ and random $\Delta$-regular
graphs by taking appropriate limits as $\Delta\ra\infty$.
More specifically, if $B(\Delta)$ is a
threshold on $\TreeD$ or the random $\Delta$-regular graph then
$\lim_{\Delta\rightarrow\infty} \Delta( B(\Delta)-1 )$ is the
corresponding threshold in the Curie-Weiss model.
In this perspective, $\Bu$ corresponds to the uniqueness/non-uniqueness threshold on $\TreeD$;
$\Bp$ corresponds to the ordered/disordered phase transition;
and $\Bh$ was conjectured by H\"{a}ggstr\"{o}m to correspond to a second uniqueness/non-uniqueness
threshold for the random-cluster model on $\TreeD$ with periodic boundaries (in particular, he
conjectured that non-uniqueness holds iff $B\in (\Bu,\Bh)$).  For a detailed exposition of
these critical points we refer the reader to \cite{CDLLPS} (see also \cite{GSVY} for their relevance for random regular graphs). We should finally remark that in the case of the Ising model ($q=2$), the three points $\Bu,\Bp,\Bh$ coincide.

The Glauber dynamics is a classical tool for studying the Gibbs distribution.  This is the class
of Markov chains with ``local'' transitions that update the configuration at a randomly chosen vertex and
which are designed so that the stationary distribution is the Gibbs distribution.
The limitation of local Markov chains, such as the Glauber dynamics,
is that they are typically slow to converge at low temperatures
(large $B$).   The Swendsen-Wang algorithm is a more sophisticated Markov chain
that  utilizes the random cluster representation of
the Potts model to potentially overcome bottlenecks that obstruct the simpler Glauber dynamics.

Specifically, the Swendsen-Wang algorithm is a Markov chain $(X_t)$ whose transitions $X_t\rightarrow X_{t+1}$
are as follows.  From a configuration $X_t\in \Omega$:
\begin{itemize}
\item Let $M$ be the set of monochromatic edges in $X_t$.
\item \emph{Percolation step:} for each edge $e\in M$, keep it independently with
probability $B/n$. 
Let $M'$ denote the set of the remaining monochromatic edges.
\item \emph{Coloring step:} in the graph $(V,M')$, independently for each connected component, choose a color uniformly at random from $[q]$
and assign to all vertices in that component the chosen color.
Let $X_{t+1}$ denote the resulting spin configuration.
\end{itemize}
It  is a standard fact that the chain is reversible with respect to the Gibbs distribution $\mu$ (and thus converges to it). We will be interested in the mixing time $\Tmix$ of the chain, which 
is defined as the number of steps from the worst initial state to get within total variation distance $1/4$ of the distribution $\mu$.

For the Swendsen-Wang algorithm for the ferromagnetic Ising model ($q=2$),
Cooper et al. \cite{CDFR} showed for the complete graph with $n$ vertices that  the mixing time satisfies $\Tmix=n^{1/2+o(1)}$ for all temperatures except for $\beta=\beta_c$, where $\beta_c$ is the uniqueness/non-uniqueness threshold.
Long et al. \cite{LNNP} showed more refined results for the complete graph establishing that the mixing time is
$\Theta(1)$ for $\beta<\beta_c$, $\Theta(n^{1/4})$ for $\beta=\beta_c$, and $\Theta(\log{n})$
for $\beta>\beta_c$. For square boxes of $\mathbb{Z}^2$, Ullrich \cite{Ullrich1,Ullrich2} proved that the mixing time of Swendsen-Wang is polynomial for all temperatures (building upon results for the Glauber dynamics  by Martinelli and Olivieri \cite{MOI,MOII}  and Lubetzky and Sly \cite{LubetzkySly}). Very recently, Guo and Jerrum \cite{GuoJerrum} showed that the mixing time of Swendsen-Wang is polynomial for any graph $G$ for all temperatures.

For the Swendsen-Wang algorithm
for the ferromagnetic Potts model ($q\geq 3$), it has been demonstrated that the mixing time can be of order $\exp(n^{\Omega(1)})$ at the ordered/disordered phase transition point (\emph{phase coexistence}). In particular, Gore and Jerrum \cite{GJ} showed for the complete graph that the mixing time is $\exp(\Omega(\sqrt{n}))$
at the critical point $B=\Bp$. Similar slow mixing results have been established for other classes of graphs at the analogous critical point: Cooper and Frieze \cite{CF} showed this for $G(n,p)$ when $p=\Omega(n^{-1/3})$, Galanis et al. \cite{GSVY} for random $\Delta$-regular graphs when $q\geq 2\Delta/\log \Delta$,
and Borgs et al. \cite{BCFKVV,BCT} for the $d$-dimensional integer lattice for $q\geq 25$. For square boxes of $\mathbb{Z}^2$, Ullrich \cite{Ullrich1,Ullrich2} proves polynomial mixing time at all temperatures except criticality building upon the results of Beffara and Duminil-Copin \cite{Beffara}. On the torus $(\mathbb{Z}/n\mathbb{Z})^2$, Gheissari and Lubetzky \cite{GL} recently showed the following bounds on the mixing time at criticality: polynomial upper bound for $q=3$, quasi-polynomial upper bound for $q=4$ and  exponential lower bound for $q>4$. 

In this paper, we study the mixing time of the Swendsen-Wang dynamics for the ferromagnetic Potts model on the complete graph. Previously, Cuff et al. \cite{CDLLPS} had detailed the mixing time of the Glauber dynamics  for the ferromagnetic Potts model on the complete graph (their results are significantly more precise than what we state here for convenience):
$\Theta(n\log{n})$ for $B<\Bu$, $\exp(\Omega(n))$ for $B>\Bu$,
and $\Theta(n^{4/3})$ mixing time for $B=\Bu$ (and a scaling window of $O(n^{-2/3})$ around $\Bu$).

Our main result is a complete classification of the mixing time of the Swendsen-Wang dynamics on the complete graph when the parameter $B$ is a constant independent of $n$.
\begin{theorem}\label{thm:main}
For all integer $q\geq 3$, the mixing time $\Tmix$ of the Swendsen-Wang algorithm on the $n$-vertex complete graph satisfies:
\begin{enumerate}
\item \label{thm:fast-uniq}
For all $B<\Bu$, $\Tmix = \Theta(1)$.
\item \label{thm:critical} For $B=\Bu$, $\Tmix=\Theta(n^{1/3})$.
\item \label{thm:slow} For all $\Bu< B<\Bh$, $\Tmix = \exp(n^{\Omega(1)})$.
\item \label{thm:fast-nonuniq}
For all $B\geq\Bh$, $\Tmix = \Theta(\log{n})$.
\end{enumerate}
\end{theorem}
In an independent work, Blanca and Sinclair \cite{BS} analyze a closely related chain to the Swendsen-Wang dynamics, known as the Chayes-Machta dynamics, which is also suitable for sampling random cluster configurations (works more generally for $q\geq 1$ with $q\in \mathbb{R}$). They provide an analogue of Theorem~\ref{thm:main}, though their analysis excludes the critical points $B=\Bu$ and $B=\Bh$. Very recently, Gheissari, Lubetzky, and Peres \cite{GLP} improved the lower bound on the mixing time in the window $\Bu< B<\Bh$ to $\exp(\Omega(n))$, both for the Swendsen-Wang and the Chayes-Machta dynamics.

In the following section, we give an overview of our proof approach. First, we discuss the critical points $\Bu,\Bp,\Bh$ in more detail. Then, we
present a function $F$ which captures a simplified view of the Swendsen-Wang
dynamics, and then we connect the behavior of $F$
with the critical points.  We also present in Section \ref{sec:proof-overview}
a high-level sketch of the proof of Theorem~\ref{thm:main}. In Section~\ref{sec:randomgraph}, we collect facts for the $G(n,c/n)$ random graph which will be relevant for analyzing one step of the Swendsen-Wang algorithm. 
In Section \ref{sepp}, we prove the slow mixing result (Part \ref{thm:slow}
of Theorem \ref{thm:main}).  We then prove the rapid mixing results for
$B>\Bh$ in Section \ref{sec:fast}, for $B=\Bh$  in Section
\ref{sec:fast-Bh}, for $B<\Bu$ in Section \ref{sec:fast-uniqueness}, and for $B=\Bu$ in Section \ref{sec:fast-Bu}.

\section{Proof Approach} \label{sec:proof-overview}

\subsection{Critical Points for Phase Transitions}\label{sec:crphasetr}
In this section, we review the thresholds $\Bu,\Bp,\Bh$ for the mean-field Potts model and their connections to the critical points of the partition function which will be relevant later. The reader is referred to \cite{BGJ, HET} for further details (\cite{BGJ} also applies to the random-cluster model).  

We first need to introduce some notation for the complete graph $G=(V,E)$ with $n$ vertices. For a configuration $\sigma:V\rightarrow[q]$ and a color $i\in [q]$, let $\alpha_i(\sigma)$ be the fraction of vertices with color $i$ in $\sigma$, i.e., $\alpha_i(\sigma) = |\{v\in V: \sigma(v) = i\}|/n$. We denote by $\alphab(\sigma)$ the vector $(\alpha_1(\sigma),\dots,\alpha_q(\sigma))$,
and refer to it as the  {\em phase} of  $\sigma$.  There are $q+1$ phases that are most relevant,  the {\em uniform} phase $\u:=(1/q,\dots,1/q)$
and the $q$ permutations of the {\em majority} phase $\m:=(a,b,\dots,b)$, for some appropriate $a>1/q$ and $b$ given by $a+(q-1)b=1$. Roughly, these phases correspond to the configurations that have dominant contribution to the partition function. 

More precisely, for a $q$-dimensional probability vector $\alphab$,  let $\Omega^{\alphab}$ be  the set of configurations $\sigma$ whose phase is  $\alphab$.\footnote{Technically, for integrality reasons, $\Omega^{\alphab}$ are the  configurations $\sigma$ whose phase  is within $O(1/n)$ from $\alphab$. This does not have any effect in the subsequent asymptotic considerations.}   Let
\begin{equation*} Z^{\alphab} = \sum_{\sigma\in\Omega^{\alphab}} w(\sigma)\mbox{ and } \Psi(\alphab) := \lim_{n\ra\infty} \frac{1}{n}\ln Z^{\alphab}.
\end{equation*}
To simplify the formulas, it turns out that it is enough to consider the following one-dimensional version of $\Psi$ corresponding to configurations where one color has density $\alpha$ and the remaining colors have density $\beta$ where $\alpha+(q-1)\beta=1$. Namely, let
\[\Psi_1 (\alpha):=\Psi(\alpha,\beta,\dots,\beta)=\Psi\Big(\alpha,\frac{1-\alpha}{q-1},\dots,\frac{1-\alpha}{q-1}\Big).\]
It is not hard to see that $Z^{\alphab}$ is given by $\binom{n}{\alpha_1n,\hdots,\alpha_q n}(1-B/n)^{-\sum_{i\in[q]}\binom{\alpha_i n}{2}}$, so using Stirling's approximation we obtain the explicit expression
\begin{equation}\label{hrz1}
\Psi_1 (\alpha)= - \alpha\ln \alpha - (1-\alpha)\ln\frac{1-\alpha}{q-1}+\frac{B}{2}\Big(\alpha^2 + \frac{(1-\alpha)^2}{q-1}\Big).
\end{equation}

With these definitions, we next relate the thresholds $\Bu,\Bp,\Bh$ to the critical points/local maxima of $\Psi_1$. Depending on the value of $B$, there are two points that are relevant, $u=1/q$ and $a>1/q$, where $a$
is a critical point of $\Psi_1$ and hence satisfies\footnote{Such a critical point $a>1/q$ exists when $B\geq \Bu$ (see Lemma~\ref{lem:Ncritical}). In the regime $\Bu<B<\Bh$ there are two critical points of $\Psi_1(a)$ with value $a>1/q$; the relevant value of $a$ is then given by the point where $\Psi_1(\alpha)$ has a local maximum, see Lemma~\ref{lem:psi1local} for details and Figure~\ref{fig:psi} for a depiction.} 
\begin{equation}\label{hrz2}
\ln\frac{(q-1)a}{1-a} = B \frac{qa-1}{q-1}.
\end{equation}
The following folklore lemma illustrates the relevant thresholds,  see also Figure~\ref{fig:psi}. For completeness, we give the proof in Section~\ref{sec:locmaxpsi}.
\begin{lemma}\label{lem:psi1local}
Let $q\geq 3$. For the function $\Psi_1$,
\begin{enumerate}
\item
For $B<\Bu$, $\Psi_1$ has a unique local maximum, at $u=1/q$, and there are no other critical points of $\Psi_1$. 
\item For $B=\Bu$, $\Psi_1$ has two critical points, at $u=1/q$ and $a>1/q$ (satisfying \eqref{hrz2}). Of these, $u=1/q$ is the only local maximum of $\Psi_1$.
\item For $\Bu<B<\Bh$, $\Psi_1$ has two local maxima, at $u=1/q$ and $a>1/q$ \mbox{(satisfying \eqref{hrz2}). Further,}
\begin{itemize}
\item[--] For $B\in(\Bu,\Bp)$, $u$ is the only global maximum of $\Psi_1$.
\item[--] For $B=\Bp$, $u$ and $a$ are the global maxima of $\Psi_1$.
\item[--] For $B\in(\Bp,\Bh)$, $a$ is the only global maximum of $\Psi_1$.
\end{itemize}
\item For $B \geq \Bh$, $\Psi_1$ has one local maximum in the interval $[1/q,1]$, at a point  $a>1/q$ (satisfying \eqref{hrz2}).
\end{enumerate}
\end{lemma}

While we will not need the following fact explicitly in our arguments, we remark for the sake of completeness that the local maxima of the multivariable function $\Psi$ correspond to  the local maxima of the function $\Psi_1$ as follows. The phases where $\Psi$ can have a local maximum is  the uniform phase $\u=(1/q,\dots,1/q)$
and the $q$ permutations of the majority phase $\m=(a,b,\dots,b)$, where $a>1/q$ is as in Lemma~\ref{lem:psi1local} and $b$ is given by $a+(q-1)b=1$. More precisely, $\u$ is a local maximum of $\Psi$  iff $u=1/q$ is a local maximum of $\Psi_1$, the majority phase $\m$ is a local maximum of $\Psi$  iff $a$ is a local maximum of $\Psi_1$, and there are no other local maxima of $\Psi$. Both $\u$ and $\m$ are global maxima of $\Psi$ only at the point $B=\Bp$.
\begin{figure}[t]
\begin{subfigure}{0.33\textwidth}
  \centering
  \scalebox{0.55}[0.55]{\includegraphics[viewport=0 0 260 190,clip]{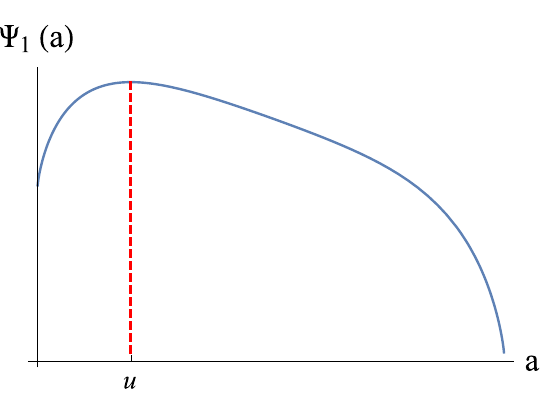}}
  \caption{$B<\Bu$}
  \label{fig:psiuniqueness}
\end{subfigure}%
\begin{subfigure}{0.33\textwidth}
  \centering
  \scalebox{0.55}[0.55]{\includegraphics[viewport=0 0 260 190,clip]{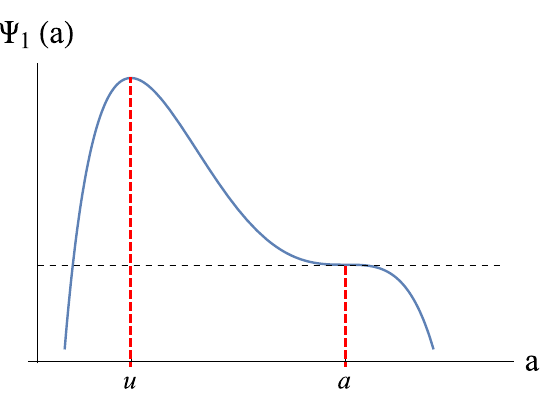}}
  \caption{$B=\Bu$}
  \label{fig:PsiBu}
\end{subfigure}
\begin{subfigure}{0.33\textwidth}
  \centering
  \scalebox{0.55}[0.55]{\includegraphics[viewport=0 0 260 190,clip]{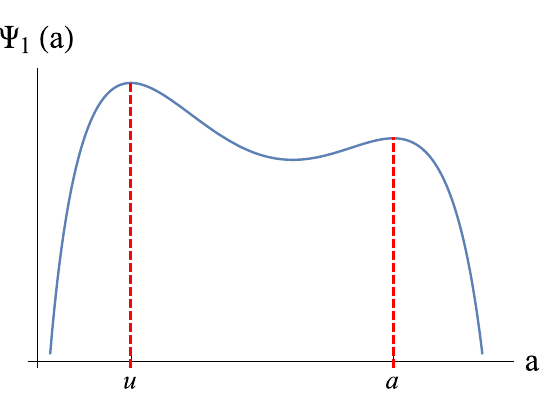}}
  \caption{$\Bu<B<\Bp$}
  \label{fig:PsiBlessBo}
\end{subfigure}
\begin{subfigure}{0.33\textwidth}
  \centering
  \scalebox{0.55}[0.55]{\includegraphics[viewport=0 0 260 190,clip]{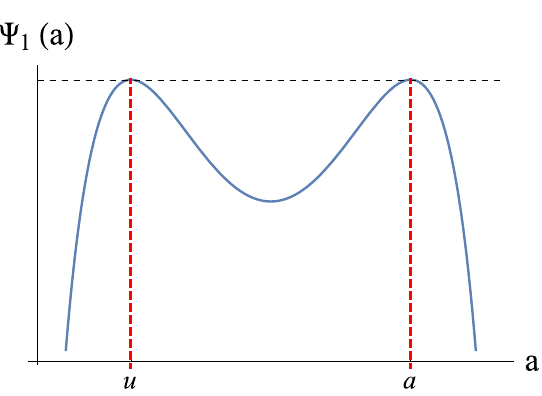}}
  \caption{$B=\Bp$}
  \label{fig:PsiBo}
\end{subfigure}
\begin{subfigure}{0.33\textwidth}
  \centering
  \scalebox{0.55}[0.55]{\includegraphics[viewport=0 0 260 190,clip]{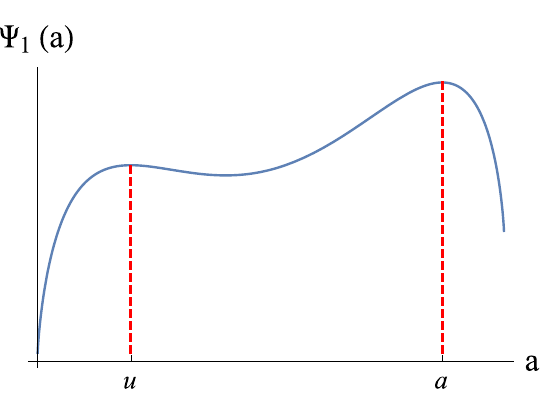}}
  \caption{$\Bp<B<\Bh$}
  \label{fig:PsigreaterBo}
\end{subfigure}
\begin{subfigure}{0.33\textwidth}
  \centering
  \scalebox{0.55}[0.55]{\includegraphics[viewport=0 0 260 190,clip]{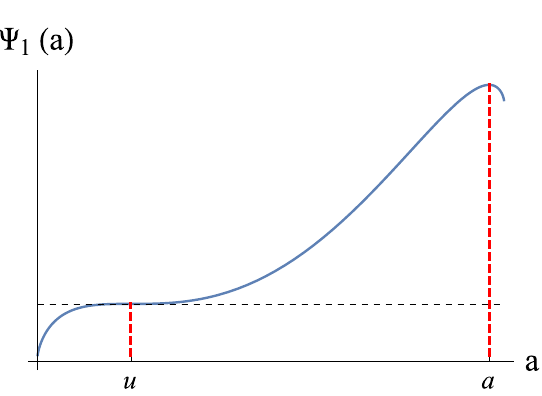}}
  \caption{$B= \Bh$}
  \label{fig:PsiBrc}
\end{subfigure}
\caption[.]{The function $\Psi_1$ (free energy) plotted in different regimes of $B$ (defined in \eqref{hrz1}). The critical points $\Bu,\Bp,\Bh$ are given by \eqref{ehrehr} and \eqref{eq:BpBh}. In the regime $B<\Bu$ (figure \ref{fig:psiuniqueness}), the function $\Psi_1$ has a unique local maximum at the disordered phase. At $B=\Bu$ (figure \ref{fig:PsiBu}), the function $\Psi_1$ has a saddle point at the ordered phase. In the regime $\Bu<B<\Bh$ (figures \ref{fig:PsiBlessBo}, \ref{fig:PsiBo} and \ref{fig:PsigreaterBo}) the function $\Psi_1$ has two local maxima; these  are both global maxima iff $B=\Bp$. In the regime $B\geq \Bh$ (figure \ref{fig:PsiBrc}), the function $\Psi_1$ has a unique local maximum in the interval $[1/q,1]$.
}
\label{fig:psi}
\end{figure}

\subsection{Connections to Simplified Swendsen-Wang}\label{sec:defF}
The following function\footnote{The argument of $F$ will typically be the density of the largest color class --- we could have extended the domain of the function $F$ to be the interval $[0,1]$ by further defining the value of $F$ in the interval $[0,1/B)$ to be $1/q$.} from $[1/q,1]$ to $[0,1]$ will capture the behavior of the Swendsen-Wang algorithm. Namely, let
\begin{equation}\label{defa}
F(z) := \frac{1}{q} + \left(1 - \frac{1}{q} \right) z x,
\end{equation}
where $x=0$ for $z\leq 1/B$ and for $z> 1/B$, $x\in (0,1]$ is the (unique) solution of
\begin{equation}\label{dexa}
x + \exp(-zB x) = 1.
\end{equation}
The function $F$ captures the size of the largest color class when there is a single heavy color where
heavy means that the color class is supercritical in the percolation step of the Swendsen-Wang process.
Hence after the percolation step this heavy color will have a giant component and the other color classes will
all be broken into small components.  So say initially the one heavy color has size $zn$ for $1/B<z<1$
and let's consider its size after one step of the Swendsen-Wang dynamics.  After the percolation step,
this heavy color will have a giant component of size roughly  $xzn$ (where $x$ is as in  \eqref{dexa})
and all other components will be of size $O(\log{n})$.  Then, a $1/q$ fraction of the small components
will be recolored the same as the giant component, and hence the size of the largest color class will be (roughly) $nF(z)$
after this one step of the Swendsen-Wang dynamics.

Our next goal is to tie together the functions $F$ and $\Psi_1$ so that we
can relate the behavior of the Swendsen-Wang dynamics with the underlying phase transitions of the model.
We first need some terminology.  A critical point $a$ of a function $f:{\mathbb R}\rightarrow{\mathbb R}$ is a
{\em hessian maximum} if the second derivative of $f$ at $a$ is negative (this is a sufficient condition for $a$ to be a local maximum). For an integer $n\geq 1$, we will denote by $f^{(n)}$ the $n$-th iterate of the function $f$. A fixpoint $a$ of $f$ is {\em attractive} if there exists $\delta>0$ such that for all $x\in(a-\delta,a+\delta)$ it holds that $f^{(n)}(x)\rightarrow a$; it is {\em repulsive} otherwise. The fixpoint $a$ is {\em jacobian attractive} if $|F'(a)|<1$; this is a sufficient condition for $a$ to be attractive. The fixpoint $a$ is {\em jacobian repulsive} if $|F'(a)|>1$; this is a sufficient condition for $a$ to be repulsive. 

Our first lemma connects the local maxima of $\Psi_1$ with the attractive fixpoints of $F$ (we restrict our attention to the interval $[1/q,1]$ since the function $F$ will be considered only in this interval).
\begin{lemma}\label{conn}
In the interval $[1/q,1]$, the hessian maxima of $\Psi_1$ correspond
to jacobian attractive fixpoints of $F$.
\end{lemma}
Lemma~\ref{conn} is proved in Section~\ref{sec:connection}. A relevant fact we should remark here and we will prove later is that, in the half-open interval $(1/q,1]$,  the critical points of $\Psi_1$ correspond to fixpoints of $F$ (see Lemma~\ref{lem:halfopen}); this actually holds for the left endpoint $1/q$ as well but only when $B\leq \Bh$ (for $B>\Bh$, $1/q$ is a critical point of $\Psi_1$ but not a fixpoint of $F$, see Lemma~\ref{lem:jacobianuniform}).

The second lemma studies the behavior of $F$ around the fixpoints and it is the main tool for proving Theorem \ref{thm:main}. Recall the earlier discussion of the uniform vector $\u:=(1/q,\dots,1/q)$ and the $q$ permutations of the majority phase $\m:=(a,b,\dots,b)$, where $a>1/q$ is as in Lemma~\ref{lem:psi1local}.
The following lemma (proved in Section \ref{sec:fixpoints}) provides some basic intuition about the proof of Theorem \ref{thm:main}, as we shall explain shortly.  A depiction of the various regimes is given in Figure~\ref{fig:F}.
\begin{lemma}
\label{lem:fixpoints}
Let $q\geq 3$. For the function $F$,
\begin{enumerate}
\item
For $B<\Bu$, $u=1/q$ is the unique fixpoint and it is jacobian attractive.
\item For $B=\Bu$, there are 2 fixpoints: $u$ and $a$. Of these, only $u$ is (jacobian) attractive. The fixpoint $a$ is repulsive but not jacobian repulsive.
\item
For $\Bu< B < \Bh$ there are 2 jacobian attractive fixpoints: $u$ and $a$. 
\item 
For $B=\Bh$, there are 2 fixpoints: $u$ and $a$. The fixpoint $u$ is jacobian repulsive, while the fixpoint $a$ is jacobian attractive. 
\item For $B>\Bh$, $a$ is the only fixpoint and it is jacobian attractive.
\end{enumerate}
\end{lemma}
The reason that $u$ abruptly changes from a jacobian  attractive fixpoint ($B<\Bh$) to a jacobian repulsive fixpoint ($B=\Bh$) stems from the fact that in the regime $B<\Bh$, $F$ is constant in a small neighborhood around $1/q$ (precisely, in the interval $[1/q,1/B]$),  which is no longer the case for $B=\Bh$.
\begin{figure}[h]
\begin{subfigure}{0.33\textwidth}
  \centering
  \scalebox{0.57}[0.57]{\includegraphics[viewport=0 0 260 190,clip]{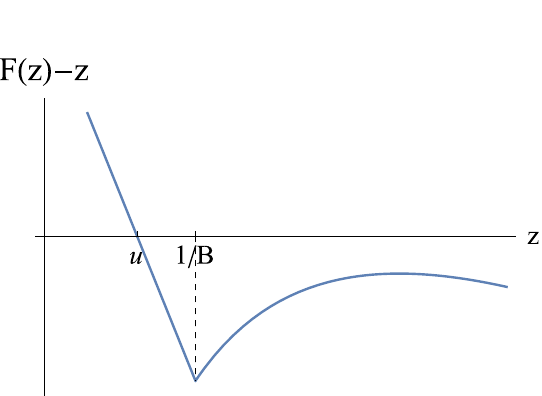}}
  \caption{$B<\Bu$}
  \label{fig:Funiqueness}
\end{subfigure}%
\begin{subfigure}{0.33\textwidth}
  \centering
  \scalebox{0.57}[0.57]{\includegraphics[viewport=0 0 260 190,clip]{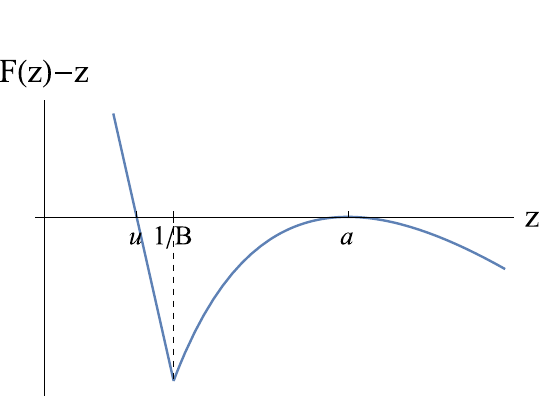}}
  \caption{$B=\Bu$}
  \label{fig:FBu}
\end{subfigure}
\begin{subfigure}{0.33\textwidth}
  \centering
  \scalebox{0.57}[0.57]{\includegraphics[viewport=0 0 260 190,clip]{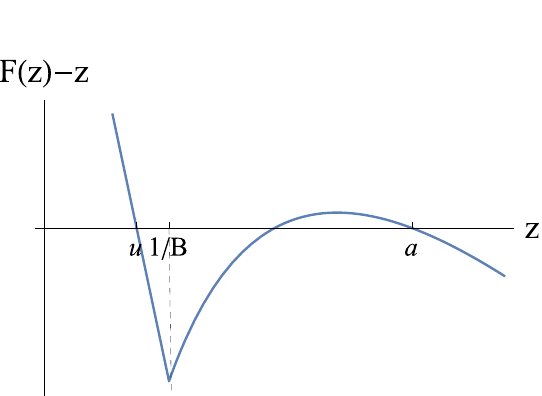}}
  \caption{$\Bu<B<\Bh$}
  \label{fig:FBlessBo}
\end{subfigure}
\vskip 0.2cm\noindent
\begin{subfigure}{0.5\textwidth}
  \centering
  \scalebox{0.57}[0.57]{\includegraphics[viewport=0 0 260 190,clip]{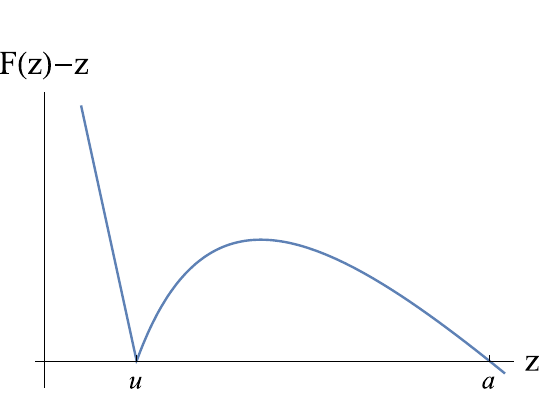}}
  \caption{$B= \Bh$}
  \label{fig:FBrc}
\end{subfigure}
\begin{subfigure}{0.5\textwidth}
  \centering
  \scalebox{0.57}[0.57]{\includegraphics[viewport=0 0 260 190,clip]{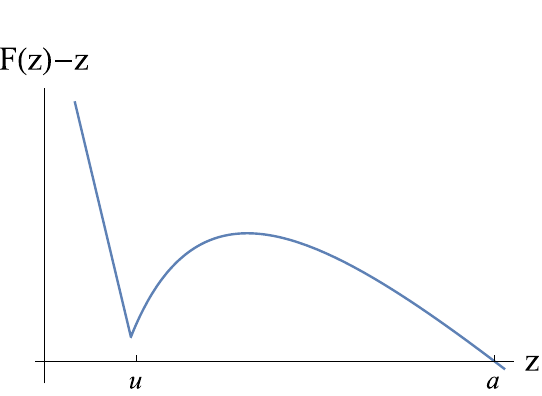}}
  \caption{$B> \Bh$}
  \label{fig:FBlargerBrc}
\end{subfigure}
\caption[.]{The drift function $F(z)-z$, where $F$ is defined by \eqref{defa}, \eqref{dexa}. The critical points $\Bu,\Bp,\Bh$ are given by \eqref{ehrehr} and \eqref{eq:BpBh}. In the regime $B<\Bu$ (figure \ref{fig:Funiqueness}), the function $F$ has a unique attractive fixpoint at the disordered phase.  At $B=\Bu$ (figure \ref{fig:FBu}), $F$ also has a (non-jacobian) \emph{repulsive} fixpoint at the ordered phase. In the regime $\Bu<B<\Bh$ (figures \ref{fig:FBlessBo}), $F$ has attractive fixpoints at the ordered and disordered phases. At $B=\Bh$ (figure \ref{fig:FBrc}), the disordered phase is no longer attractive; it is jacobian repulsive.  Finally, in the regime $B>\Bh$ (figure \ref{fig:FBlargerBrc}), the function $F$ has a unique attractive fixpoint at the ordered phase.}
\label{fig:F}
\end{figure}

\subsection{Proof Sketches}\label{sec:proofsketches}

We explain the high-level proof approach for the various parts of Theorem~\ref{thm:main}
before presenting the detailed proofs in subsequent sections.

{\bf Slow mixing for $B\in(\Bu, \Bh)$:}  The main idea
is that the function $F$ has 2 attractive fixpoints (see Lemma \ref{lem:fixpoints}).    At least one of the
corresponding phases, $\u$ or $\m$, is a global maximum for $\Psi$.  Consider the other phase, say it
is $\u$ for concreteness.
Consider the local ball around $\u$, these are configurations that are close in $\ell_\infty$ distance
from $\u$.  The key is that since $u=1/q$ is an attractive fixpoint for $F$, if the initial state is in this local ball around $\u$ 
then with very high probability after one step of the Swendsen-Wang dynamics it will still be in the local ball
(see Lemma \ref{l2}, and Lemma \ref{l3} for the analogous lemma for $\m$).
The result then follows since one needs to sample from the local ball around the phase which corresponds
to the global maximum of $\Psi$ to get close to the stationary distribution. The full argument is given in Section~\ref{sepp}.

{\bf Fast mixing for $B>\Bh$:}  For a configuration $\sigma$ and spin $i$, say the color class is heavy
if the number of vertices with spin $i$ is $>n/B$ and light if it is $<n/B$.  If a color class is heavy then
it is supercritical for the percolation step of Swendsen-Wang and hence there will be a giant component.
The key is that for any initial state $X_0$,  with constant probability, the largest components from all of the colors will choose the same new color
and consequently there will be only one heavy color class and the other $q-1$ colors will be light.
Hence we can assume there is one heavy color class and $q-1$ light color classes, and then the function $F$
suitably describes the size of the largest color class during the evolution of the Swendsen-Wang dynamics.
Since the only fixpoint of $F$ corresponds to the majority phase $\m$, after $O(\log{n})$ steps
we'll be close to $\m$ -- the difference will be due to the stochastic nature of the process.  The remaining bit of the proof is then to define a coupling for two chains $(X_t,Y_t)$ whose initial states $X_0,Y_0$
are close to $\m$ so that after $T=O(\log{n})$ steps we have that $X_T=Y_T$ (this latter part is fairly standard). The proof of the upper bound on the mixing time is given in Section~\ref{sec:fast}; the lower bound on the mixing time is proved in Section~\ref{sec:lowerboundBBh}.

{\bf Fast mixing for $B=\Bh$:}  The basic outline is similar to the $B>\Bh$ case except here
the argument is more intricate when the heaviest color lies in the scaling window (for the onset of a giant component). We need a more involved argument that we get away from initial configurations
that are close to the uniform phase; informally, the uniform fixpoint of $F$ is jacobian repulsive, so an initial displacement increases geometrically by a constant factor. The proof of the upper bound on the mixing time is given in Section~\ref{sec:fast-Bh}; the lower bound on the mixing time is proved in Section~\ref{sec:lowerboundBBh}.

{\bf Fast mixing for $B<\Bu$:}
Here the argument is similar to the $B>\Bh$ case; namely, the evolution of the density of the largest color class is captured by the iterates of the function $F$. Now, the only fixpoint of $F$ corresponds to the uniform phase $\u$, so after $O(\log{n})$ steps
the chain will get close to $\u$. In fact, this bound can now be improved to $O(1)$ steps: once the largest color class reaches
density $<1/B$ (which happens in $O(1)$ steps), then in the next step the chain jumps close to $1/q$,  i.e., we get close to the uniform phase \emph{abruptly}; this is the reason that the mixing time for $B<\Bu$ is $O(1)$ rather than $O(\log n)$. Once we are close to the uniform phase, we can then adapt a symmetry argument of \cite{LNNP} to show that we can couple two copies of the SW chain in one more step.  The details can be found in Section~\ref{sec:fast-uniqueness}.

{\bf Fast mixing for $B=\Bu$:} This is the most difficult part. As
in the $B>\Bh$ case with constant probability there will be at
most one heavy color class after one step. We then track the
evolution of the size of the heavy color class. The difficulty
arises because the size of the component does not decrease in
expectation at the majority fixpoint.  However variance moves the
size of the component into a region where the size of the
component decreases in expectation. The formal argument uses a
carefully engineered potential function  that decreases because of
the variance (the function is concave around the fixpoint) and
expectation (the function is increasing) of the size of the
largest color class, see Section~\ref{sec:fast-Bu}.

\section{Phases of the Gibbs distribution and stability analysis of fixpoints of $F$}
\subsection{Analysis of the local maxima of $\Psi_1$: Proof of Lemma \ref{lem:psi1local}}\label{sec:locmaxpsi}
In this section, we analyze the critical points/local maxima of $\Psi_1$ and prove Lemma \ref{lem:psi1local}.

The following formulas for the derivatives of $\Psi_1$ will be useful: 
\begin{equation}\label{sddr}
\Psi_1'(\alpha)=-\ln\frac{(q-1)\alpha}{1-\alpha}+B\frac{q\alpha-1}{q-1},\quad 
\Psi_1''(\alpha)= B \frac{q}{q-1} - \frac{1}{\alpha(1-\alpha)}.
\end{equation}
Recall that at a critical $a$ of $\Psi_1$ it holds that $\Psi_1'(a)=0$ and hence $a$ satisfies
\begin{equation*}\tag{\ref{hrz2}}
\ln\frac{(q-1)a}{1-a} = B \frac{qa-1}{q-1}.
\end{equation*}

We will need the following bound on the critical points of $\Psi_1$ in the interval $(1/q,1]$.
\begin{lemma}\label{lehehe}
Let $a>1/q$ be a critical point of $\Psi_1$ and $b=(1-a)/(q-1)$. Then $aB>1$ and $bB<1$.
\end{lemma}
\begin{proof}
Since $a>1/q$, there is $z>0$ such that $a=(z+1)/(z+q)$. Equation~\eqref{hrz2} becomes
\begin{equation}\label{hrz3}
\ln (1+z) = \frac{zB}{z+q}.
\end{equation}
Then, using \eqref{hrz3}, we have that 
\begin{gather*}
aB = \frac{B(z + 1)}{z+q} = (1+1/z) \ln(1+z)>1,\\
bB = (1-a)B/(q-1) = \frac{B}{z+q} = \frac{1}{z}\ln(1+z)<1.
\end{gather*}
where the inequalities hold for any $z>0$. This finishes the proof.
\end{proof}

\begin{lemma}\label{nononoz}
Let $B>\Bu$. A critical point $a>1/q$ of $\Psi_1$ has non-zero second derivative.
\end{lemma}
\begin{proof}
For the sake of contradiction, let $a>1/q$ be a critical point of $\Psi_1$ such that $\Psi_1''(a)=0$. Using \eqref{sddr}, $\Psi_1''(a)=0$ yields   that $1/q = 1-Ba(1-a)$. Plugging the value of $q$
into~\eqref{hrz2} we obtain
\begin{equation}\label{skk}
\ln\frac{Ba^2}{1-Ba(1-a)} = \frac{Ba-1}{a}.
\end{equation}
Let $w=B-1/a$. Since $a>1/q$,  by Lemma~\ref{lehehe} we have $w>0$. Equation~\eqref{skk} becomes
\begin{equation}\label{skk2}
\ln (1-w(1-w/B)) = - w.
\end{equation}
We thus obtain the following parametrization of $B,a,q$ in terms of $w$:
\begin{equation}\label{ery}
B = \frac{w^2}{{\mathrm e}^{-w} + w - 1},\quad
a=\frac{1}{1-{\mathrm e}^{-w}} - \frac{1}{w},\quad
q=\frac{{\mathrm e}^w + {\mathrm e}^{-w}-2}{{\mathrm e}^{-w} + w - 1}.
\end{equation}
Since $B>\Bu$, by the definition~\eqref{ehrehr} of the threshold $\Bu$, there exists $B'<B$ and $z>0$ such that
\[\frac{B'-z}{B'+(q-1)z} = \exp(-z)\]
and hence
\begin{equation}\label{dddd}
B' = z + \frac{qz}{{\mathrm e}^z - 1}.
\end{equation}
We will now prove that, for $B$ and $q$ as in \eqref{ery}, for any $z>0$ we have
\begin{equation}\label{dddd2}
B \leq z + \frac{qz}{{\mathrm e}^z - 1},
\end{equation}
contradicting~\eqref{dddd} and $B'<B$.

To prove \eqref{dddd2}, our goal is to show that for any $w>0$ and any $z>0$
\begin{equation*}
\frac{w^2}{{\mathrm e}^{-w} + w - 1} \leq z + \frac{{\mathrm e}^w + {\mathrm e}^{-w}-2}{{\mathrm e}^{-w} + w - 1} \frac{z}{{\mathrm e}^z - 1}.
\end{equation*}
Since ${\mathrm e}^{-w} + w - 1>0$ and ${\mathrm e}^z - 1>0$, multiplying out this inequality yields the equivalent 
\begin{equation}\label{dddd4}
0 \leq   z ({\mathrm e}^z - {\mathrm e}^w)({\mathrm e}^{-w} - 1) - w (w - z) ({\mathrm e}^z - 1) =: G_1(w,z).
\end{equation}
We have
$$
G_1(s+y,2s)=(s^2-y^2)({\mathrm e}^{2s}-1)
-2s({\mathrm e}^s-{\mathrm e}^y)({\mathrm e}^s-{\mathrm e}^{-y})=:G_2(s,y).
$$
We will show $G_2(s,y)\geq 0$ for all $s>0$ and $y\geq -s$. We have
$G_2(s,-s) = 0$ and $\lim_{y\rightarrow\infty} G_2(s,y) = \infty$.
Thus it is enough to explore the critical points of $G_3(y):=G_2(s,y)$
for each $s$. We have
$$
\frac{\partial}{\partial y} G_3 (y) = 2{\mathrm e}^s \Big(s ({\mathrm e}^y - {\mathrm e}^{-y}) - y({\mathrm e}^s - {\mathrm e}^{-s}) \Big).
$$
The function $y\mapsto ({\mathrm e}^y - {\mathrm e}^{-y})/y$ is
monotone for $y\geq 0$ (this follows from the series expansion)
and hence the only critical points of $G_3(y)$ are $y=0$ and
$y=\pm s$. For $y=\pm s$ we have $G_3(y)=0$. For $y=0$ we have
$$
G_3(0) = s^2({\mathrm e}^{2s}-1)
-2s({\mathrm e}^s-1)^2 = \sum_{i=5}^{\infty} \frac{2^i(i-5)+16}{4(i-1)!} s^i > 0.
$$
This establishes non-negativity of $G_3(y)$ for $y\geq -s$ for all
$s>0$. This completes the proof of~\eqref{dddd2} and hence the proof of the lemma.
\end{proof}

The following lemma details the number of critical points of $\Psi_1$ in the interval $(1/q,1]$.
\begin{lemma}\label{lem:Ncritical}
Let $N$ be the number of critical points of $\Psi_1$ in the interval $(1/q,1]$. Then,
\begin{enumerate}
\item\label{it:34w1a} for $B<\Bu$, $N$ equals $0$, 
\item\label{it:34w1b} for $B=\Bu$, $N$ equals $1$,
\item\label{it:34w1c} for $\Bu<B< \Bh$, $N$ equals $2$, 
\item\label{it:34w1d} for $B\geq \Bh$, $N$ equals $1$.
\end{enumerate}
\end{lemma}
\begin{proof}
Consider the function $g(z)=z+\frac{qz}{{\mathrm e}^{z}-1}$ for $z\geq 0$. By the definition \eqref{eq:BpBh} of $\Bu$, we have that \[\Bu=\min_{z\geq 0}g(z).\]
Since $g''(z)=\frac{ q{\mathrm e}^z ({\mathrm e}^z (z-2) + z+2)}{({\mathrm e}^z-1)^3}$ and $\lim_{z\downarrow 0}g'(z)=1-q/2$, we have that $g(z)$ is a convex function of $z$ and that, for $q\geq 3$, its minimum is attained (uniquely) at a point $z_0>0$.

Note that if $B\geq \Bu$ and $z>0$ satisfy $B=g(z)$, then $\alpha=\frac{{\mathrm e}^z}{{\mathrm e}^z+q-1}> 1/q$ is a critical point of $\Psi_1$ (cf. \eqref{hrz2}); similarly a critical point of $\Psi_1$ in the interval $(1/q,1]$ yields  $z>0$ such that $B=g(z)$. It follows that for $B<\Bu$, $\Psi_1$ has no critical point in the interval $(1/q,1]$. Since $\lim_{z\uparrow +\infty}g(z)=+\infty$ and $g$ is continuous, we have that for $B\geq \Bu$, $\Psi_1$ has at least a critical point in the interval $(1/q,1]$. Since the function $g(z)$ is convex  and $\lim_{z\downarrow 0}g(z)=\Bh$, we obtain that $\Psi_1$, in the interval $(1/q,1]$, has exactly two critical points for $B\in(\Bu,\Bh)$ and exactly one critical point for $B\geq \Bh$.
\end{proof}

We are now ready to prove Lemma~\ref{lem:psi1local}.
\begin{proof}[Proof of Lemma~\ref{lem:psi1local}]
Note that $\Psi_1'(\alpha)\uparrow\infty$ for $\alpha\downarrow 0$ and $\Psi_1'(\alpha)\downarrow-\infty$ for $\alpha\uparrow 1$, so all the local maxima correspond to critical points of $\Psi_1$.

By \eqref{sddr}, we have that $\Psi_1'(1/q)=0$ and hence $u=1/q$ is a critical point of $\Psi_1$ for all $B>0$. In fact, we have that $\Psi_1''(1/q)<0$ for $B<\Bh$ and $\Psi_1''(1/q)>0$ for $B>\Bh$. At $B=\Bh$, we have $\Psi_1''(1/q)=0$ and $\Psi_1'''(1/q)\neq 0$ (using $q\geq 3$). It follows that 
\begin{equation}\label{eq:u1qun}
\mbox{$u=1/q$ is a local maximum of $\Psi_1$ iff  $B<\Bh$.}
\end{equation}
We also have that $\Psi_1''$ is monotone in the interval $[0,1/2]$ (since $1/(\alpha(1-\alpha))$ is monotone). For $B<\Bh$, we have that $\Psi_1''(1/q)<0$ and $\lim_{\alpha\downarrow 0}\Psi_1''(\alpha)<0$, so we obtain that $\Psi_1'$ is decreasing in the interval $[0,1/q]$ and hence 
\begin{equation}\label{eq:in01q}
\mbox{for $B<\Bh$, there are no critical points/local maxima of $\Psi_1$ in the interval $[0,1/q)$}.
\end{equation}

We next search for the existence of critical points/local maxima in the interval $(1/q,1]$. We have the following case analysis.

\textbf{Case I.} For $B<\Bu$, by \eqref{eq:u1qun}, \eqref{eq:in01q} and Item~\ref{it:34w1a} of Lemma~\ref{lem:Ncritical}, we have that $u=1/q$ is the unique critical point of $\Psi_1$ and it is a local maximum of $\Psi_1$.

\textbf{Case II.} For $B= \Bu$, by \eqref{eq:in01q} and Item~\ref{it:34w1b} of Lemma~\ref{lem:Ncritical} we have that $\Psi_1$ has exactly two critical points, at $u=1/q$ and $a>1/q$. By \eqref{eq:u1qun}, we have that $\Psi_1$ has a local maximum at $u=1/q$. $\Psi_1$ cannot have a local maximum at $a$,  otherwise $\Psi_1$ must have at least one  critical point in the interval $(u,a)$ which contradicts the fact that $\Psi_1$ has exactly one critical point in $(1/q,1]$ (Item~\ref{it:34w1b} of Lemma~\ref{lem:Ncritical}).

\textbf{Case III.} For $B\in (\Bu,\Bh)$, by \eqref{eq:in01q} and Item~\ref{it:34w1c} of Lemma~\ref{lem:Ncritical}, we have that $\Psi_1$ has exactly three critical points, at $u=1/q$ and $a_1,a_2>1/q$ with $a_1< a_2$. By \eqref{eq:u1qun}, we have that $\Psi_1$ has a local maximum at $u=1/q$. From this, it follows that $\Psi_1$ does not have a local maximum at $a_1$, otherwise $\Psi_1$ would have a critical point in the interval $(u,a_1)$; so, $\Psi_1''(a_1)\geq 0$. In fact, by Lemma~\ref{nononoz}, we can conclude that $\Psi_1$ has a local minimum at $a_1$. It follows that $\Psi_1$ cannot have a local minimum at $a_2$ (otherwise there would be a critical point of $\Psi_1$ between $a_1$ and $a_2$). Again by Lemma~\ref{nononoz}, we conclude that $\Psi_1$ has a local maximum at $a_2$.

The analysis of the values of $B$ where the two local maxima of $\Psi_1$ correspond to global maxima is well-known and can be found in, e.g., \cite{HET}. Roughly, denoting by $a$ the point where $\Psi_1$ has a local maximum in the interval $(1/q,1]$, it can be shown that $\Psi_1(a)-\Psi_1(u)$ is increasing with respect to $B$; then, one only needs to observe that, at $B=\Bp$, it holds that $a=(q-1)/q$ and $\Psi_1(a)=\Psi_1(u)$.

\textbf{Case IV.} For $B\geq\Bh$, by \eqref{eq:in01q} and Item~\ref{it:34w1d} of Lemma~\ref{lem:Ncritical}, we have that $\Psi_1$ has exactly two critical points in the interval $[1/q,1]$, at $u=1/q$ and $a>1/q$. By \eqref{eq:u1qun}, we have that $\Psi_1$ does not have a local maximum at $u=1/q$. Since $\Psi_1'(\alpha)\downarrow-\infty$ for $\alpha\uparrow 1$, we obtain that $\Psi_1$ cannot have a local minimum at $a$ (otherwise there would be a critical point in the interval $(a,1)$). By Lemma~\ref{nononoz}, we conclude that $\Psi_1$ has a local maximum at $a$.
\end{proof}

\subsection{Connection: Proof of Lemma \ref{conn}}
\label{sec:connection}
In this section, we prove Lemma \ref{conn} presented in Section \ref{sec:proof-overview}
connecting the critical points of the function $\Psi_1$
with the fixpoints of the function $F$. Recall that the function $F$ captures the density of the largest color class after one step of the SW algorithm (see \eqref{defa} and \eqref{dexa} for the definition of $F$).  

We first prove the following lemma. The  lemma corresponds to the intuitive fact that $F(z)$ is an increasing function of the initial density $z$ and that the rate of increase, i.e., $F'(z)$, is a decreasing function of $z$. 
\begin{lemma}\label{lem:Fshape}
For every $B>0$, the function $F$ satisfies $F'(z)>0$ and $F''(z)<0$ for all $z\in (1/B,1]$, i.e., $F$ is strictly increasing and concave in the interval $[1/B,1]$.
\end{lemma}
\begin{proof}
We may assume that $B> 1$ (otherwise there is nothing to prove). Let $z\in(1/B,1]$ and recall that $x\in(0,1)$ is the (unique) solution of  
\begin{equation*}\tag{\ref{dexa}}
x + \exp(-zB x) = 1.
\end{equation*}
We view \eqref{dexa} as an equation that defines $x$ as an implicit function of $z$. Differentiating \eqref{dexa} two times we obtain
\begin{align*}
\frac{\partial x}{\partial z}&=\frac{Bx{\mathrm e}^{-z B x}}{1-zB{\mathrm e}^{-z B x}},\\
\frac{\partial^2 x}{\partial z^2}&=-\frac{B^2z{\mathrm e}^{-z B x}\big(2{\mathrm e}^{-z B x}(1-zB{\mathrm e}^{-z B x})+x\big)}{(1-zB{\mathrm e}^{-z B x})^3}.
\end{align*}
Since $F(z)=\frac{1}{q} + \left(1 - \frac{1}{q} \right) z x$, we obtain
\begin{align}
F'(z)&=\frac{q-1}{q}\left(x+z\frac{\partial x}{\partial z}\right)=\frac{(q-1)x}{q(1-zB{\mathrm e}^{-z B x})},\label{eq:firstderivative}\\
F''(z)&=\frac{q-1}{q}\left(2\frac{\partial x}{\partial z}+z\frac{\partial^2 x}{\partial z^2}\right)=-\frac{(q-1)B x {\mathrm e}^{-z B x} \left(z B( x+2{\mathrm e}^{-z B x})-2\right)}{q \left(1-z B{\mathrm e}^{-z B x}\right)^3}. \notag
\end{align}

We first show that $F'(z)>0$ for all $z\in(1/B,1]$. Since $x$ is positive for all $z>1/B$, it suffices to show that $1-zB{\mathrm e}^{-z B x}>0$. Since $x$ satisfies
\eqref{dexa}, we have $zB=-\ln(1-x)/x$, so we have
\begin{equation}\label{eq:xlnarb}
1-zB{\mathrm e}^{-z B x}=\frac{x+(1-x)\ln(1-x)}{x}>0,
\end{equation}
for all $0<x<1$ (the inequality holds since the derivative of the numerator is $-\ln
(1-x)$ and its value at $x=0$ is $0$). Thus $F'(z)>0$ for $z>1/B$.

We next show that $F''(z)<0$ for all $z\in(1/B,1]$. We have already shown that the denominator in the expression for $F''(z)$ is positive, so we only need to show that $z B( x+2{\mathrm e}^{-z B x})-2>0$. Using again that  $zB=-\ln(1-x)/x$, we have
\[z B( x+2{\mathrm e}^{-z B x})-2=-\frac{2x+(2-x)\ln(1-x)}{x}>0,\]
for all $0<x<1$ (the inequality holds since the numerator at $x=0$ is 0 and the first derivative   of the numerator is $-\frac{x+(1-x) \ln (1-x)}{1-x}<0$ from \eqref{eq:xlnarb}). This concludes the proof.
\end{proof}

We next prove the following correspondence. 
\begin{lemma}\label{lem:halfopen}
Let $B>0$. For any $a>1/q$,  $\Psi_1$ has a critical point at $a$ iff $F$ has a fixpoint at $a$.
\end{lemma}
\begin{proof}
Consider first a critical point $a$ of $\Psi_1$ (with $a>1/q$). We use the same parametrization as in the proof of Lemma~\ref{lehehe}, i.e., we set $a=(z+1)/(z+q)$ where $z> 0$, so that $z$ satisfies
\begin{equation*}\tag{\ref{hrz3}}
\ln (1+z) = \frac{zB}{z+q}.
\end{equation*}

Now, consider a fixpoint $a$ of $F$ (with $a>1/q$). Note that $a>1/B$ (this is immediate for $B\geq q$ since then $1/q\geq 1/B$; for $B<q$, we have that $F(z)=1/q$ for all $z\in[1/q,1/B]$, so there is no fixpoint of $F$ with $a\in (1/q,1/B]$). Therefore, from \eqref{defa}, we have  $F(a)=\frac{1}{q} + \left(1 - \frac{1}{q} \right) a x$, where $x\in (0,1]$ is the unique solution of
\begin{equation*}\tag{\ref{dexa}}
x + \exp(-aB x) = 1.
\end{equation*} 
Under the parametrization $a=(z+1)/(z+q)$, equation $F(a) = a$ becomes
\begin{equation}\label{vrz1}
x = \frac{z}{z+1},
\end{equation}
and~\eqref{dexa} becomes
\begin{equation}\label{vrz2}
x + \exp\Big(-\frac{z+1}{z+q} B x\Big) = 1.
\end{equation}
Plugging~\eqref{vrz1} into~\eqref{vrz2} and taking logarithm of both sides yields \eqref{hrz3}. This proves the lemma.
\end{proof}

\begin{lemma}\label{lem:jacobianuniform}
The function $F$ has a fixpoint at $u=1/q$ iff $B\leq \Bh$. For $B<\Bh$, the fixpoint $u=1/q$ of $F$ is  jacobian attractive. For $B=\Bh$, the fixpoint $u=1/q$ is jacobian repulsive.
\end{lemma}
\begin{proof}
Recall from \eqref{defa} that $F(z)=\frac{1}{q} + \left(1 - \frac{1}{q} \right) z x$, where $x=0$ for $z\leq 1/B$ and for $z> 1/B$, $x\in (0,1]$ is the (unique) solution of
\begin{equation*}\tag{\ref{dexa}}
x + \exp(-zB x) = 1.
\end{equation*} 
Note that when $z=u=1/q$, we obtain that $x=0$ for $B\leq \Bh$ and $x>0$ for $B>\Bh$. Hence $F(u)=u$ iff $B\leq \Bh$, proving the first part of the lemma.

For $B<\Bh$, we have that $F$ is constant throughout $[1/q,1/B]$, so trivially $F'(1/q)=0$ and hence $u$ is jacobian attractive. 

For $B=\Bh$, rewrite \eqref{dexa} as 
\begin{equation}\label{eq:epseps}
zq=f(x), \mbox{ where } f(x):=-\frac{\ln(1-x)}{x}.
\end{equation}
Note that as $x\downarrow 0$, we have $z\downarrow 1/q$. Then, for all sufficiently small $x>0$, an expansion of $f$ around $x=0$ yields
\[z=\frac{1}{q}\left(1+\frac{x}{2}+\frac{x^2}{3}\right)+O(x^3).\] 
It is not hard from here to conclude 
\[x=2q (z-1/q)+O((z-1/q)^2),\]
for all $z$ in a small neighborhood of $1/q$. It follows that 
\[F'(1/q)=2(q-1)/q>1,\]
for all $q\geq 3$, and hence $u$ is jacobian repulsive.
\end{proof}

We are now ready to give the proof of Lemma \ref{conn}.
\begin{proof}[Proof of Lemma \ref{conn}]
Our goal is to show that, in the interval $[1/q,1]$, the hessian local maxima of $\Psi_1$ and the jacobian attractive fixpoints of $F$ are in one-to-one correspondence.

We first prove the correspondence in the half-open interval $(1/q,1]$. By Lemma~\ref{lem:halfopen}, we have that every critical point $a>1/q$ of $\Psi_1$ is also a fixpoint of $F$ (and vice versa). Therefore, it suffices to show that a critical point $a$ of $\Psi_1$ is a hessian local maximum of $\Psi_1$ iff $a$ is also a jacobian attractive fixpoint of $F$. 

From \eqref{sddr}, we have 
\begin{equation*}\tag{\ref{sddr}}
\Psi_1''(a) = B \frac{q}{q-1} - \frac{1}{a(1-a)}.
\end{equation*}
By Lemma~\ref{lehehe}, we have that $a>1/B$, so from \eqref{eq:firstderivative}, we have that 
\begin{equation}\label{zzzz}
F'(a) = \Big(1-\frac{1}{q}\Big) \frac{x}{1-aB\exp(-aBx)},
\end{equation}
where $x\in (0,1]$ satisfies
\begin{equation*}\tag{\ref{dexa}}
x + \exp(-aB x) = 1.
\end{equation*}
Since $a$ is also a fixpoint of $F$, we have $F(a)=a$, which yields $x=(qa-1)/((q-1)a)$. From \eqref{dexa}, we also have $\exp(-aBx)=1-x$. Plugging these values in \eqref{zzzz}, we obtain
\begin{equation}\label{zzzz2}
F'(a) = 1 + \frac{B \frac{q}{q-1} - \frac{1}{a(1-a)}}{\frac{q}{1-a} - B\frac{q}{q-1}} = 1 +
\frac{\Psi_1''(a)}{\frac{q}{1-a} - B\frac{q}{q-1}}.
\end{equation}
The denominator of~\eqref{zzzz2} is positive (since $a>1/B$) and hence we have
\begin{equation}\label{kk1w}
F'(a) < 1\ \Longleftrightarrow\ \Psi_1''(a)<0.
\end{equation}
We also have by Lemma~\ref{lem:Fshape} that $F'(a)>0$, so we can rewrite~\eqref{kk1w} as
\begin{equation}\label{kk2w}
\left| F'(a)\right| < 1\ \Leftrightarrow\ \Psi_1''(a)<0.
\end{equation}
This establishes the lemma in the interval $(1/q,1]$.

We next consider the left-endpoint of the interval $[1/q,1]$, i.e., the point $u=1/q$. From \eqref{sddr}, we have that $u$ is a critical point of $\Psi_1$ for all $B>0$ and it is a hessian local maximum of $\Psi_1$ (i.e., it holds that $\Psi_1''(u)<0$) iff $B<\Bh$. By Lemma~\ref{lem:jacobianuniform}, we have that $u=1/q$ is a jacobian attractive fixpoint of $F$ precisely when $B<\Bh$. 

This concludes the proof of the lemma.
\end{proof}

\subsection{Analysis of the fixpoints of $F$: Proof of Lemma \ref{lem:fixpoints}}
\label{sec:fixpoints}

In this section, we prove Lemma \ref{lem:fixpoints}, i.e., we analyze the fixpoints of the function $F$ in the interval $[1/q,1]$.  Lemma~\ref{lem:jacobianuniform} details when $u=1/q$ is a (jacobian attractive) fixpoint of $F$, therefore we will focus on the interval $(1/q,1]$.

Recall, by Lemma~\ref{lem:halfopen}, a point $a\in(1/q,1]$ is a fixpoint of $F$ iff $a$ is a critical point of $\Psi_1$. Therefore, the number of fixpoints of $F$ in the interval $(1/q,1]$ is the same as the number of critical points of $\Psi_1$ in the interval $(1/q,1]$. Lemma~\ref{lem:Ncritical} therefore yields the following corollary.

\begin{corollary}
\label{lem:exist}
For $B<\Bu$, there is no fixpoint of $F$ in the interval $(1/q,1]$. For $B=\Bu$, there is a unique fixpoint of $F$ in the interval $(1/q,1]$. For $B\in(\Bu,\Bh)$, there are two fixpoints of $F$ in the interval $(1/q,1]$. For $B\geq \Bh$, there is a unique fixpoint of $F$ in the interval $(1/q,1]$.
\end{corollary}

By Lemma~\ref{nononoz}, every local maximum of $\Psi_1$ in the interval $(1/q,1]$ is in fact a hessian maximum of $\Psi_1$. By Lemma~\ref{conn}, a hessian maximum of $\Psi_1$ is also a jacobian attractive fixpoint of $F$. Therefore, Lemma~\ref{lem:psi1local}, which details the local maxima of $\Psi_1$, yields the following.

\begin{corollary}\label{corp4}
For $B>\Bu$, the function $F$ has a unique  jacobian attractive fixpoint in the interval $(1/q,1]$, namely the point $a>1/q$ where $\Psi_1$ has a local maximum. 
\end{corollary}

Corollaries~\ref{lem:exist} and~\ref{corp4} classify the number of fixpoints of $F$ and when these are jacobian attractive for all $B\neq \Bu$. The following lemma addresses the case $B=\Bu$. 
\begin{lemma}\label{lem:notJacobianuniqueness}
For $B=\Bu$, consider the fixpoint $a$ of $F$ in the interval $(1/q,1]$. Then, $F'(a)=1$.
\end{lemma}
\begin{remark}\label{rem:notattractive}
Note,  the non-attractiveness of the fixpoint $a$  for $B=\Bu$ follows from $F'(a)=1$ and $F''(a)\neq 0$ (Lemma~\ref{lem:Fshape}).
\end{remark}
\begin{proof}
From \eqref{zzzz2}, it suffices to show that $\Psi_1''(a)=0$.

Recall also that $\Psi_1'(a)=0$, i.e., $a$ is a critical point of $\Psi_1$. Using Lemma~\ref{lem:Ncritical},  we therefore have that the critical points of $\Psi_1$ in the interval $[1/q,1]$ are precisely $u=1/q$ and $a$. 

By Lemma~\ref{lem:psi1local}, $u=1/q$ is the unique local maximum of $\Psi_1$ and hence it must be the case that $\Psi_1''(a)\geq 0$ (otherwise $a$ would also be a local maximum).  We also have that $\Psi_1''(a)\leq 0$: otherwise, $\Psi_1$ has a local minimum at $a$. Since $\Psi_1'(\alpha)\downarrow -\infty$ as $\alpha\uparrow 1$, we would then obtain that $\Psi_1$ has a critical point in the interval $(a,1)$, contradicting that, for $B=\Bu$, $\Psi_1$ has a unique critical point in the interval $(1/q,1]$ (Lemma~\ref{lem:Ncritical}).

Thus, $\Psi_1''(a)=0$, as wanted.
\end{proof}
We are now ready to prove Lemma~\ref{lem:fixpoints} from Section~\ref{sec:proof-overview}.
\begin{proof}[Proof of Lemma~\ref{lem:fixpoints}]
Lemma~\ref{lem:jacobianuniform} details when $u=1/q$ is a fixpoint of $F$ and when it is jacobian attractive. It therefore remains to classify the fixpoints in the interval $(1/q,1]$.

For $B<\Bu$, there are no fixpoints of $F$ in the interval $(1/q,1]$ by Corollary~\ref{lem:exist}. 

For $B>\Bu$, by  Corollary~\ref{corp4}, there is precisely one jacobian attractive fixpoint of $F$ in the interval $(1/q,1]$ and it coincides with the point where $\Psi_1$ has a local maximum.

For $B=\Bu$, by  Corollary~\ref{corp4}, there is precisely one  fixpoint $a$ of $F$ in the interval $(1/q,1]$. By Lemma~\ref{lem:notJacobianuniqueness} and Lemma~\ref{lem:Fshape}, we have that $F'(a)=1$ and $F''(a)\neq 0$, so $a$ is repulsive but not jacobian repulsive.

This completes the proof of the lemma.
\end{proof}

\section{Random Graph Lemmas}\label{sec:randomgraph}

In this section, we collect relevant results from the literature for the sizes of the components in $G(n,p)$ where $p\sim 1/n$. We will use these to analyze one step of the SW algorithm. 

For a graph $G$, we denote by $C_1,C_2,\hdots$ the connected components of $G$ in decreasing order of size; throughout the paper we refer to the size of a component $C$ as the number of vertices in it and use $|C|$ to denote its size. Roughly, in one iteration of the SW algorithm, the size of the largest component after the percolation step controls the size of the largest color class, and the fluctuations are determined by the sum of squares of the sizes of the components.

\subsection{The supercritical regime}
We will need several known results on the $G(n,p)$ model in the
supercritical regime ($p=c/n$, where $c>1$). The size of the giant
component is asymptotically normal~\cite{Stepanov} and satisfies moderate deviation inequalities around its mean value~\cite{AL}. We will use
the following moderate deviation inequalities for the sizes of the largest and second largest components of $G$.
These are used to track the evolution of the SW dynamics for an exponential number of steps in the slow mixing regime $\Bu<B<\Bh$.

\begin{lemma}\label{l1al1}
Let $G\sim G(n,c/n)$ where $c>1$. Let $\beta\in (0,1)$ be the
solution of $x + \exp(-c x) = 1$. Let $C_1, C_2$ be the largest and second largest components of $G$ respectively. Then, for every constant $\epsilon\in (0,1/3]$ it holds that
\begin{gather}
P(\big||C_1| - \beta n\big| \geq n^{1/2+\epsilon})\leq \exp(-\Theta(n^{2\epsilon})),\label{st1aa}\\
P(|C_2| \geq n^{\epsilon})\leq \exp(-\Theta(n^{\epsilon})). \label{st2bb}
\end{gather}
\end{lemma}
\begin{proof}
Equation~\eqref{st1aa} is proved in \cite[Lemma 5.4]{LNNP}. 
We next prove equation~\eqref{st2bb}. All the elements are contained in the proof of~\cite[Theorem 5.4]{JLR}. The probability that there exists a
component of size from the interval $\{n^{\epsilon},\dots,n^{2/3}\}$ is bounded
 by (see \cite[p. 110, line 11]{JLR}):
\begin{equation}\label{ez1}
n^2 \exp (-( (c-1)^2/(9c) ) n^{\epsilon}).
\end{equation}
The probability that there exist two or more components of size at
least $n^{2/3}$ is bounded by (see \cite[p. 110, line 24]{JLR}):
\begin{equation}\label{ez2}
n^2 \exp (-( (c-1)^2c/4 ) n^{1/3}).
\end{equation}
Using the union bound (combining~\eqref{ez1} and~\eqref{ez2}) we obtain~\eqref{st2bb}, that is,
with high probability we have only one component of size $\geq n^{\epsilon}$.
  \end{proof}

The following lemma will be used to analyze the evolution of the SW chain when $B=\Bu$.
\begin{lemma}\label{lem:supercritical}
Let $G\sim G(n,c/n)$ where $c_0<c<c_1$ for absolute constants $c_0,c_1>1$ ($c$ may otherwise depend on $n$). Let $\beta\in(0,1)$ be the unique solution of $\beta+\exp(-\beta c)=1$. Denote by $C_1$ the largest component in $G$.

Then, for every constant $\epsilon>0$, for all sufficiently large $n$ it holds that
\begin{equation}\label{eq:nlittleo1}
n\beta-n^\epsilon\leq E[|C_1|]\leq n \beta+n^\epsilon.
\end{equation}
Moreover, there exist constants $K_1,K_2, K_3>0$ (depending only on $c_0,c_1$) such that for all sufficiently large $n$ it holds that 
\begin{gather}
K_1 n\leq Var[|C_1|]\leq K_2 n, \quad E\Big[\sum_{j\geq 2} |C_j|^2\Big] \leq K_3 n\label{eq:variancesuper}.
\end{gather}
Finally, there exists a constant $U>0$ (depending only on $c_0,c_1$) such that for all sufficiently large $n$, for all $u\geq U$, it holds that
\begin{equation}\label{eq:moddev}
P\big(\big||C_1|-n\beta\big|\geq u\sqrt{n}\big)\leq U/u^2.
\end{equation}
\end{lemma}
\begin{proof}
The bounds on $E[|C_1|]$ and $Var[|C_1|]$ can be found in \cite[Theorem 5]{Oghlan}. The bound on $E\Big[\sum_{j\geq 2} |C_j|^2\Big]$ is an immediate corollary of \cite[Corollary 5.6]{LNNP}. The  probability bound in \eqref{eq:moddev} can be derived by Chebyshev's inequality using the bounds on $E[|C_1|]$ and $Var[|C_1|]$.
\end{proof}

\subsection{The scaling window \& subcritical regimes}

We use the following well-known result about the size of the giant component in the subcritical regime.
 
\begin{lemma}[see, e.g.,~\cite{JLR}, p.109]\label{l1}
Let $t\in (0,1]$ be a constant. Let $G\sim G(n,c/n)$ where $c<1$ is a constant, and $C_1$ be the largest component of $G$. Then,
\begin{equation*}
P(|C_1|\geq n^{t})\leq \exp(- \Theta(n^{t})).
\end{equation*}
\end{lemma}

The following lemma considers the size of the components in the scaling window.

\begin{lemma}\label{lemia1}
There exist constants $K,c,c'>0$ such that for any $n$ and
\begin{enumerate}
\item \label{it:sub} any $\eps\in (0,1)$ for random $G$ from $G(n,(1-\eps)/n)$ we have
$$
E\Big[\sum_{i\geq 1} |C_i|^2\Big]\leq \frac{K n}{\eps},
$$
\item \label{it:sup} for any $\eps\in [1/n^{1/3},c]$  for random $G$ from $G(n,(1+\eps)/n)$ we have
$$
E\Big[\sum_{i\geq 2} |C_i|^2\Big]\leq \frac{K n}{\eps},
$$
\item \label{it:suplargest} for any $\eps\in [c'/n^{1/3},c']$ for random $G$ from $G(n,(1+\eps)/n)$ we have
\begin{equation*}
P\Big(\big[|C_1|< (7/4)\eps n\big] \cup \big[|C_1|> 3\eps n\big] \Big)\leq K \exp(-c\eps^3 n).
\end{equation*}
\end{enumerate}
\end{lemma}
\begin{proof}
Part~\ref{it:sub} follows from \cite[Lemma 5.3 \& Theorem 5.12]{LNNP}. Part~\ref{it:sup} is \cite[Theorem 5.13, Part (ii)]{LNNP}. Part~\ref{it:suplargest} follows from \cite[Lemma 5.4 \& Theorem 5.9]{LNNP}.
\end{proof}

\begin{lemma}\label{lemia3}
Let $G\sim G(n,p)$, $p\geq (1-An^{-1/3})/n$ where $A$ is a  large constant. Let $C_1,C_2,\hdots$ be the connected components of $G$ in decreasing order of size. Then, for all sufficiently large  constant $L>0$,  there exists  a positive constant $p'$ such that
for all  $n$ sufficiently large it holds that $P\big(|C_1|\geq Ln^{2/3},\ \sum_{j\geq 2} |C_j|^2 \leq n^{4/3}\big)\geq p'$.
\end{lemma}

The proof of Lemma~\ref{lemia3} is based on \cite[Proof of Lemma 8.26]{LNNP}. We will use the following special case of \cite[Theorem 5.20]{JLR}.
\begin{corollary}[{\cite[Theorem 5.20]{JLR}}]\label{strictercritical}
Let $t$ be a positive integer and $d,a_1,\hdots,a_t,$ $b_1,\hdots,b_t$ be such that $\infty\geq a_1> b_1> a_2> b_2> \hdots> a_t> b_t>d>0$. Let $c$ be a constant (not necessarily positive) and let $p=(1+cn^{-1/3})/n$.

For $G\sim G(n,p)$ denote by $C_1,C_2,\hdots$ the connected components of $G$ in decreasing order of their sizes. There exists $\ell:=\ell(c,t,d,a_1,\hdots,a_t,b_1,\hdots,b_t)>0$ such that for all sufficiently large $n$, it holds that
\begin{equation*}
P\Big(a_1 \geq \frac{|C_1|}{n^{2/3}}\geq b_1 ,\hdots,a_t \geq \frac{|C_t|}{n^{2/3}}\geq b_t , d\geq \frac{|C_{t+1}|}{n^{2/3}}\Big)\geq \ell.
\end{equation*}
\end{corollary}
\begin{proof}
The statement of \cite[Theorem 5.20]{JLR}  is for the  {E}rd\"os-{R}\'enyi random graph model $G(n,M)$ with $M=(n/2)+cn^{2/3}$. Since for $G\sim G(n,p)$ with $p=(1+2cn^{-1/3})/n$ the number of edges is $(n/2)+cn^{2/3}+O(\sqrt{n})$ with  probability $\Omega(1)$, the corollary follows.
    \end{proof}

\begin{proof}[Proof of Lemma~\ref{lemia3}]
Let $A>0$ be a large constant. We consider two cases. If $p\geq (1+An^{-1/3})/n$, we have from Part~\ref{it:sup} of Lemma~\ref{lemia1} that
$E[\sum_{j\geq 2} |C_j|^2]\leq K n^{4/3}/A$, so by Markov's inequality
\[P\Big(\sum_{j\geq 2} |C_j|^2 \leq n^{4/3}\Big)\geq 1-\frac{K}{A}.\]
From Corollary~\ref{strictercritical} (with $t=1$, $b_1=L$) we obtain that for $p=1/n$, $|C_1|$ is greater than $Ln^{2/3}$ with asymptotically positive probability $p_1$ for any constant $L>0$. Note that for $p>1/n$ we can couple $G\sim G(n,1/n)$ and $G'\sim G(n,p)$
so that $G$ is a subgraph of $G'$. Since $|C_1|$ is monotone, it follows that for $p>1/n$, $|C_1|$ is greater than $Ln^{2/3}$ with  positive probability $p_1$. Provided that $A$ is sufficiently large (depending on $K,p_1$), by a union bound we have that
\[P\Big(|C_1|\geq Ln^{2/3},\ \sum_{j\geq 2} |C_j|^2 \leq n^{4/3}\Big)\geq p_1-\frac{K}{A}>0.\]
If $(1-An^{-1/3})/n\leq p\leq (1+An^{-1/3})/n$, we have from Corollary~\ref{strictercritical} (with $t=1$, $d=1$, $b_1=L$) and the argument in \cite[Proof of Lemma 8.26]{LNNP} that for all sufficiently large $L$, it holds that
\[P\Big(|C_1|\geq Ln^{2/3},\ \sum_{j\geq 2} |C_j|^2 \leq n^{4/3}\Big)\geq p_2>0,\]
where $p_2$ is a constant. The lemma follows.
    \end{proof}

We will also use the following upper bound on the size of the giant component in the critical window.
\begin{lemma}[{\cite[Corollary 5.6]{SS}, see also \cite[Theorems 1 \& 7]{NP}}]\label{lem:criticalper}
Let $G\sim G(n,p)$ with $p=(1\pm cn^{-1/3})/n $ where $c$ is a sufficiently large constant. Let $C_1$ be the largest component in $G$. Then there exists a constant $r>0$ such that for positive $A$ larger than an absolute constant, it holds that
\[P(|C_1|>An^{2/3})\leq \exp(-rA^3).\]
\end{lemma}

\subsection{Concentration Inequalities}
We conclude this section by recording the following version of Azuma's inequality that we will use. 
\begin{lemma}[Azuma's inequality, see, e.g.,~{\cite[p.37]{JLR}}]\label{lem:azuma}
Let $X_1,\hdots, X_n$ be independent random variables such that, for $i=1,\hdots,n$ it holds that $a_i\leq X_i\leq b_i$. Let $X=X_1+\cdots+X_n$. Then, for all $t\geq 0$, it holds that
\[\Pr\big(|X-E(X)|\geq t\big)\leq 2\exp\Big(-\frac{t^2}{\sum^n_{i=1}(b_i-a_i)^2}\Big).\]
\end{lemma}

\section{Slow Mixing for $\Bu<B<\Bh$}\label{sepp}
In this section, we show that the SW chain mixes slowly when $B\in (\Bu,\Bh)$.

Let ${\cal B}(\vv,\delta)$ be the $\ell_\infty$-ball of configuration vectors of the $q$-state Potts model in $K_n$
around  $\vv$ of radius $\delta$, that is,
\begin{equation}\label{eq:ballball}
{\cal B}(\vv,\delta) = \{ \w\in{\mathbb Z}^q\,|\, \|\w/n - \vv\|_\infty\leq \delta\}.
\end{equation}

We will show that for $B<\Bh$ the Swendsen-Wang algorithm is
extremely unlikely to leave the vicinity of the uniform configuration. More precisely, we show the following.
\begin{lemma}\label{l2}
Assume $B<\Bh$. There exists $\eps_0>0$ such that, for all constant $\eps\in (0,\eps_0)$, for $S = {\cal B}(\u,\eps)$, it holds that
\begin{equation*}
\Pr (X_{1}\in S\mid X_0\in S) \geq 1 - \exp(-\Theta(n^{1/2})).
\end{equation*}
\end{lemma}

The reason for Lemma~\ref{l2} failing for $B>\Bh$ is that the percolation step of the Swendsen-Wang
algorithm on a cluster of size $n/q$ yields linear sized connected components, and these allow
the algorithm to escape the neighborhood of $\u$ (a somewhat similar argument applies for $B=\Bh$ as well, though in this case one has to account more carefully for the fluctuations of the largest components since the percolation step of the SW dynamics is in the critical window for such configurations).

\begin{proof}[Proof of Lemma \ref{l2}]
Let $X_0\in S$. The first step of the Swendsen-Wang algorithm chooses, for each color class,
a random graph from $G(m,p)$, where $p=B/n$ and $m$ is the number of vertices of that color.
For all sufficiently small $\eps$ we have
$$
p = \frac{B}{m}\frac{m}{n} \leq \frac{d}{m},
$$
where $d<1$ (we used $B<q$ and $m\leq n/q+\eps n$). Now Lemma~\ref{l1} (with $t=1/2$) implies that with
probability at least
\begin{equation}\label{zle1}
1 - n \exp(-\Theta(n^{1/2}))
\end{equation}
all components after the first step have size $\leq n^{1/2}$. The second step of the Swendsen-Wang algorithm colors each component
by a uniformly random color; call the resulting state $X_1$. Let $Z_i$ be the number of
vertices of color $i$ in $X_1$. By symmetry, $E[Z_i] = n/q$.

Now assume that all components have size $\leq n^{1/2}$. By Azuma's inequality (see Lemma~\ref{lem:azuma}),
\begin{equation}\label{zle2}
\Pr( |Z_i-n/q|\geq\eps n ) \leq \exp(-\Theta(n^{1/2})),
\end{equation}
and hence $\Pr(X_1\in S)\geq 1 - n \exp(-\Theta(n^{1/2}))$, which combined with~\eqref{zle1} yields the lemma.
    \end{proof}

We also analyze the behavior of the algorithm around the majority
configuration $\m$ (recall, for the configuration to exist we need $B\geq\Bu$).
\begin{lemma}\label{l3}
Assume $B>\Bu$ and let $\m=(a,b,\dots,b)$ where $a>1/q$ is the jacobian attractive fixpoint of $F$ of Lemma~\ref{lem:fixpoints}.
There exists constant $\eps_0>0$ such that, for all sufficiently large $n$, for all $\eps\in (n^{-1/7},\eps_0)$, for $S = {\cal B}(\m,\eps)$, we have
\begin{equation}\label{eq:PSWSS}
\Pr (X_1\in S\mid X_0\in S) \geq 1 - \exp(-\Theta(n^{1/3})).
\end{equation}
\end{lemma}
The reason that Lemma~\ref{l3} does not hold for $B=\Bu$ is that the fixpoint $a>1/q$ of $F$ is no longer attractive; indeed, in Section~\ref{sec:fast-Bu} we show that the Swendsen-Wang algorithm escapes the vicinity of this fixpoint in $O(n^{1/3})$ steps.

\begin{proof}[Proof of Lemma \ref{l3}]
Let $X_0\in S$ and let $\gamma:=F'(a)$ (recall that $|\gamma|<1$, since $a$ is Jacobian attractive fixpoint by Lemma~\ref{lem:fixpoints}).
The first step of the Swendsen-Wang algorithm chooses,
for each color class, a random graph from $G(m,p)$, where $p=B/n$ and $m$ is the number of vertices of that color.
Let $m_1$ be the number of vertices of the dominant color. Since $X_0\in S$ we have $m_1/n = a + \tau =: a'$ where
$|\tau|\leq \eps$. We can write
$$
p = (m_1 B/n) / m_1 = (a'B)/m_1,
$$
where $a'B>1$ for sufficiently small $\eps_0>0$ (using $aB>1$ from Lemmas~\ref{lehehe} and~\ref{lem:halfopen}). This means that
the $G(m,p)$ process for the dominant color class is supercritical. Let $\beta\in (0,1]$ be the root of $x+\exp(-a'B x) = 1$.
By Lemma~\ref{l1al1}
the random graph will have, with probability
$\geq 1-\exp(-\Theta(n^{1/3}))$, one component of size $a' \beta n \pm n^{2/3}$ and all the other components will
have size at most $n^{1/3}$.

Let $m_2$ be the number of vertices in one of the non-dominant colors. Since $X_0\in S$ we have $m_2/n =: b'$ where
\begin{equation}
b-\eps_0 \leq b-\eps \leq b' \leq b+\eps\leq b+\eps_0.
\end{equation}
We can write
$$
p = (m_2 B/n) / m_2 = (b' B)/m_2,
$$
where $b' B<1$ for sufficiently small $\eps_0>0$ (using $bB<1$ from Lemmas~\ref{lehehe} and~\ref{lem:halfopen}). This means that
the $G(m,p)$ process in this component is subcritical.  By Lemma~\ref{l1} (with $t=1/3$), with probability
$\geq 1-\exp(-\Theta(n^{1/3}))$ the random graph will have all components of size at most $n^{1/3}$.

To summarize: starting from a configuration in $S$ after the first step of the Swendsen-Wang algorithm
we have, with probability $\geq 1-q\exp(-\Theta(n^{1/3}))$ one large component of size
$a' \beta n \pm n^{2/3}$ and the remaining components are of size $\leq n^{1/3}$ (small components).
In the second step of the algorithm the components get colored by a random color.
By symmetry,  in expectation each color obtains $(n-a' \beta n \mp n^{2/3})/q$ vertices from the small components
and by Azuma's inequality this number is $(n-a' \beta n \mp n^{2/3})/q \pm n^{5/6}$ with probability
$\geq 1 - \exp(-\Theta(n^{1/3}))$.
Combining the analysis of the first and the second step we obtain that at the end with
probability $\geq 1-2q\exp(-\Theta(n^{1/3}))$ we have 
\begin{equation}\label{grg4c}
\Big\|\alphab(X_{t+1}) -
\Big(F(a'),\frac{1-F(a')}{q-1},\dots,\frac{1-F(a')}{q-1}\Big)\Big\|_\infty \leq 2n^{-1/6}.
\end{equation}

For sufficiently small $\eps_0>0$ there exists $\gamma'\in (\gamma,1)$ such that for all $|\tau|<\eps_0$ we
have $|F(a+\tau)-a|<\gamma'\tau\leq\gamma'\epsilon$. Therefore, for all  sufficiently large $n$ and $\epsilon\in (n^{-1/7},\epsilon_0)$, we have
\begin{equation}\label{eq:newstateb}
|F(a')\pm 2n^{-1/6} - a|\leq \eps \mbox{ and }\Big|\frac{1-F(a')}{q-1}\pm 2n^{-1/6} - b\Big|\leq\eps.
\end{equation} 
Combining \eqref{grg4c} and \eqref{eq:newstateb} gives that $X_1\in S$ with probability at least $1-\exp(-\Omega(n^{1/3}))$, which  finishes the proof of the  lemma. 
\end{proof}

Combining Lemmas~\ref{l2} and \ref{l3} we obtain Part \ref{thm:slow} of Theorem \ref{thm:main}.
\begin{corollary}
For $B\in (\Bu,\Bh)$ the Swendsen-Wang algorithm has mixing time $\exp(\Omega(n^{1/3}))$.
\end{corollary}
\begin{proof}
For some small constant $\epsilon>0$, let $S_\u=\mathcal{B}(\u,\epsilon)$ and  $S_\m=\mathcal{B}(\m,\epsilon)$.  We can choose $\epsilon$ so that $S_\u\cap S_\m=\emptyset$ (since $\u\neq \m$) and further, by Lemmas~\ref{l2} and~\ref{l3},
\begin{equation}\label{eq:pswum}
\Pr(X_1\in S_\u\mid X_0\in S_\u)\geq 1-\exp(-C n^{1/3}), \quad \Pr(X_1\in S_\m\mid X_0\in S_\m)\geq 1-\exp(-C n^{1/3}),
\end{equation}
where $C>0$ is a constant independent of $n$.

 Let $\mu$ be the stationary distribution of the SW chain, i.e., $\mu$ is the Potts distribution given in \eqref{eq:Gibbs}. Let $S=S_\u$ if  $\mu(S_\u)\leq \mu(S_\m)$ and $S=S_\m$ otherwise, so that $\mu(S)\leq 1/2$. We will use $\overline{S}$ to denote the set of configurations which are not in $S$.

Let $X_0\in S$ and $T=\frac{1}{10}\exp( Cn^{1/3})$.  Then, using \eqref{eq:pswum}, we have that 
\[\Pr(X_T\in S)\geq (1-\exp(-C n^{1/3}))^{T}\geq 1-T\exp(-C n^{1/3}))\geq 9/10.\]
Observe now that 
\[d_{TV}(X_T,\mu)=\max_{A\subseteq \Omega} |\mu(A)-\Pr(X_T\in A)|\geq |\mu(\overline{S})-\Pr(X_T\in \overline{S})|\geq \frac{1}{2}-\frac{1}{10}>\frac{1}{4}.\]
It follows from the definition of mixing time that $\Tmix\geq T$, as claimed.
\end{proof}
We remark that for $B\in (\Bu,\Bh)$ and $B\neq \Bp$, the subset of initial configurations where the mixing of Swendsen-Wang is slow has exponentially small mass in the Gibbs distribution (known as \emph{essential mixing}, see \cite{CDLLPS}). More precisely, for $B\neq \Bp$, the  Swendsen-Wang algorithm started from a typical configuration of the Gibbs distribution gets within total variation distance $1/poly(n)$ from the stationary distribution in $O(\log n)$ steps. For $B\in (\Bu,\Bp)$,  this follows by considering starting configurations which are close to uniform and then using the upcoming Lemmas~\ref{lem:vertices-agree},~\ref{lem:phases-uniform} and~\ref{zirafa1}; for $B\in (\Bp,\Bh)$, this follows by considering starting configurations which are close to a majority phase and then  using Lemmas~\ref{lem:vertices-agree},~\ref{lem:phases-agree} and~\ref{slon8}.

\section{Basic rapid mixing results}

Recall from Section~\ref{sec:crphasetr} the definition of a phase of a configuration. In this section, we will consider two copies of the SW chain and, utilizing the symmetry of the complete graph, we give sufficient conditions on their phases that ensure the existence of a coupling.

The first lemma asserts that once the phases of the two chains align, we can couple the chains (so that the configurations agree). 
\begin{lemma}[{\cite[Lemma 4]{CDFR}}]\label{lem:vertices-agree}
For any constant $B>0$, for all $q\geq 2$, all constant $\eps>0$, for $T=O(\log{n})$ there is a coupling where 
$\Pr(X_T \neq Y_T\mid \alphab(X_0) = \alphab(Y_0)) \leq \eps.$
\end{lemma}
Lemma~\ref{lem:vertices-agree} is essentially identical to \cite[Lemma 4]{CDFR}, which is also used in \cite[Lemma 4.1]{LNNP}. For completeness, we include the proof of the lemma.

\begin{proof}[Proof of Lemma \ref{lem:vertices-agree}]
Let $A_t = \{v: X_t(v) = Y_t(v)\}$ and $D_t = V\setminus A_t$.
We will define a one-step coupling which maintains $\alphab(X_t)=\alphab(Y_t)$
and where
\begin{equation}
\label{eq:contract-couple}
E[|D_{t+1}|\mid X_t,Y_t] = (1-1/q)|D_t|.
\end{equation}
We'll define a matching $\tau:V \ra V$.
For $v\in A_t$ let $\tau(v)=v$.
For $V\setminus A_t$ define $\tau$ so that
for all $v\in V$, $X_t(v)=Y_t(\tau(v))$.
In words, $\tau$ matches vertices with the same color
(this is always possible since $\alphab(X_t)=\alphab(Y_t)$)
and it uses the identity matching on those vertices whose colors
agree in the two chains.
In the percolation step of the Swendsen-Wang process, first
perform the step for chain $X_t$, then for $Y_t$ for a pair $v,w$
where $Y_t(v)=Y_t(w)$ we delete the edge iff the edge $(\tau(v),\tau(w))$
is deleted.  Therefore, the component sizes are identical for the two chains
and we can couple the recoloring in the same manner so that
if $v\in A_t$ then $v\in A_{t+1}$ and \eqref{eq:contract-couple}
holds.
Then, by applying Markov's inequality,
\[ \Pr(X_t\neq Y_t\mid X_0,Y_0) \leq n(1-1/q)^t \leq \eps
\] for $t=O(\log{n})$.   
\end{proof}

It is enough to get the phases within $O(\sqrt{n})$ distance
and then there is a coupling so that with constant probability the phases will be
identical after one additional step. More precisely, we have the following lemmas which are analogous to \cite[Theorem 6.5]{LNNP} for the $q=2$ case.

\begin{lemma}\label{lem:phases-uniform}
Let $B< \Bh$ and $\u=(1/q,\hdots,1/q)$. Let $X_0,Y_0$ be a pair of configurations where
$\|\alphab(X_0) - \u\|_\infty\leq Ln^{-1/2},
\|\alphab(Y_0) - \u\|_\infty\leq Ln^{-1/2},$
for an arbitrarily large constant $L>0$. For all sufficiently large $n$, there exists a coupling such that with prob. $\Theta(1)$, $\alphab(X_1) = \alphab(Y_1)$.
\end{lemma}
\begin{lemma}\label{lem:phases-agree}
Let $B\geq \Bu$ and $\m=(a,b,\dots,b)$ where $a>1/q$ is the attractive fixpoint of $F$ of Lemma~\ref{lem:fixpoints}. Let $X_0,Y_0$ be a pair of configurations where
$\|\alphab(X_0) - \m\|_\infty\leq Ln^{-1/2},
\|\alphab(Y_0) - \m\|_\infty\leq Ln^{-1/2},$
for an arbitrarily large constant $L>0$. For all sufficiently large $n$, there exists a coupling such that with prob. $\Theta(1)$, $\alphab(X_1) = \alphab(Y_1)$.
\end{lemma}

For completeness, we include the proof of the lemmas.

\begin{proof}[Proof of Lemmas~\ref{lem:phases-uniform} and~\ref{lem:phases-agree}]
We focus on proving Lemma~\ref{lem:phases-agree} which is (slightly) more involved than Lemma~\ref{lem:phases-uniform}, and then explain the small modification needed to obtain Lemma~\ref{lem:phases-uniform}. Our proof closely follows the approach in \cite[Theorem 6.5]{LNNP} (which is for the case $q=2$) with small differences in some of the technical details.

Perform the percolation step of the Swendsen-Wang algorithm independently for the chains $X_0$ and $Y_0$.
By Lemma 5.7 in \cite{LNNP}, there is a constant $C>0$ such that with probability $1-O(1/n)$, there are $\geq Cn$ isolated vertices in each chain (i.e., components of size 1).  Our goal will be to couple the colorings of the components using the $C n$ isolated vertices to guarantee that $\alphab(X_1) = \alphab(Y_1)$. 

In each chain, order the components by decreasing size.  Next, couple the coloring step so that the largest component in each chain gets the same color. For the remaining components, color them independently in each chain in order of decreasing size, but leave the last $Cn$ components uncolored.  As noted earlier,
the remaining $Cn$ uncolored components in each chain are isolated vertices (with probability $1-O(1/n)$). Let $\hat{X_1},\hat{Y_1}$ denote the configuration except on these $Cn$ uncolored components and denote by $x_i,y_i$ the number of vertices  which are assigned color $i$  under $\hat{X_1},\hat{Y_1}$ respectively. 

We will show that under this coupling, for a (large) constant $L'>0$, with probability $\Theta(1)$, it holds that 
\begin{equation}\label{eq:truncated}
|x_i-y_i|\leq L'\sqrt{n}\mbox{ for all } i=1,\hdots,q.
\end{equation}
We will do this shortly, let us assume \eqref{eq:truncated} for the moment and conclude the coupling argument. For $i\in[q]$, let $\ell_i:=x_i-y_i$, so that $|\ell_i|\leq L'\sqrt{n}$ and denote by $\boldsymbol{\ell}$ the vector with coordinates $\ell_1,\hdots,\ell_q$.  Further, denote the remaining $Cn$ uncolored vertices  as $v_1,\hdots,v_{Cn}$. Let $Z_i$ be the r.v. which denotes the number of vertices from $v_1,\hdots,v_{Cn}$ that get color $i$ in $X_1$, and let $Z'_i$ denote the respective r.v. for $Y_1$. We will couple $\mathbf{Z}:=(Z_1,\hdots,Z_q)$ with $\mathbf{Z}':=(Z_1',\hdots,Z_q')$ so that 
\begin{equation}\label{eq:ZZpZZp}
\Pr(\mathbf{Z}'=\mathbf{Z}+\boldsymbol{\ell})=\Omega(1).
\end{equation}
From this, we clearly obtain a coupling such that with probability $\Theta(1)$ we have $\alpha_i(X_1) = \alpha_i(Y_1)$ for $i\in[q]$. The coupling in \eqref{eq:ZZpZZp} is nearly identical to the one used in \cite[Lemma 6.7]{LNNP}, we give the details for completeness. 

Consider $\mathbf{W}:=(W_1,\hdots,W_q)$, where $\mathbf{W}$ follows the multinomial distribution $\mathrm{Mult}\big(Cn,(\frac{1}{q},\hdots,\frac{1}{q})\big)$ and note that $\mathbf{Z},\mathbf{Z}'$ have the same distribution as $\mathbf{W}$. For $t>0$, let 
\[I(t):=\bigg\{\mathbf{w}=(w_1,\hdots,w_q)\in\mathbb{Z}^q\,\Big| \, w_1,\hdots,w_q\in \Big[\frac{Cn}{q}-t\sqrt{n},\frac{Cn}{q}+t\sqrt{n}\Big], \ w_1+\hdots+w_q=Cn\bigg\}.\] 
Standard deviation bounds (or, alternatively, using Stirling's approximation) yield that, for every constant $t>0$, for $\mathbf{w}=(w_1,\hdots,w_q)\in I(t)$, it holds that 
\begin{equation}\label{eq:qazqazqaz}
\Pr\big(\mathbf{W}=\mathbf{w}\big)\geq \frac{C_0}{(\sqrt{n})^{q-1}},
\end{equation}
for some absolute constant $C_0>0$ (depending only on $q,C,t$). Note that the variance of any coordinate $W_i$ is $\Theta(n)$, but since the sum of $W_i$'s is equal to $Cn$, the random vector $\mathbf{W}$ lies in a $(q-1)$-dimensional space, yielding the denominator $(\sqrt{n})^{q-1}$ in \eqref{eq:qazqazqaz}.

The coupling $\mu$ of $\mathbf{Z},\mathbf{Z}'$ will be defined to be optimal on pairs of the form $(\mathbf{w},\mathbf{w}+\boldsymbol{\ell})$ with $\mathbf{w}\in I(L')$. More precisely, for $\mathbf{w}=(w_1,\hdots,w_q)\in I(L')$, we set
\begin{equation}\label{eq:rfvbgt}
\mu\big(\mathbf{Z}=\mathbf{w}, \mathbf{Z}'=\mathbf{w}+\boldsymbol{\ell}\big):=\min\big\{\Pr\big(\mathbf{W}=\mathbf{w}\big),\Pr\big(\mathbf{W}=\mathbf{w}+\boldsymbol{\ell}\big) \big\}\geq \Omega(n^{-(q-1)/2}),
\end{equation}
where in the last inequality we used \eqref{eq:qazqazqaz} for $t=2L'$ (recall that the coordinates of $\boldsymbol{\ell}$ are bounded in absolute value by $L'\sqrt{n}$). For pairs $(\mathbf{w},\mathbf{w}')\notin \{(\mathbf{w},\mathbf{w}+\boldsymbol{\ell)} \mid \mathbf{w}\in I(L')\}$, the coupling is independent (the construction is analogous to the one used in the proof of the Coupling lemma, see \cite[Section 4.2]{LPW}). Now note that 
\[\mu(\mathbf{Z}=\mathbf{Z}'+\boldsymbol{\ell})\geq \sum_{\mathbf{w}\in I(L')}\mu\big(\mathbf{Z}=\mathbf{w}, \mathbf{Z}'=\mathbf{w}+\boldsymbol{\ell}\big) = \Omega(1),\]
where in the last inequality we used \eqref{eq:rfvbgt} and the fact that the number of $\mathbf{w}$ in $I(L')$ is $\Omega\big((\sqrt{n})^{q-1}\big)$. This proves \eqref{eq:ZZpZZp} with the coupling $\mu$, and hence, modulo the proof of \eqref{eq:truncated} which is given below,  the proof of Lemma~\ref{lem:phases-agree} is complete.

To prove \eqref{eq:truncated}, we may assume w.l.o.g. that the largest component received color 1 (in each of the chains, by the coupling). Let $n'=n-Cn$ and denote by $C_{1,X},C_{1,Y}$ the largest components after the percolation step of the SW dynamics on $X_0, Y_0$ respectively. Since the configurations $X_0$ and $Y_0$ are close to $\m=(a,b,\hdots,b)$, in each of these configurations, exactly one color class is supercritical and the remaining color classes are subcritical in the percolation step (using that $aB>1$ and $bB<1$ from Lemmas~\ref{lehehe} and~\ref{lem:halfopen}). Therefore, by Lemma~\ref{lem:supercritical}, we have with probability $\Theta(1)$  that 
\[\big||C_{1,X}|-|C_{1,Y}|\big|\leq K_0 \sqrt{n}\]
for some (large) constant $K_0>0$.  We will further show that for a (large) constant $K_1>0$, with probability $\Theta(1)$,  it holds that 
\begin{equation}\label{eq:x1xiy1yi}
\bigg|x_1-\Big(\frac{n'}{q}+\big(1-\frac{1}{q}\big)|C_{1,X}|\Big)\bigg|\leq K_1\sqrt{n} \mbox{ and } \Big|x_i-\frac{n'-|C_{1,X}|}{q}\Big|\leq K_1\sqrt{n} \mbox{ for } i\neq 1,
\end{equation}
and, by an identical argument, the analogous inequalities for the $y_i$'s. Combining these, we obtain \eqref{eq:truncated} with $L'=2(K_0+K_1)$.

It remains to show \eqref{eq:x1xiy1yi}. Consider the configuration $X_0$. W.l.o.g. we may assume that color 1 induces the largest color class in $X_0$, so that the assumption 
$\|\alphab(X_0) - \m\|_\infty\leq Ln^{-1/2}$ translates into
\begin{equation*}
|\alpha_1(X_0)-a|\leq Ln^{-1/2}, \quad |\alpha_i(X_0)-b|\leq Ln^{-1/2} \mbox{ for } i\neq 1.
\end{equation*}
From this, we have that color 1 is supercritical in the coloring step of the SW dynamics, while the colors $2,\hdots,q$ subcritical (since it holds that $aB>1$ and $bB<1$ by Lemmas~\ref{lehehe} and~\ref{lem:halfopen}). Let $C_1,C_2,\hdots$ be the components in decreasing order of size after performing the percolation step in $X_0$ and note that $C_1=C_{1,X}$.  We have
\begin{equation*}
E\Big[\sum_{j\geq 2} |C_j|^2\Big] \leq K n
\end{equation*}
for some absolute constant $K>0$. To see this, use the bound in Lemma~\ref{lem:supercritical} and equation \eqref{eq:variancesuper} for the (supercritical) color class 1 and Item~\ref{it:sub} of Lemma~\ref{lemia1}  for each of the (subcritical) color classes $2,\hdots,q$. By Markov's inequality (and restricting our attention to components other than the isolated vertices $\{v_1,\hdots,v_{Cn}\}$) we obtain that with probability $\Theta(1)$ it holds that
\begin{equation}\label{eq:qwertyqwerty}
\sum_{j\geq 2;\, C_j\neq \{v_1\},\hdots, \{v_{Cn}\}} |C_j|^2 \leq K' n
\end{equation}
for some absolute constant $K'>0$. Now, for $i=1,\hdots,q$ let $J_i$ be the number of vertices colored with $i$ among the vertices other than $v_1,\hdots,v_{Cn}$ and those that belonged to the component $C_{1,X}$. Note that 
\begin{equation}\label{eq:Jia}
x_1=|C_{1,X}|+J_1, \quad x_i=J_i \mbox{ for } i\neq 1.
\end{equation}
Observe that $E[J_i]=(n'-|C_{1,X}|)/q$. Further, using \eqref{eq:qwertyqwerty} and Azuma's inequality, we obtain that with probability $\Theta(1)$ it holds that 
\begin{equation}\label{eq:Jib}
\Big|J_i-\frac{n'-|C_{1,X}|}{q}\Big|\leq K'' \sqrt{n} \mbox{ for } i=1,\hdots,q.
\end{equation}
for some absolute constant $K''>0$. Combining \eqref{eq:Jia} and \eqref{eq:Jib} yields \eqref{eq:x1xiy1yi} (with $K_1=K''$), as wanted. In turn, this completes the proof of \eqref{eq:truncated} and hence the proof of Lemma~\ref{lem:phases-agree}.

As mentioned earlier, the proof of Lemma~\ref{lem:phases-uniform} is completely analogous. The only difference is that now, where the configurations $X_0,Y_0$ are close to $\u=(1/q,\hdots,1/q)$, all color classes are subcritical in the percolation step of the SW algorithm (using that $B<\Bh$). Hence, there is no need to consider the size of the biggest components $C_{1,X}$ and $C_{1,Y}$. In particular, adapting the above arguments yields the following analogue of \eqref{eq:x1xiy1yi}:
\begin{equation}\label{eq:x1xiy1yib}
\Big|x_i-\frac{n'}{q}\Big|\leq K_1\sqrt{n} \mbox{ for } i\in [q].
\end{equation}
Using \eqref{eq:x1xiy1yib} (and the analogous inequalities for $y_i$'s), we obtain \eqref{eq:truncated}; the remaining bit of the proof of Lemma~\ref{lem:phases-uniform} is in all other respects identical to the proof  of Lemma~\ref{lem:phases-agree} (i.e., using the isolated vertices to couple $X_1$ and $Y_1$).

This concludes the proofs.
\end{proof}

\section{Fast mixing for $B>\Bh$}
\label{sec:fast}
In this section, we prove that the SW algorithm mixes in $O(\log n)$ steps for all $B>\Bh$.

Let $\eps>0$  and consider a state $X_t$ of the SW algorithm. We say that  a color $i$ is $\eps$-heavy if $\alpha_i(X_t)\geq
(1+\eps)/B$; it is $\eps$-light if $\alpha_i(X_t)\leq
(1-\eps)/B$.  We will show that the
SW algorithm has a reasonable chance of moving into a state where
one color is $\eps$-heavy and the remaining $q-1$ colors are
$\eps$-light.

\begin{lemma}\label{hroch1}
Assume $B>\Bh$ is a constant. There exists a constant $\eps>0$ such that the following hold for all sufficiently large $n$.  For any initial state $X_0$, with probability $\Theta(1)$ the next state $X_1$ has one $\eps$-heavy color and the remaining $q-1$ colors are $\eps$-light. Moreover, if $X_0$ has one $\eps$-heavy color and the remaining $q-1$ colors are $\eps$-light, the same is true for $X_1$ with probability $1-o(1)$.
\end{lemma}

Before proving Lemma~\ref{hroch1} we will need the following
function $g:[0,1]\ra [0,1]$ which roughly captures the size
of the largest component in $G(zn,B/n)$. Specifically, for
$z\leq 1/B$ we set $g(z)=0$; for $z> 1/B$ we set $g(z) = zx$,
where $x$ is the unique solution of $x+\exp(-zBx) = 1$ in
$(0,1]$. Note that the functions $F$ and $g$ are connected by the  relation 
\[F(z)=\frac{1}{q}+\Big(1-\frac{1}{q}\Big)g(z) \mbox{ for all $z\in(1/B,1]$}.\]
The following inequality will be used to conclude the existence of heavy colors. 
\begin{lemma}\label{hroch2}
Assume $B>\Bh$.  Then, for all $\alpha_1,\hdots,\alpha_q\geq 0$ with $\alpha_1+\cdots+\alpha_q=1$, it holds that 
\begin{equation*}
\sum_{i\in[q]} g(\alpha_i)\geq g\Big(1-\frac{q-1}{B}\Big) > 1 -\frac{q}{B}.
\end{equation*}
\end{lemma}
\begin{proof}[Proof of Lemma \ref{hroch2}]
For convenience, let $W:=\sum_{i\in[q]} g(\alpha_i)$. Note that $g(z)$ is increasing and concave for $z>1/B$ (this follows by Lemma~\ref{lem:Fshape} since $F(z)=\frac{1}{q}+(1-\frac{1}{q})g(z)$ for $z\in (1/B,1]$). 

To minimize $W$, observe that
\begin{enumerate}
\item \label{it:fff3f2f4} If $\alpha_i>1/B$ and $\alpha_j<1/B$ then we can decrease the value
of $W$ by decreasing $\alpha_i$ and increasing $\alpha_j$ (since
$g(z)=0$ for all $z\leq 1/B$ and $g(z)$ is increasing for
$z>1/B$). 
\item \label{it:21w1s}If $1/B<\alpha_i<\alpha_j$ then we can decrease the value of $W$ by decreasing $\alpha_i$ and increasing $\alpha_j$ (since
$g(z)$ is concave for $z>1/B$).
\end{enumerate}
Since $B>\Bh=q$ and $\alpha_1+\cdots+\alpha_q=1$, we have that at least one of the $\alpha_i$'s is strictly larger than $1/B$. Thus, from Items~\ref{it:fff3f2f4} and~\ref{it:21w1s}, it follows that $W$ is minimized when all but one of the $\alpha_i$'s are equal to $1/B$ (the value of the remaining $\alpha_i$ is given by the condition $\alpha_1+\cdots+\alpha_q=1$). Since $g(1/B)=0$, it follows that  
\[W\geq  g\Big(1-\frac{q-1}{B}\Big).\]

It remains to show that $g(z)> 1 -\frac{q}{B}$, where $z:=1-(q-1)/B$. Note that $z>1/B$ from $B>q$. Let $x\in (0,1)$ be the solution of $x+\exp(-zBx) = 1$. The inequality $g(z)> 1 -\frac{q}{B}$ is equivalent to
\begin{equation}\label{prrrr1}
x > \frac{B-q}{B-q+1}.
\end{equation}
For the sake of contradiction, suppose that~\eqref{prrrr1} is false, that is, $x\leq (B-q)/(B-q+1)$. Then,
\begin{equation}\label{prrrr2}
1-\frac{q-1}{B} = -\frac{\ln(1-x)}{Bx} \leq \frac{(B-q+1)\ln
(B-q+1)}{B(B-q)},
\end{equation}
where the equality follows from $x+\exp(-zBx) = 1$ and the inequality follows from the fact that $x\mapsto
-\frac{\ln(1-x)}{x}$ is an increasing function on $(0,1)$.
Inequality~\eqref{prrrr2} yields that $B-q\leq\ln (1+B-q)$,
which is false (since $B-q>0$), and hence we have a contradiction.
This shows that~\eqref{prrrr1} is true.
\end{proof}

We are now ready to prove Lemma~\ref{hroch1}.
\begin{proof}[Proof of Lemma~\ref{hroch1}.]
Let $W:=g\big(1-\frac{q-1}{B}\big)$. By Lemma~\ref{hroch2}, there exists a small constant $\epsilon>0$ such that
\[W-\epsilon\geq 1-\frac{q}{B}(1-2\epsilon).\] 
Since the function $g(z)$ is continuous and $g(z)=0$ for all $z\leq 1/B$, there exists a small constant $\eta>0$ such that for all $z\leq (1+\eta)/B$ it holds that $g(z)\leq \epsilon/q$.

For $i\in[q]$, let $m_i$ be the number of vertices of color $i$ in $X_0$ and let $\alpha_i=m_i/n$.  By Lemma~\ref{hroch2},  
\[\sum_{i\in[q]}g(\alpha_i)\geq W.\]
Perform the percolation step of the SW algorithm on the color class $i$ and denote by $G_i$ be the resulting graph. Moreover, let $C_1^{(i)}, C_2^{(i)},\hdots$ be the components of $G_i$ in decreasing order of size. Note that $G_i$ is distributed as $G(n\alpha_i, B/n)$.

To prove the first part of the lemma, note that for each color $i\in[q]$ the following hold with probability $1-o(1)$:
\begin{itemize}
\item If $B\alpha_i\geq 1+\eta$, the size of the largest component in $G_i$ is $ng(\alpha_i)+o(n)$ (by Lemma~\ref{lem:supercritical}).
\item If $B\alpha_i\leq 1+\eta$, by the choice of $\eta$ we have $g(\alpha_i)\leq \epsilon/q$ and therefore the largest component in $G_i$ is trivially at least $g(\alpha_i)n-\frac{\epsilon}{q}n$. 
\end{itemize}
Moreover, with $A$ being the constant in Lemma \ref{lemia3}, we have that for each color $i\in[q]$ the following hold with positive probability (not depending on $n$):
\begin{enumerate}
\item   If $B\alpha_i\geq (1-A m_i^{-1/3})/m_i$, then  $\sum_{j> 1} |C^{(i)}_j|^2 \leq m_i^{4/3}\leq n^{4/3}$ (by Lemma~\ref{lemia3}).
\item If $(1-A m_i^{-1/3})/m_i>B\alpha_i$, then $\sum_{j\geq 1} |C^{(i)}_j|^2 =O(  n^{4/3})$ (by Item~\ref{it:sub} of Lemma~\ref{lemia1}).
\end{enumerate}
It follows that for all sufficiently large $n$, with probability $\Theta(1)$, after the percolation step of the SW algorithm, it holds that 
\[\sum_{i\in[q]}|C_1^{(i)}|\geq (W-\epsilon)n\geq \Big(1-\frac{q}{B}(1-2\epsilon)\Big)n \mbox{\ \ and\ \ } \sum_{i\in[q]}\sum_{j\geq 2}|C_j^{(i)}|^2=o(n^2).\]

Now,  in the coloring step of the SW algorithm, with probability
$q^{-q} = \Theta(1)$,   all of  the components $C_1^{(i)}$ with $i\in [q]$ receive color $1$. Conditioned on that, the expected number of vertices which get the color $k\neq 1$ after the coloring step of the SW algorithm is 
\[\frac{n-\sum_{i\in [q]} \big|C_1^{(i)}\big|}{q}\leq n(1-2\epsilon)/B.\]
Thus, using Azuma's inequality, we obtain that with probability $\Theta(1)$, for all colors $k\neq 1$, the number of vertices which get the color $k$ after the coloring step of the SW algorithm is at most $n(1-\epsilon)/B$, which implies that the number of vertices which get the color 1 is at least $n(1-(q-1)(1-\epsilon)/B)\geq n(1+\epsilon)/B$. Thus, combining all the above, we obtain that, after one step of the SW algorithm, with probability $\Theta(1)$, color 1 is $\epsilon$-heavy and all other colors are $\epsilon$-light.

The second part of the lemma is analogous, the only difference is that now there is a unique $\epsilon$-heavy color class $i$, which is therefore supercritical in the percolation step; all other color classes are $\epsilon$-light and therefore subcritical. This allows us to improve the probability bounds in the previous analysis.  In particular, by Lemma~\ref{lem:supercritical} (applied to the supercritical color) and Lemma~\ref{l1} (applied to the subcritical colors), we obtain that 
with probability $1-o(1)$, after the percolation step of the SW algorithm,  there is just one linear-sized component  of size $g(\alpha_i)n+o(n)$ and the remaining components have size $o(n)$. Since $g(\alpha_j)=0$ for all $j\neq i$, Lemma~\ref{hroch2} yields that $g(\alpha_i)\geq W$. W.l.o.g., we may assume that this unique linear-sized  component receives the color 1. Then, using Azuma's inequality just as above, we obtain that, after one step of the SW algorithm, with probability $1-o(1)$, all  colors other than color 1 are $\epsilon$-light and color $1$ is $\epsilon$-heavy.

This completes the proof of Lemma \ref{hroch1}.
\end{proof}

After applying Lemma \ref{hroch1} the behavior of the 
SW algorithm in one step will be
controlled by the function $F$ (cf. Section~\ref{sec:defF}). We use this to show that, after
$O(1)$ steps, with constant probability, the state of SW will be close to the majority phase $\m$; recall that $\m=(a,b,\hdots,b)$ where $a>1/q$ is the unique fixpoint of $F$ and $b=(1-a)/(q-1)$.
\begin{lemma}\label{slon7}
Assume $B>\Bh$ is a constant. For any constant $\delta>0$, for all sufficiently large $n$ and any starting
state $X_0$, after $T=O(1)$ steps, with probability $\Theta(1)$
the SW algorithm moves to a state $X_T$ such that
$\|\alphab(X_T)-\m\|_{\infty} \leq \delta$.
\end{lemma}

\begin{proof}
Let $\epsilon>0$ be the constant in Lemma~\ref{hroch1}. Then, with probability $\Theta(1)$, the state $X_1$ has
one $\eps$-heavy color and the remaining $q-1$ colors are
$\eps$-light. 

Assume that at time $t\geq 1$ we are at a state $X_t$ with one $\eps$-heavy color and $q-1$ colors which are $\eps$-light. Then by the second part of Lemma~\ref{hroch1}, the same is true for the state $X_{t+1}$ with probability $1-o(1)$. In fact, let $zn$ be the number of vertices of the heavy color class in $X_t$; we claim that with probability $1-o(1)$, in $X_{t+1}$ the heavy color class has $F(z)n+o(n)$ vertices, while all the other color classes have $\frac{1-F(z)}{q-1}n+o(n)$ vertices. Indeed, in the percolation step of the SW dynamics, exactly one color class is supercritical and the remaining $q-1$ color classes are subcritical. By Lemma~\ref{lem:supercritical} (applied to the supercritical color) and Lemma~\ref{l1} (applied to the subcritical colors), we obtain that 
with probability $1-o(1)$, after the percolation step of the SW algorithm,  there is just one linear-sized component $C$  of size $g(z)n+o(n)$ and the remaining components have size $o(n)$.  W.l.o.g., we may assume that the component $C$ receives the color 1. Then, using Azuma's inequality just as in the proof of Lemma~\ref{hroch1}, we obtain that, with probability $1-o(1)$, for each color $k\neq 1$, $\frac{1-g(z)}{q}n+o(n)$ vertices receive the color $k$ and the remaining $\frac{n}{q}+\frac{q-1}{q}g(z)n+o(n)=F(z)n+o(n)$ vertices receive the color 1, as claimed. 

We thus obtain that for any constant integer $T\geq 2$, with probability $\Theta(1)$
the SW algorithm moves to a state $X_T$ where one color class has $\alpha n+o(n)$ vertices and each of the remaining color classes has $\frac{1-\alpha}{q-1} n+o(n)$ vertices, where $\alpha$ belongs to the interval $F^{(T)}([1/q,1])$ (recall that $F^{(T)}$ is the $T$-th iterate of the function $F$). Since $F$ is increasing (Lemma~\ref{lem:Fshape}), we have $F^{(T)}([1/q,1])=[F^{(T)}(1/q),F^{(T)}(1)]$. Since $B>\Bh$, by Lemma~\ref{lem:fixpoints} we have that $F$ has a unique fixpoint $a>1/q$. Hence, using also again that $F$ is increasing, for any constant $\delta>0$, there is a constant $T$  such that
$[F^{(T)}(1/q),F^{(T)}(1)] \subseteq [a-\delta/2,a+\delta/2]$. Thus
in $T$ steps, with probability $\Theta(1)$, we are within $\ell_\infty$-distance $\delta$ of $\m$ (with room to spare to absorb the $o(n)$ fluctuations of the color classes).
\end{proof}

Then we show that once we are at constant distance from $\m$ then in $O(\log n)$ steps
the distance to $\m$ further decreases to $O(n^{-1/2})$. 

\begin{lemma}\label{slon8}
For $B>\Bu$, there exist $\delta, L>0$ such that the following is true. Suppose that we
start at a state $X_0$ such that $\|\alphab(X_0)-\m\|_{\infty} \leq
\delta$. Then in $T=O(\log n)$ steps with probability $\Theta(1)$
the SW algorithm ends up in a state $X_t$ such that
\begin{equation}\label{hrb1}
\|\alphab(X_T)-\m\|_{\infty}\leq L n^{-1/2}.
\end{equation}
\end{lemma}
\begin{proof}
Recall that $\m=(a,b,\hdots,b)$ where $a>1/q$ is a jacobian attractive fixpoint of $F$ and $b=(1-a)/(q-1)$. Moreover, by Lemmas~\ref{lehehe} and~\ref{lem:halfopen}, it holds that $aB>1$ and $bB<1$.

Let $\delta>0$, $c\in(0,1)$ be constants such that for all  $z\in [a-\delta,a+\delta]$ it holds that $|F(z)-a|\leq c |z-a|$ and $z B>1$, $(1-z) B/(q-1)<1$. Note that the existence of such constants $\delta,c$ is guaranteed by the jacobian attractiveness of the fixpoint $a$ throughout the regime $B>\Bu$ (Lemma~\ref{lem:fixpoints}) and the facts $a B>1,\, bB<1$.

Define the geometrically decreasing sequence $\{w_t\}_{t\geq 0}$ by setting $w_0=\delta n^{1/2}$ and $w_{t}=\frac{1+c}{2}w_{t-1}$. Further, let $T:=\big\lceil\frac{\frac{1}{2}\log n}{\log \frac{2}{1+c}}\big\rceil-K$ where $K>0$ is a large constant to be chosen later.  Note that for any constant $K$,  it holds that
\[\frac{1+c}{2}L\leq w_T\leq L,\mbox{ where }L:=\delta (2/(1+c))^{K}.\]
Thus, to prove the lemma, it suffices to show the following (slightly stronger) statement: there exists a constant $K>0$ such that with probability $\Theta(1)$, 
\begin{equation}\label{eq:alphaxtwt}
\mbox{for all $t=0,1,\hdots, T$, it holds that $\norm{\alphab(X_t)-\m}_{\infty}\leq w_t n^{-1/2}$}.
\end{equation}

The main step in the proof is to track one step of the SW dynamics. Specifically, we will show that there exist constants $L',C>0$ such that for all $w_t\in[L',\delta n^{1/2}]$, for any state $X_t$ such that $\norm{\alphab(X_t)-\m}_{\infty}\leq w_t n^{-1/2}$, with probability at least $\exp(-C/w_t)$ it holds that 
\begin{equation}\label{eq:xtmxtm}
\|\alphab(X_{t+1})-\m\|_{\infty} \leq w_{t+1}n^{-1/2}.
\end{equation} 
To conclude \eqref{eq:alphaxtwt} from \eqref{eq:xtmxtm}, note that by choosing $K$ large, we can ensure that $w_0\geq\hdots\geq w_T\geq L'$ and hence the probability of the event in \eqref{eq:alphaxtwt} is at least $\prod^{T}_{t=0}\exp(-C/w_t)$. The latter product is bounded by a positive constant, since $w_t$ is a geometrically decreasing sequence.  

It remains to show \eqref{eq:xtmxtm}. In particular, assume that at time $t$ it holds that $\norm{\alphab(X_t)-\m}_{\infty}\leq w_t n^{-1/2}$ where $w_t\in[L',\delta n^{1/2}]$ for some large constant $L'$ to be specified later. By the choice of the constant $\delta$, in the percolation step of the SW dynamics, exactly one color class is supercritical and the remaining $q-1$ color classes are subcritical. Denote by $C_1,C_2, \dots$ all the connected components  after
the percolation step, sorted in decreasing order of size. By the second inequality in \eqref{eq:variancesuper} of  Lemma~\ref{lem:supercritical} (applied to the supercritical color) and part~\ref{it:sub} of Lemma~\ref{lemia1} (applied to the subcritical colors), we obtain that there exists a constant $K'>0$ such that
$$
E\Big[\sum_{i\geq 2} |C_i|^2 \Big] \leq K' n.
$$
Let $w_t':=\frac{(1-c)}{2(1+\sqrt{K'})}w_t$; the choice of $w_t'$ will become apparent shortly. Note that by choosing $L'$ to be a large constant, we can ensure that $w_t'$ is larger than any desired constant (whenever $w_t\in[L',\delta n^{1/2}]$). 

By Markov's inequality,  it holds that
\begin{equation}\label{grg1}
P\Big(\sum_{i\geq 2} |C_i|^2 \leq w_t' K' n\Big) \geq 1-1/w_t'.
\end{equation}
Assuming that the event in~\eqref{grg1} occured, by Azuma's inequality, in the
coloring step of the SW algorithm the number $Z_i$ of vertices in $C_2\cup C_3\dots$
that receive color $i$ is concentrated around the expectation, i.e., 
\begin{equation}\label{grg2}
P\Big(|Z_i - E[Z_i]|\geq w_t' \sqrt{K' n}\Big)\leq 2\exp(-w_t'/2).
\end{equation}
Let $zn$ be the number of vertices in the largest color class of $X_t$; by the choice of $\delta$ in the beginning, we have that $zB>1$ and hence the largest color class is supercritical in the percolation step of SW. Therefore, by Lemma~\ref{lem:supercritical} (equation \eqref{eq:moddev}), 
\begin{equation}\label{grg3}
P( \big||C_1| - g(z) n \big| \geq w_t' \sqrt{n} ) \leq U/(w_t')^2.
\end{equation}
Combining~\eqref{grg1}, \eqref{grg2}, and~\eqref{grg3} (and choosing $L'$ to be a large constant relative to $K',U,1/(1-c),q$), we obtain
that with probability at least
\begin{equation}\label{ios1}
(1-1/w_t')\big(1-2q\exp(-w_t'/2)-U/(w_t')^2\big)\geq \exp(-10/w_t')=\exp(-C/w_t), \quad C:=\frac{20(1+\sqrt{K'})}{1-c},
\end{equation}
we have that
\begin{equation}\label{grg4}
\Big\|\alphab(X_{t+1}) -
\Big(F(z),\frac{1-F(z)}{q-1},\dots,\frac{1-F(z)}{q-1}\Big)\Big\|_\infty \leq w_t'
(1+\sqrt{K'}) n^{-1/2}.
\end{equation}
By the choice of the constants $\delta,c$, we have
\begin{equation}\label{grg5}
\Big\|\Big(F(z),\frac{1-F(z)}{q-1},\dots,\frac{1-F(z)}{q-1}\Big) - \m\Big\|_\infty \leq c
\| \alphab(X_t) - \m\|_\infty\leq c w_t n^{-1/2}.
\end{equation}
Equations~\eqref{grg4} and~\eqref{grg5} combined yield that with probability $\geq \exp(-C/w_t)$, it holds that
\begin{equation}\label{grg6}
\begin{split}
\|\alphab(X_{t+1})-\m\|_\infty \leq w_t' (1+\sqrt{K'}) n^{-1/2} + c
w_tn^{-1/2} \leq \frac{c+1}{2} w_tn^{-1/2}=w_{t+1}n^{-1/2},
\end{split}
\end{equation}
where the last inequality follows from  $w_t'=\frac{(1-c)}{2(1+\sqrt{K'})}w_t$. This proves \eqref{eq:xtmxtm} and therefore completes the proof of Lemma~\ref{slon8}.
\end{proof}

From Lemmas \ref{lem:vertices-agree}, \ref{lem:phases-agree}, \ref{slon7} and  \ref{slon8}
we conclude the following.

\begin{corollary}\label{cor:BgBh}
Let $B>\Bh$ be a constant. The mixing time of the Swendsen-Wang algorithm
on the complete graph on $n$ vertices is $O(\log n)$.
\end{corollary}
\begin{proof}
Let $\epsilon>0$ be a small constant and consider two copies $X_t,Y_t$ of the SW chain. We will show that for some $T=O(\log n)$, there exists a coupling such that $\Pr(X_T\neq Y_T)\leq \epsilon$. 

Let $\delta,L$ be as in Lemma~\ref{slon8}. By Lemma~\ref{slon7}, for some $T_1=O(1)$ with probability $\Theta(1)$ we have that
\begin{equation}\label{eq:XT1YT1}
\|\alphab(X_{T_1}) - \m\|_\infty\leq \delta\mbox{ and }\|\alphab(Y_{T_1}) - \m\|_\infty\leq \delta.
\end{equation}
By Lemma~\ref{slon8}, for some  $T_2=O(\log n)$ with probability $\Theta(1)$, we have that 
\begin{equation}\label{eq:couplecouple}
\|\alphab(X_{T_1+T_2}) - \m\|_\infty\leq Ln^{-1/2}\mbox{ and }\|\alphab(Y_{T_1+T_2}) - \m\|_\infty\leq Ln^{-1/2}.
\end{equation}
Let $T':=T_1+T_2$. Conditioning on \eqref{eq:couplecouple}, by Lemma~\ref{lem:phases-agree} there exists a coupling that with probability $\Theta(1)$, for $T_3=T'+1$, it holds that $\alphab(X_{T_3})=\alphab(Y_{T_3})$. Conditioned on $\alphab(X_{T_3})=\alphab(Y_{T_3})$, by Lemma~\ref{lem:vertices-agree}, for every constant $\epsilon'>0$ there exists $T_4=O(\log n)$ and a second coupling such that $\Pr(X_{T_3+T_4}\neq  Y_{T_3+T_4})\leq \epsilon'$. By letting $\epsilon'$ to be a sufficiently small constant, we obtain a coupling and some $T=O(\log n)$ such that $\Pr(X_T\neq Y_T)\leq \epsilon$, as wanted.
\end{proof}

\section{Fast mixing for $B=\Bh$}
\label{sec:fast-Bh}

The proof resembles the case $B>\Bh$, though we have to account more carefully for the mixing time of the chain for configurations which are close to uniform. In particular, for starting configurations which are $\epsilon$-far from being uniform, a straightforward modification of the proof for $B>\Bh$ gives that the SW chain mixes rapidly. The main difficulty in the case $B=\Bh$ is to show that the chain escapes from starting configurations which are close to uniform. 
We will show that this happens after roughly $\log n$ steps. More precisely, we have the following lemma.

\begin{lemma}\label{hrochhroch1}
Assume $B=\Bh$. There exists constant $\eps>0$ such that for
any $n$ and any initial state $X_0$ with probability $\Theta(1)$
after $T_1=O(\log n)$ steps, $X_{T_1}$ has an $\eps$-heavy color and the remaining
$q-1$ colors are $\eps$-light.
\end{lemma}

Lemma~\ref{hrochhroch1} yields the following analogue of Lemma~\ref{slon7} (note here the logarithmic bound on $T$).
\begin{lemma}\label{slonslon7}
Assume $B=\Bh$. For any constant $\delta>0$ and any starting
state $X_0$, after $T=O(\log n)$ steps, with probability $\Theta(1)$
the SW algorithm moves to state $X_T$ with
$\|\alphab(X_T)-\m\|_{\infty} \leq \delta$.
\end{lemma}
\begin{proof}[Proof of Lemma~\ref{slonslon7}]
From Lemma~\ref{hrochhroch1}, for some (small) constant $\epsilon>0$, we have that for $T_1=O(\log n)$, with probability $\Theta(1)$, $X_{T_1}$ has an $\eps$-heavy color and the remaining
$q-1$ colors are $\eps$-light. Using Lemma~\ref{lem:Fshape} ($F$ is increasing), the second part of Lemma~\ref{lem:jacobianuniform} (the uniform fixpoint is jacobian repulsive) and Corollary~\ref{lem:exist} (there exists a unique fixpoint of $F$ in the interval $(1/q,1]$), we obtain that for constant $T_2$ (depending on $\delta$) we have $F^{(T_2)}([(1+\eps)/q,1]) \subseteq [a-\delta/2,a+\delta/2]$, so the same arguments as in the proof of Lemma~\ref{slon7} yield  that $\norm{\alphab(X_{T_1+T_2})-\m}_{\infty}\leq \delta$ with probability $\Theta(1)$, as wanted.
    \end{proof}

Using Lemma~\ref{lem:vertices-agree} (note that it applies to all $B>0$) and Lemmas~\ref{lem:phases-agree} and~\ref{slon8} (note that these apply to all $B>\Bu$), we may conclude the following from Lemma~\ref{slonslon7}.
\begin{corollary}\label{cor:BBh}
Let $B=\Bh$. The mixing time of the Swendsen-Wang algorithm
on the complete graph on $n$ vertices is $O(\log n)$.
\end{corollary}
\begin{proof}
The proof is completely analogous to the proof of Corollary~\ref{cor:BgBh}, the only difference is that now we use Lemma~\ref{slonslon7} to argue that \eqref{eq:XT1YT1}  holds with probability $\Theta(1)$ for $T_1=O(\log n)$.
\end{proof}

We next turn to the proof of Lemma~\ref{hrochhroch1}. We will use the following definition. For $W>0$, a state $X$ will be called \emph{$W$-good} if $X$ has a $W$-heavy color and the remaining $q-1$ colors are $(W/2q)$-light.

\begin{lemma}\label{lem:wgood}
Let $B=\Bh$. For any starting state $X_0$ and an arbitrary constant $w>0$, with probability at least $p(w)>0$ (not depending on $n$) the next state $X_1$ of the SW dynamics  is $wn^{-1/3}$-good.
\end{lemma}

\begin{lemma}\label{lem:epsgood}
Let $B=\Bh$. There exist absolute constants $c_1,c_2,C>0$ such that for all $n$ sufficiently large the following holds. For all $w$ such that $c_1\leq w \leq c_2 n^{1/3}$,  for every $wn^{-1/3}$-good starting state $X_0$, the next state of the SW dynamics $X_1$ is $(13/12)w n^{-1/3}$-good with probability at least $\exp(-C/w)$.
\end{lemma}
Before proceeding, let us briefly motivate Lemmas~\ref{lem:wgood} and~\ref{lem:epsgood}. First, we explain the origin of the constant 13/12 in Lemma~\ref{lem:epsgood}, whose value is somewhat arbitrary, any constant strictly smaller than 4/3 (and greater than 1) would work for all $q\geq 3$. To understand where the  constant 4/3 comes from, recall from Lemma~\ref{lem:jacobianuniform} that the uniform phase $u=1/q$ is a jacobian repulsive fixpoint of $F$ (for $B=\Bh$) and, more precisely, $F'(1/q)=2 (q - 1)/q$ (note that $F'(1/q)>1$ for all $q>2$). Then, just observe that $\min_{q\geq 3}\{2 (q - 1)/q\}=4/3$. 

Thus, for any $4/3>c>1$ (or, slightly less loosely, when $F'(1/q)>c>1$), whenever  $\norm{\alphab(X_t)-\u}_\infty$ is sufficiently small, for all sufficiently large $n$, one would expect that  
\[\norm{\alphab(X_{t+1})-\u}_\infty \geq  c\norm{\alphab(X_{t})-\u}_\infty.\] 
We show that this indeed holds by accounting carefully for color classes which are in the critical window  for the percolation step of the SW dynamics (technically,  to establish the probability bound in Lemma~\ref{lem:epsgood}, we need that $\norm{\alphab(X_t)-\u}_\infty=\Omega(n^{-1/3})$).  Lemma~\ref{lem:epsgood} thus proves that an initial displacement of $\Omega(n^{-1/3})$, which is guaranteed with constant probability from Lemma~\ref{lem:wgood}, increases geometrically.

Lemma~\ref{hrochhroch1} follows immediately from Lemmas~\ref{lem:wgood}  and~\ref{lem:epsgood}.
\begin{proof}[Proof of Lemma~\ref{hrochhroch1}]
Let $c_1,c_2,C$ be the constants from Lemma~\ref{lem:epsgood}. Define $w_t$ by $w_1=c_1$ and $w_{t}=(13/12)w_{t-1}$. Moreover, let $0<\epsilon_0<c_2$ and set $t_0=\left\lfloor \log (\epsilon_0 n)/\log  (13/12)\right\rfloor$. By Lemmas~\ref{lem:wgood} and~\ref{lem:epsgood}, for any starting state $X_0$, the state $X_t$ is $w_t$-good for all $t=1,\hdots,t_0$ with probability at least $p(w_1)\prod^{t_0}_{t=2}\exp(-C/w_t)=:L$. Note that the product in the expression for $L$ is bounded by an absolute positive constant, since the series $\sum_{t\geq 1}1/w_t$ converges.

It follows that for any positive $\epsilon<\epsilon_0/(10q)$, with positive probability (not depending on $n$), $X_{t_0}$ has an $\epsilon$-heavy color and the remaining $q-1$ colors are  $\epsilon$-light, as wanted.
    \end{proof}

We next prove Lemmas~\ref{lem:wgood} and~\ref{lem:epsgood}.
\begin{proof}[Proof of Lemma~\ref{lem:wgood}]
We will write $\alpha_i$ as a shorthand for $\alpha_i(X_0)$, and denote $m_i=n\alpha_i$. In each step of the Swendsen-Wang algorithm, the percolation step for color $i$ picks a graph $G_i$ from $G(m_i,q\alpha_i/m_i)$. Let $C_1^{(i)},C_2^{(i)}, \hdots$ be the components of $G_i$ in decreasing order of  size.

Let $A,L$ be the constants in Lemma \ref{lemia3} and let $w\geq L$. For each color $i$ the following hold with positive probability (not depending on $n$):
\begin{enumerate}
\item  \label{it:sup2} If $q\alpha_i\geq (1-A m_i^{-1/3})/m_i$, then  $|C^{(i)}_1|\geq 100wq^2n^{2/3}$, $\sum_{j> 1} |C^{(i)}_j|^2 \leq m_i^{4/3}\leq n^{4/3}$ (by Lemma~\ref{lemia3}).\footnote{We remark that the choice of the constant 100 in the bound for $|C^{(i)}_1|$ is somewhat arbitrary, any sufficient large constant would work; similar remarks apply for the explicit constants 80 and 50 appearing in the proof of Lemma~\ref{lem:wgood}.}
\item \label{it:sub2} If $(1-A m_i^{-1/3})/m_i>q\alpha_i$, then $\sum_{j\geq 1} |C^{(i)}_j|^2 \leq  n^{4/3}$ (by Item~\ref{it:sub} of Lemma~\ref{lemia1}).
\end{enumerate}

Note that for at least 1 color we have $q\alpha_i\geq 1$ (since the $\alpha_i$'s sum to 1). Let $S=\{i\in[q]: q\alpha_i\geq 1\}$ and consider all the components \emph{different} from $C^{(i)}_1$, $i\in S$.  Color these
components  independently by a uniformly random color from $[q]$. Let $A_i$ be the number of vertices of color
$i$. By Azuma's inequality and a union bound we have that with probability at least $1-2q \exp( - 50w^2q)$, for each $i\in[q]$ it holds that
\begin{equation*}
\Big|A_i -\frac{n-\sum_{i\in S}|C^{(i)}_1|}{q}\Big|\leq (10wq)n^{2/3}.
\end{equation*}
With probability at least $q^{-q}$ each of $C^{(i)}_1$ with $i\in S$ receives color 1. Let $A_i'$ be the
number of vertices of color $i$ after the coloring step of the SW algorithm. Note, we have $A_1'=A_1 + \sum_{i\in S}|C^{(i)}_1|$ and $A_i' = A_i$ for $i\geq 2$. We obtain that with probability at
least $q^{-q}\big(1-2q\exp( -50w^2q)\big)>0$
\begin{equation*}
|A_1'| \geq \frac{n}{q} +\left(\sum_{i\in S}|C^{(i)}_1|\right)\left(1-\frac{1}{q}\right)-(10wq)n^{2/3} \geq \frac{n}{q} + (80wq^2) n^{2/3},
\end{equation*}
and for all $i\in\{2,\dots,q\}$
\begin{equation*}
|A_i'| \leq \frac{n}{q} -\frac{1}{q}\left(\sum_{i\in S}|C^{(i)}_1|\right)+(10wq)n^{2/3}\leq \frac{n}{q} - (90wq) n^{2/3}.
\end{equation*}
This concludes the proof.
    \end{proof}

\begin{proof}[Proof of Lemma~\ref{lem:epsgood}]
W.l.o.g., we may assume that the color classes $S_1,S_2,\hdots,S_q$ of $X_0$ satisfy
\begin{equation}\label{easum}
|S_1|\geq \frac{n}{q} + w n^{2/3}\quad\mbox{and}\quad
|S_i|\leq \frac{n}{q} - \frac{w}{2q} n^{2/3}\quad\mbox{for}\ i\in\{2,\dots,q\}.
\end{equation}
Now we make a step of the Swendsen-Wang algorithm.
Let $C_1,C_2,\dots,$ be all the connected components after the percolation step of the
Swendsen-Wang algorithm, listed in decreasing size. By Lemma~\ref{lemia1} (first part for
the color classes $i=2,\hdots,q$ and second part for the 
 color class $i=1$) we have
$$
E\Big[\sum_{j\geq 2} |C_j|^2\Big] \leq \frac{2Kn^{4/3}}{w}.
$$
By Markov's inequality
\begin{equation}\label{zuquo1}
P\Big(\sum_{j\geq 2} |C_i|^2\geq n^{4/3} \Big)\leq \frac{2K}{w}.
\end{equation}
By Lemma~\ref{lemia1} (part \ref{it:suplargest}), there exists a constant $c>0$ such that 
\begin{equation}\label{zuquo2}
P\Big(|C_1|\leq (7/4) wn^{2/3} \Big)\leq K\exp(-c q^2 w^3).
\end{equation}

For all sufficiently large $w$, we may assume that the events in~\eqref{zuquo1} and~\eqref{zuquo2} occurred, that is, $|C_1|\geq (7/4) w n^{2/3}$ and $\sum_{i\geq 2} |C_i|^2\leq n^{4/3}$. Now we
color the components $C_2,C_3,\dots$ independently by a uniformly random color from $[q]$
(for now we leave the component $C_1$ uncolored). Let $A_i$ be the number of vertices of color
$i$. We have by Azuma's inequality that 
\begin{equation}\label{hopso}
P\Big( \Big|A_i - \frac{n-|C_1|}{q}\Big|  \geq \frac{wn^{2/3}}{4q}\Big) \leq  2 \exp( - w^2/(32 q^2) ).
\end{equation}
Now we color $C_1$, and assume w.l.o.g. that it receives color $1$. Let $A_i'$ be the
number of vertices of color $i$ now (we have $A_1'=A_1 + |C_1|$ and $A_i' = A_i$ for $i\geq 2$). Applying union bound
to~\eqref{hopso} we obtain that with probability at
least $1-2q\exp( - w^2/(32 q^2) )$ we have
\begin{equation}\label{ipsy}
|A_1'| \geq \frac{n}{q} + w n^{2/3} \Big( \frac{7}{4} (1-1/q) - 1/(4q) \Big) \geq \frac{n}{q} + \frac{13}{12} w n^{2/3},
\end{equation}
and for all $i\in\{2,\dots,q\}$
\begin{equation}
|A_i'| \leq \frac{n}{q} - w n^{2/3} \Big( \frac{7}{4} (1/q) - 1/(4q) \Big) \leq \frac{n}{q} - \frac{13}{12} w/(2q) n^{2/3}.
\end{equation}
Note that, in the second inequality in~\eqref{ipsy}, we used the fact that $q\geq 3$.

Let $w'=(13/12)w$. Summarizing all the steps we obtain that from a state satisfying~\eqref{easum}
we get to a state satisfying
\begin{equation}\label{easum2}
|S_1|\geq \frac{n}{q} + w'n^{2/3}\quad\mbox{and}\quad
|S_i|\leq \frac{n}{q} - \frac{w'}{2q} n^{2/3}\quad\mbox{for}\ i\in\{2,\dots,q\},
\end{equation}
with probability at least
\begin{equation}
\Big(1-\frac{2K}{w}-K\exp(-cq^2 w^3)\Big)\Big(1-2q\exp(-w^2/(32q^2))\Big).
\end{equation}
For all sufficiently large $w$, the last expression is greater than $\exp(-C/w)$, where $C$ is a positive constant (depending on $K,c,q$), as wanted.
\end{proof}

\section{Lower bound on the mixing time for $B\geq \Bh$}\label{sec:lowerboundBBh}
In this section, we prove that the SW algorithm mixes in $\Omega(\log n)$ steps for all $B>\Bh$.

Recall from Section~\ref{sepp} that ${\cal B}(\vv,\delta)$ is the $\ell_\infty$-ball of configuration vectors of the $q$-state Potts model in $K_n$ around  $\vv$ of radius $\delta$, cf. equation \eqref{eq:ballball}. Let
\[S:=\mathcal{B}(\m,n^{-1/7}),\]
and denote the set of configuration vectors which are not in $S$ by $\overline{S}$. 

We first establish the following (crude) bound on the probability mass of configurations in $\overline{S}$ in the Potts distribution. (Far more precise bounds are known and can be found in, e.g., \cite{HET}; the following estimate follows easily from our upper bound on the mixing time.) 
\begin{lemma}\label{lem:measureconc}
Let $B\geq \Bh$. For the Potts distribution $\mu$ in \eqref{eq:Gibbs}, for all sufficiently large $n$, it holds that $\mu(\overline{S})\leq 1/8$.
\end{lemma}
\begin{proof}
For any starting state $X_0$, we have that for $T=O(\log n)$, it holds that
\[\Pr(X_T\in S)\geq \epsilon,\]
where $\epsilon>0$ is a constant independent of $n$. (For $B>\Bh$ this follows by Lemmas~\ref{slon7} and~\ref{slon8}, and for $B=\Bh$ this follows by Lemmas~\ref{slonslon7} and \ref{slon8}.) It follows that for all non-negative integers $j$ it holds that 
\[\Pr(X_{(j+1)T}\in S\mid X_{jT}\notin S)\geq \epsilon.\]
Further, by  Lemma~\ref{l3}, for integer $t\geq 0$,  it holds that 
\[\Pr(X_{t+1}\in S\mid X_{t}\in S)\geq 1-\exp(-\Omega(n^{1/3})).\]
We thus obtain that for some positive integer $j=j(\epsilon)$, for all sufficiently large $n$, for all  integer $t\geq j T$, it holds that 
\begin{equation}\label{eq:XtS}
\Pr(X_{t}\in S)\geq 15/16.
\end{equation}

Let $T^*=\max\{j T,\, 2\Tmix\}$. Recall that $\Tmix=O(\log n)$ (cf. Corollaries~\ref{cor:BgBh} and~\ref{cor:BBh}), so $T^*=O(\log n)$ as well. Since $\Tmix$ is the time needed to get within total variation distance $\leq 1/4$ from $\mu$, we have that for any $\epsilon'>0$, for $t\geq \Tmix \log_2 (1/\epsilon')$,  it holds that $d_{TV}(X_{t},\mu)\leq \epsilon'$ (see \cite[Section 4.5]{LPW}). Thus, we have that 
\begin{equation}\label{eq:vdvdvd}
\mu(\overline{S})-\Pr(X_{T^{*}}\in \overline{S})\leq \max_{A\subseteq \Omega}|\mu(A)-\Pr(X_{T^*}\in A)|=d_{TV}(X_{T^*},\mu)\leq 1/16.
\end{equation}
Combining \eqref{eq:XtS} and \eqref{eq:vdvdvd} yields $\mu(\overline{S})\leq 1/8$, as wanted.
\end{proof}

\begin{lemma}\label{lem:implowerbound}
For $B\geq \Bh$, there exist constants $\delta_1,\delta_2>0$ such that the following is true. Suppose that we start at a state $X_0$ such that $X_0\notin S$ and $\delta_2\leq \|\alphab(X_0)-\m\|_{\infty} \leq \delta_1$. Then for some $T=\Omega(\log n)$, with probability $\geq 1/2$, it holds that $X_T\notin S$.
\end{lemma} 
\begin{proof}[Proof of Lemma~\ref{lem:implowerbound}]
Recall that $\m=(a,b,\hdots,b)$ where $a>1/q$ is a fixpoint of $F$. Let $\delta>0$ be such that for some $0<c_l<c_u<1$ for all $z\in
[a-\delta,a+\delta]$ we have 
\begin{equation}\label{eq:plmplmplm}
c_l|z-a|\leq |F(z)-a|\leq c_u |z-a|.
\end{equation} Note that the existence of
such $\delta$ is guaranteed throughout the regime $B\geq \Bh$, since $|F'(a)|<1$ by Lemma~\ref{lem:fixpoints}, $F'(a)>0$ by Lemma~\ref{lem:Fshape} and $F'$ is continuous in a neighbourhood around $a$. Let $\delta_1,\delta_2$ be arbitrary constants satisfying $0<\delta_2<\delta_1<\delta$.

Suppose that we are at $X_t$ such that $n^{-1/7}< \|\alphab(X_{t})-\m\|_{\infty}
\leq \delta$ (note that for such $X_t$, we have $X_t\notin S$).   Let $m_1$ be the number of vertices in the largest color class and note that  $m_1/n = a + \tau =: a'$ where
$|\tau|<\delta$. Exactly as in the proof of Lemma~\ref{l3} (cf. equation \eqref{grg4c}), we obtain that  with probability $\geq 1-2q\exp(-\Theta(n^{1/3}))$  it holds that 
\begin{equation}\label{grg4b}
\Big\|\alphab(X_{t+1}) -
\Big(F(a'),\frac{1-F(a')}{q-1},\dots,\frac{1-F(a')}{q-1}\Big)\Big\|_\infty \leq n^{-1/6}.
\end{equation}
Using \eqref{eq:plmplmplm}, we have
\begin{equation}\label{grg5b}
c_l
\| \alphab(X_t) - \m\|_\infty\leq \Big\|\Big(F(a'),\frac{1-F(a')}{q-1},\dots,\frac{1-F(a')}{q-1}\Big) - \m\Big\|_\infty \leq c_u
\| \alphab(X_t) - \m\|_\infty.
\end{equation}
Equations~\eqref{grg4b} and~\eqref{grg5b} combined yield that for all sufficiently large $n$ we have the following two bounds:
\begin{equation}\label{grg6b}
\|\alphab(X_{t+1})-\m\|_\infty\geq  c_l
\| \alphab(X_t) - \m\|_\infty - n^{-1/6} \geq \frac{c_l}{2}\|\alphab(X_{t})-\m\|_\infty,
\end{equation}
\begin{equation}\label{eq:rfccfr}
\|\alphab(X_{t+1})-\m\|_\infty\leq n^{-1/6} + c_u
\| \alphab(X_t) - \m\|_\infty\leq \delta.
\end{equation}
Let $c'=-\frac{1}{8}/\log (\frac{c_l}{2})$. Applying \eqref{grg6b} for $t=0,\hdots, \lfloor c'\log n\rfloor$ (note that \eqref{eq:rfccfr} guarantees that we remain sufficiently close to $\m$ so that \eqref{grg6b} indeed applies), we obtain that with probability $1-o(1)$ it  holds that 
\[ \|\alphab(X_{c'\log n})-\m\|_\infty\geq n^{-1/8}\|\alphab(X_0)-\m\|_\infty\geq \delta_2 n^{-1/8}> n^{-1/7}.\]
This completes the proof.
\end{proof}

Using Lemma~\ref{lem:implowerbound}, we obtain the following corollary.

\begin{corollary}
Let $B\geq \Bh$. Then the mixing time $T_{mix}$ of the SW dynamics on the $n$-vertex complete graph satisfies $\Tmix=\Omega(\log n)$.
\end{corollary}
\begin{proof}
Let $\delta_1,\delta_2$ be as in Lemma~\ref{lem:implowerbound}. Consider $X_0$ such that $X_0\notin S$ and $\delta_2\leq \|\alphab(X_0)-\m\|_{\infty} \leq \delta_1$. Then, by Lemma~\ref{lem:implowerbound}, for some $T=\Omega(\log n)$ we have that 
\[\Pr\big(X_T\notin S\big)\geq 1/2.\]
On the other hand, by Lemma~\ref{lem:measureconc} we have that $\mu(\bar{S})\leq 1/8$. It follows that 
\[d_{TV}(X_T,\mu)=\max_{A\subseteq \Omega}|\mu(A)-\Pr(X_T\in A)|\geq \Pr(X_T\in \bar{S})-\mu(\bar{S})\geq 1/2-1/8>1/4.\]
It follows from the definition of mixing time that $\Tmix\geq T$, as claimed.
\end{proof}

\section{Fast mixing for $B<\Bu$}
\label{sec:fast-uniqueness}
In this section, we prove that the SW algorithm mixes in $O(1)$ steps for all $B>\Bh$. The proof for establishing mixing in the uniqueness regime will be similar to the $B>\Bh$ case.
We begin with the following analogue of Lemma~\ref{hroch1}.

\begin{lemma}\label{hroch1B}
Assume $B<\Bh$ is a constant. There exists a constant $\eps>0$ such that for, any initial state $X_0$, with probability $\Theta(1)$
the next state $X_1$ has at least $q-1$ colors that are
$\eps$-light.
\end{lemma}

\begin{proof}
The proof is analogous to that of Lemma~\ref{hroch1}, the only difference is that now we do not need to argue that there is an $\eps$-heavy color (and hence the proof is simpler). 

Let $\eps\in(0,1/10)$ be a small enough constant such that $B (1+2\eps) <q$. As in
the proof of the first part of Lemma~\ref{hroch1} with probability $q^{-q}=\Theta(1)$ all the biggest components of each color class receive the color $1$ and the sum of squares of (the sizes of) the remaining components is $o(n^2)$ with probability $\Theta(1)$. Condition on these events happening. Then, the expected number of vertices that receive a color $i=2,\hdots,q$ is at most $n/q$.  Therefore, using Azuma's inequality, with probability $\Theta(1)$, there are at most $n(1+\eps/2)/q$ vertices which have color $i=2,\hdots,q$ in $X_1$. 
By the choice of $\eps$, we have that $(1+\eps/2)/q\leq (1+\eps/2)/(B(1+2\eps))\leq  (1-\eps)/B$, and therefore there are $q-1$ colors which are $\eps$-light in $X_1$.
\end{proof}

We then have the following lemma, which is an analogue of Lemmas \ref{slon7} and \ref{slon8} in the $B>\Bh$ case,
showing that we get within distance $O(n^{-1/2})$ from the uniform phase.

\begin{lemma}\label{zirafa1}
Assume $B<\Bu$ is a constant. There exists a constant $L$ such that for any starting state $X_0$ after $T=O(1)$ steps
with probability $\Theta(1)$ the SW algorithm moves to state $X_T$ such that $\|\alphab(X_T)-\u\|_{\infty} \leq L n^{-1/2}$.
\end{lemma}

\begin{proof}[Proof of Lemma~\ref{zirafa1}.]
Let $\eps>0$ be as in Lemma~\ref{hroch1B}. By Lemma~\ref{hroch1B} starting from any $X_0$ with
constant probability we move to $X_1$ where $q-1$ colors are $\eps$-light. As in Lemma~\ref{slon7}, the evolution of the largest color class is then captured by the iterates of the function $F$.  Since $1/q$ is the only fixpoint of $F$ (by Lemma~\ref{lem:fixpoints}), we
have that for any constant $\delta>0$ there exists constant $T$ such that $F^{(T)}([0,1])\subseteq [1/q-\delta/2,1/q+\delta/2]$. Therefore, with probability
$1-o(1)$, after at most $T$ steps the size of the largest color class becomes less than $1/q + \delta$ (see the proof of Lemma~\ref{slon7} for details). In the next step even the largest color class is subcritical (by taking $\delta$ to be a small constant) and we end up, with probability $1-o(1)$, in
a state where each color occurs $(1+o(1))n/q$ times. In the next step the components sizes after
the percolation step satisfy, by Lemma~\ref{lemia1}
$$
E\Big[\sum_{i} |C_i|^2\Big] = O(n).
$$
Hence, after the coloring step, with constant probability (using the same argument as in~\eqref{grg1} and~\eqref{grg2})
we have color classes of size $(n+O(n^{1/2}))/q$.
\end{proof}

To show that the mixing time of SW is $O(1)$ when $B<\Bu$, we extend the strategy of \cite{LNNP} for $q=2$ to $q\geq 3$. In \cite{LNNP}, a certain projection of the SW chain is defined, called the magnetization chain. For us, the magnetization chain can be defined as follows. Let $\{V_1,\hdots, V_q\}$ be a fixed partition of the vertex set of the complete graph into $q$ parts. The magnetization chain is a Markov chain $\mathcal{A}_t=(A_{ij,t})_{i,j\in [q]}$ with $A_{ij,t}$ being the number of vertices in $V_i$ with color $j$ at time $t$ (the fact that the magnetization chain is a Markov chain is due to the symmetry). Note that for every  $t=0,1,\hdots$, for every $i\in [q]$ it holds that 
$\sum_{j}A_{ij,t}=|V_i|$.

The following lemma is the analogue of \cite[Proposition 7.3]{LNNP} and can be proved analogously to Lemma~\ref{lem:phases-agree}.
\begin{lemma}\label{lem:magnetization}
Assume $B<\Bh$ is a constant. Let $\{V_1,\hdots, V_q\}$ be a  partition of the vertex set of the complete graph on $n$ vertices into $q$ parts. Let $\mathcal{A}_t$ and $\mathcal{A}_t'$ be two copies of the magnetization chain. Further, denote by $a_{j,t},a_{j,t}'$ the total number of vertices with color $j$ in  $\mathcal{A}_t$ and  $\mathcal{A}_t'$, respectively, i.e.,   
\[a_{j,t}=\sum_{i\in[q]} A_{ij,t}, \quad a_{j,t}'=\sum_{i\in[q]} A_{ij,t}'.\]

Let $L>0$ be an arbitrarily large constant and suppose that at time $t$ it holds that 
\[|a_{j,t}-n/q|\leq L\sqrt{n},\quad  |a_{j,t}'-n/q|\leq L\sqrt{n}\mbox{ for all } j\in [q].\]
Then, there exists a coupling of $\mathcal{A}_{t+1},\mathcal{A}_{t+1}'$ such that with probability $\Theta(1)$, it holds that $\mathcal{A}_{t+1}=\mathcal{A}_{t+1}'$.
\end{lemma}
\begin{proof}
The proof is completely analogous to \cite[Proof of Proposition 7.3]{LNNP} and resembles the proof of Lemmas~\ref{lem:phases-uniform} and~\ref{lem:phases-agree} given earlier.  We therefore highlight the key differences. 

Perform the percolation step of the Swendsen-Wang algorithm independently for the chains $\mathcal{A}_t$ and $\mathcal{A}_t'$.
By Lemma 5.7 in \cite{LNNP}, there is a constant $c>0$ such that with probability $\Theta(1)$, in each chain, in each part $V_i$, there are $\geq c|V_i|$ isolated vertices (i.e., components of size 1).  Next, perform the coloring step in each of the two chains independently but leaving, in each chain and for each part $V_i$, these $c|V_i|$ isolated vertices uncolored. For $i,j\in[q]$, let $\hat{a}_{ij},\hat{a}'_{ij}$ be the number of vertices  which are assigned color $j$  in part $V_i$ (excluding the $c|V_i|$ isolated vertices which are not yet colored). We claim that there exists a (large) constant $L>0$ such that with probability $\Theta(1)$, for all $i,j\in [q]$, it holds that
\begin{equation}\label{eq:truncbb}
|\hat{a}_{ij}-\hat{a}'_{ij}|\leq L\sqrt{|V_i|}.
\end{equation}
Assuming this, then, just as in the proof of Lemmas~\ref{lem:phases-agree} and~\ref{lem:phases-uniform} (cf. \eqref{eq:truncated} and the coupling thereafter), we can couple the coloring of the $c|V_i|$ isolated vertices in each part $V_i$ to equalize the counts with probability $\Theta(1)$, i.e., with probability $\Theta(1)$, the coupling of the two chains satisfies $\mathcal{A}_{t+1}=\mathcal{A}_{t+1}'$.

We focus therefore on proving \eqref{eq:truncbb}. Let $\{C_k\}_{k\geq 1}$ denote the components in the first chain  after the percolation step. Then, for each $i\in [q]$, we will show that with probability $\Theta(1)$ it  holds that
\begin{equation}\label{eq:rvr4g5g3}
\sum_{k\geq 1} |C_k\cap V_i|^2=O(|V_i|).
\end{equation}
To see this, for a vertex $v$, let $C(v)$ be the component that $v$ belongs to after the percolation step. Then, note that
\[\sum_{k\geq 1} |C_k\cap V_i|^2\leq \sum_{v\in V_i}|C(v)|.\]
Since by the assumption of the lemma all colors are subcritical in the percolation step, we have that $E[|C(v)|]=O(1)$ for all $v\in V$. Using Markov's inequality, we therefore obtain  \eqref{eq:rvr4g5g3}. 

Let $\hat{n}_i$ be the number of vertices in $V_i$ excluding the isolated vertices. We obtain (using Azuma's inequality) that, with probability $\Theta(1)$,  for all $i,j\in[q]$ it holds that  
\[|\hat{a}_{ij}-\hat{n}_i/q|=O(\sqrt{|V_i|})\]
Identically, we obtain an analogous bound for $\hat{a}'_{ij}$ which yields \eqref{eq:truncbb}, as needed. 
\end{proof}

Using Lemmas~\ref{zirafa1} and~\ref{lem:magnetization}, we conclude the following corollary.
\begin{corollary}
Let $B<\Bu$ be a constant. The mixing time of the Swendsen-Wang algorithm
on the complete graph on $n$ vertices is $\Theta(1)$.
\end{corollary}
\begin{proof}
Let $\mu$ be the stationary distribution of the Swendsen-Wang algorithm (cf. \eqref{eq:Gibbs}). Consider two copies of the SW algorithm $X_t$ and $Y_t$, where $X_0$ is an arbitrary starting configuration and $Y_0$ is distributed according to $\mu$. It suffices to show that there is $T=O(1)$ and a coupling of $X_T, Y_T$ such that $X_T=Y_T$ with probability $\Omega(1)$.

We will use the magnetization chain for an appropriate partition $\{V_1,\hdots,V_q\}$ of the vertices of the complete graph. Namely, for a color $i\in [q]$, let $V_i$ be the set of vertices with color $i$ in $X_0$. Let $\mathcal{A}_t=\{A_{ij,t}\}_{i,j\in[q]}$, $\mathcal{A}_t':=\{A_{ij,t}'\}_{i,j\in[q]}$ be such that $A_{ij,t}$, $A_{ij,t}'$  is the number of vertices with color $j$ in $V_i$ in $X_t$ and $Y_t$, respectively. The key idea is that, due to symmetry, the probability that the SW chain at time $t$ is at a particular configuration $\sigma$ depends only on the counts $|V_i\cap\sigma^{-1}(j)|$ for $i\in [q]$ and $j\in [q]$. It follows that for every $t$, it holds that 
\begin{equation}\label{eq:symtricky}
d_{TV}(X_t,Y_t)=d_{TV}(\mathcal{A}_t,\mathcal{A}_t').
\end{equation}
It thus suffices to show that for $T=O(1)$, there is a coupling of $\mathcal{A}_T$ and $\mathcal{A}_T'$ such that $\mathcal{A}_T=\mathcal{A}_T'$ with probability $\Theta(1)$.

Let $L$ be the constant in Lemma~\ref{zirafa1}. By Lemma~\ref{zirafa1}, we have that for $T_1=O(1)$, with probability $\Theta(1)$ it holds that 
\begin{equation}\label{eq:2wsx2wsx}
\|\alphab(X_{T_1}) - \u\|_\infty\leq Ln^{-1/2},\quad \|\alphab(Y_{T_1}) - \u\|_\infty\leq Ln^{-1/2}.
\end{equation}
Conditioned on \eqref{eq:2wsx2wsx}, Lemma~\ref{lem:magnetization} shows that there exists a coupling of $\mathcal{A}_{T_1+1}$ and $\mathcal{A}_{T_1+1}'$ such that with probability $\Theta(1)$ it holds that $\mathcal{A}_{T_1+1}=\mathcal{A}_{T_1+1}'$. Using \eqref{eq:symtricky}, we thus conclude that the mixing time of the Swendsen-Wang algorithm is $O(1)$, as wanted. 
\end{proof}

\section{Mixing Time at $B=\Bu$}
\label{sec:fast-Bu}
For $B=\Bu$, our goal is to show that the SW chain reaches the uniform phase in $O(n^{1/3})$ steps. To do this, let $S_t$ be the size of the largest color class in state $X_t$ of the SW chain; throughout this section, we will focus on tracking $S_t$.

In Section~\ref{sec:oneiteration}, we first give some relevant statistics of $S_t$ after one iteration of the SW chain; the main lemma we will use later is Lemma~\ref{lestart}. In Sections~\ref{sec:upperboundBu} and~\ref{sec:lowerboundBu}, we use these statistics to outline our potential function argument for deriving the upper and lower bounds on the mixing time. Finally, in Section~\ref{sec:potfun}, we give in detail the construction of the potential function which is the most technical part of the proof.

\subsection{Tracking one iteration of the SW dynamics}\label{sec:oneiteration}
As a starting point, we have the following analogue of Lemma~\ref{hroch1}. 
\begin{lemma}\label{lestart2}
For sufficiently small (constant) $\epsilon>0$, for any state $X_t$ of the SW chain, with probability $\Theta(1)$, there are at least $q-1$ colors in state $X_{t+1}$ which are $\epsilon$-light. Further, if state $X_t$ has $q-1$ colors which are $\epsilon$-light, then with probability $1-\exp(-n^{\Omega(1)})$, the same is true for $X_{t+1}$.
\end{lemma}

\begin{proof}
For a color $i\in[q]$, we will write $\alpha_i$ as a shorthand for $\alpha_i(X_t)$, and denote $m_i=n\alpha_i$. In each step of the Swendsen-Wang algorithm, the percolation step for color $i$ picks a graph $G_i$ from $G(m_i,B\alpha_i/m_i)$. Let $C_1^{(i)},C_2^{(i)}, \hdots$ be the components of $G_i$ in decreasing order of  size. 

The beginning of the proof is analogous to the beginning of the proof of Lemma~\ref{lem:wgood}. Let $A$ be the constant in Lemma \ref{lemia3}. For each color $i\in[q]$ the following hold with positive probability (not depending on $n$):
\begin{enumerate}
\item  \label{it:pomp} If $B\alpha_i\geq (1-A m_i^{-1/3})/m_i$, then  $\sum_{j> 1} |C^{(i)}_j|^2 \leq m_i^{4/3}\leq n^{4/3}$ (by Lemma~\ref{lemia3}).
\item \label{it:qomp} If $(1-A m_i^{-1/3})/m_i>B\alpha_i$, then $\sum_{j\geq 1} |C^{(i)}_j|^2 \leq  n^{4/3}$ (by Item~\ref{it:sub} of Lemma~\ref{lemia1}).
\end{enumerate}
 Let $S=\{i\in[q]: B\alpha_i\geq 1\}$  (note that the set $S$ may be empty). Consider all the components \emph{different} from $C^{(i)}_1$, $i\in S$.  Color these
components  independently by a uniformly random color from $[q]$. For $i\in [q]$, let $A_i'$ be the number of vertices of color
$i$. Let $w>0$ be a constant such that $1>2q \exp( -w^2/2)$. By Azuma's inequality and a union bound we have that with probability at least $1-2q \exp( -w^2/2)>0$, for each $i\in[q]$ it holds that
\begin{equation*}
\Big|A_i' -\frac{n-\sum_{i\in S}|C^{(i)}_1|}{q}\Big|\leq w n^{2/3}.
\end{equation*}
For $i\in[q]$, let $A_i$ be the
number of vertices of color $i$ after the coloring step of the SW algorithm. With probability at least $q^{-q}$ each of $C^{(i)}_1$ with $i\in S$ receives color 1.  Note, we have $A_1=A_1' + \sum_{i\in S}|C^{(i)}_1|$ and $A_i = A_i'$ for $i\geq 2$. We obtain that with probability at
least $q^{-q}\big(1-2q\exp( -w^2/2)\big)>0$, for all $i\geq 2$,
\begin{equation}\label{eq:boundprime}
|A_i| \leq \frac{n}{q} -\frac{1}{q}\left(\sum_{i\in S}|C^{(i)}_1|\right)+w n^{2/3}\leq \frac{n}{q} + wn^{2/3}.
\end{equation}
Since $\Bu<q$, we have that for sufficiently small constant $\epsilon>0$, for all $n$ sufficiently large, it holds that $|A_i|\leq (1-\epsilon)n/B$ for all $i\neq 1$, and thus the colors $2,\hdots,q$ are $\epsilon$-light with probability $\Theta(1)$ as wanted.

For the second part of the lemma where we know that in $X_t$ there are $q-1$ colors which are $\epsilon$-light, the proof is analogous. The difference is that now we need upper bounds for  the sum of squares of the components (other than the largest component --- there can be at most one of those by the assumption) which hold with probability $1-\exp(-n^{\Omega(1)})$. Note, for a color class $i$, we have the (crude) bounds 
\begin{equation}\label{eq:countsquares1}
\sum_{j\geq 1}|C_j^{(i)}|^2\leq n\, |C_1^{(i)}|\mbox{ and }\sum_{j\geq 2}|C_j^{(i)}|^2\leq n\, |C_2^{(i)}|.
\end{equation}
For each of the $(q-1)$ $\epsilon$-light colors, the first inequality in \eqref{eq:countsquares1}  together with Lemma~\ref{l1} bounds the sum of squares of the components  by $n^{7/4}$ with probability $1-\exp(-\Theta(n^{3/4}))$. For the remaining color class (i.e., the one that we do not have an upper bound on its density by the assumption), to bound the sum of squares of the components we obtain the same bound $n^{7/4}$ with probability $1-\exp(-n^{\Omega(1)})$ by considering cases. If the color class is supercritical we use Lemma~\ref{l1al1} and the second inequality in \eqref{eq:countsquares1}. If the color class is in the critical window we use Lemma~\ref{lem:criticalper} and the first inequality in \eqref{eq:countsquares1}. If the color class is subcritical we use Lemma~\ref{l1} and the first inequality in \eqref{eq:countsquares1}. The only modification needed in the argument is to replace $wn^{2/3}$ in \eqref{eq:boundprime} by $n^{9/10}$ and the remaining part holds verbatim.
\end{proof}

The key part of our arguments is to track the evolution of the size $S_t$ of the largest colors when there are $q-1$ colors which are $\epsilon$-light.

We first do this in the easier case when $S_t$ has density close to $1/B$ (in the complementary regime, we will need more statistics of $S_t$). In this regime, the following lemma roughly says that a step of the SW dynamics makes the density of the largest color class roughly $1/q$. (Intuitively, this follows by a ``continuity" argument since $F(1/B)=1/q$.)

\begin{lemma}\label{lem:probdeviation}
Let $\epsilon>0$ be a sufficiently small constant. Suppose that $X_t$ is such that $q-1$ colors are $\epsilon$-light  and  that $S_t<(1+\epsilon)n/B$. Then with probability $1-\exp(-n^{\Omega(1)})$ it holds that $S_{t+1}<(1+3q\epsilon)n/q$.
\end{lemma}
\begin{proof}[Proof of Lemma~\ref{lem:probdeviation}]
The proof is analogous to the proof of Lemma~\ref{lestart2} and as such we follow the notation in there. The only difference is that now we have to account slightly more accurately for the size of the largest color class in $X_{t+1}$.  

Assume that the $q-1$ $\epsilon$-light colors in $X_t$ are $2\hdots,q$ and assume w.l.o.g. that (the perhaps linear sized) $C_1^{(1)}$ gets colored with color 1 (in state $X_{t+1}$). The color classes of $2\hdots,q$ in $X_t$ are subcritical and thus fall into Item~\ref{it:qomp} of the analysis in the proof of Lemma~\ref{lestart2}. For the remaining color class 1 in $X_t$, it may fall either into Item~\ref{it:pomp} or~\ref{it:qomp}. 

It follows that the bounds for $A_i'$ in \eqref{eq:boundprime} still hold  and in particular the colors $2,\hdots,q$ have size at most $(1/q)n+o(n)$ (since they did not receive a giant component).

For the color class 1 in $X_{t+1}$, note that $A_1=|C_1^{(1)}|+A'_i$. For all sufficiently small (constant) $\epsilon>0$, the largest component $C_1^{(1)}$, with probability $1-\exp(-n^{\Omega(1)})$, has size at most $3\epsilon (n/B)$ (by  Item~\ref{it:suplargest} of Lemma~\ref{lemia1}). Note that $\Bu\geq 1$ for all $q\geq 3$ (follows, e.g., by definition \eqref{eq:BpBh}) and hence $3\epsilon (n/B)\leq 3\epsilon n$. It follows that for all sufficiently large $n$, $A_1$ is at most $(1+3q \epsilon)n/q$, as wanted.
\end{proof}

The following lemma gives some statistics of $S_t/n$ throughout the range $(1/B,1]$, i.e., when the largest color class is supercritical in the percolation step of the SW dynamics. Recall the function $F$ defined in \eqref{defa},\eqref{dexa}.
\begin{lemma}\label{lestart}
Let $\epsilon>0$ be an arbitrarily small constant and condition on the event that $X_t$ has $q-1$ colors which are $\epsilon$-light.  Assume that $\zeta$ satisfies  $(1+\epsilon)/B\leq \zeta/n\leq 1$. Let $W:=E[S_{t+1}\, |\, S_t=\zeta]$. 

For all constant $\epsilon'>0$, for all sufficiently large $n$, it holds that
\begin{equation}\label{eko1}
n F(\zeta/n)-n^{\epsilon'}\leq W\leq n F(\zeta/n) + n^{\epsilon'}.
\end{equation} Also, there exist absolute constants $Q_1,Q_2$ (depending only on $\epsilon$) such that
\begin{equation}\label{eko2}
n Q_2 \leq Var[S_{t+1}\, |\, S_t=\zeta] \leq n Q_1,
\end{equation}
Finally, for every integer $k\geq 3$  and constant $\epsilon'>0$, there exists a constant  $c>0$ such that 
\begin{equation}\label{eko3}
E\Big[ \big|S_{t+1}-W \big|^k  |\, S_t=\zeta \Big] \leq  cn^{k/2+\epsilon'}.
\end{equation}
\end{lemma}

\begin{proof}
To avoid overloading notation, we assume throughout that we condition on $S_t=\zeta$. 

We will write $\alpha_i$ as a shorthand for $\alpha_i(X_t)$, and denote $m_i=n\alpha_i$. W.l.o.g. we will assume that the color class with largest size is the one corresponding to color 1, so that $\alpha_1= \zeta/n\geq (1+\epsilon)/B$. Since the remaining $(q-1)$ colors are $\epsilon$-light, for each $i\in\{2,\hdots,q\}$ we have  $\alpha_i\leq (1-\epsilon)/B$. 

In each step of the Swendsen-Wang algorithm, the percolation step for color $i$ picks a graph $G_i$ from $G(m_i,B\alpha_i/m_i)$. Let $C_1^{(i)},C_2^{(i)}, \hdots$ be the components of $G_i$ in decreasing order of  size. Note that $G_1$ is in the supercritical regime, while $G_2,\hdots,G_q$ are in the subcritical regime. By Lemma~\ref{lem:supercritical}, for every constant $\epsilon'>0$ we have  that
\begin{equation}\label{eq:c1xxp}
E\big[\big|C_1^{(1)}\big|\big]=\beta \zeta\pm \zeta^{\epsilon'}=\beta \zeta\pm n^{\epsilon'},
\end{equation}
where $\beta\in(0,1)$ satisfies $\beta+\exp(-\beta \frac{B \zeta}{n})=1$. Note that
\[n F(\zeta/n)=\frac{n}{q}+\Big(1-\frac{1}{q}\Big)\beta\zeta.\]

Let $A_i$ be the number of vertices with color $i$ in $X_{t+1}$ and w.l.o.g. assume that $C_1^{(1)}$ receives the color 1 in the coloring step of the SW dynamics. We will show that with  probability $1-\exp(-n^{\Omega(1)})$ it holds that $S_{t+1}=A_1$, so the estimates on the moments of $S_{t+1}$ will follow from those of $A_1$. 

More precisely, with a scope to also prove \eqref{eko3}, we will show that for every sufficiently small constant $\epsilon'>0$ it holds with probability $1-\exp(\Theta(n^{-\epsilon'}))$  that 
\begin{equation}\label{A1bound}
|A_1-n F(\zeta/n)|\leq 2n^{1/2+\epsilon'} \mbox{ and } A_i\leq (n-\beta \zeta)/q+n^{1/2+\epsilon'} \mbox{ for }i\in \{2,\hdots,q\}.
\end{equation}
Since $\beta \zeta=\Omega(n)$, we will then obtain that $A_1>A_i$ for all $i\neq 1$.

From Lemma~\ref{l1al1} equation \eqref{st1aa} (applied to color $i=1$) and Lemma~\ref{l1} (applied to colors $i=2,\hdots,q$), with probability  $1-q\exp(-\Theta(n^{\epsilon'}))$, we have 
\begin{equation}\label{eq:sizebound1}
|C_j^{(1)}|\leq n^{\epsilon'} \mbox{ for } j\geq 2,\qquad |C_j^{(i)}|\leq n^{\epsilon'} \mbox{ for } i\in \{2,\hdots,q\},\, j\geq 1.
\end{equation}
From Lemma~\ref{l1al1} equation \eqref{st2bb}, with probability $1-\exp(-\Theta(n^{\epsilon'}))$, we also have 
\begin{equation}\label{eq:sizebound2}
|C_1^{(1)}-\beta \zeta|\leq n^{1/2+\epsilon'}.
\end{equation}
(Note that $\beta \zeta=\Omega(n)$.) Condition on the event that the bounds in \eqref{eq:sizebound1} and \eqref{eq:sizebound2} hold. From \eqref{eq:sizebound1}, we have the crude bound 
\begin{equation}\label{eq:azumavariancebound}
\sum_{j\geq 2}(|C_j^{(1)}|)^2+\sum_{q\geq i\geq 2}\sum_{j\geq 1}(|C_j^{(i)}|)^2\leq n^{1+\epsilon'}.
\end{equation}
Consider now the coloring step of the SW algorithm and, in particular, color independently all the components \emph{different} from $C^{(1)}_1$ by a uniformly random color from $[q]$. Let $A_i'$ be the number of vertices of color $i$ in this process. Note that $A_1=|C_1^{(1)}|+A'_1$ and $A_i=A'_i$ for $i=2,\hdots,q$. Using \eqref{eq:azumavariancebound}, by Azuma's inequality we have that with probability $1-2q\exp(-n^{\epsilon'})$ for all $i\in [q]$ it holds that
\begin{equation}\label{eq:fluctuation}
\Big|A_i'-\frac{n-|C_1^{(1)}|}{q}\Big|\leq n^{1/2+\epsilon'}.
\end{equation}
From \eqref{eq:sizebound2} and \eqref{eq:fluctuation} we obtain that for all sufficiently large $n$, with probability $1-\exp(-\Theta(n^{\epsilon'}))$ it holds that $S_{t+1}=A_1$. It follows that $E[S_{t+1}]=E[A_1]+o(1)$ and $E[|S_{t+1}-E[S_{t+1}]|^k]=E[|A_1-E[A_1]|^k]+o(1)$ for all integer $k\geq 2$. Thus, the bounds in \eqref{eko1}, \eqref{eko2}, \eqref{eko3} will follow from
\begin{gather}
E[A_1]= n F(\zeta/n)\pm n^{\epsilon'},\label{eq:unconditionalexpec}\\
Q_1 n\leq Var[A_1]\leq Q_2 n,\label{eq:unconditionalvar}\\
E\big[|A_1-E[A_1]|^k\big]\leq Kn^{k/2+\epsilon'},\label{eq:unconditionalmoments}
\end{gather}
where $k\geq 3$ is an integer, $\epsilon'>0$ is an arbitrarily  small constant,  $Q_1,Q_2>0$ are absolute constants and $K$ is a constant depending on $k$. 

 We start by proving \eqref{eq:unconditionalexpec} and \eqref{eq:unconditionalvar} where we need more precise bounds. By the second inequality in \eqref{eq:variancesuper} of  Lemma~\ref{lem:supercritical} (applied to color 1) and part~\ref{it:sub} of Lemma~\ref{lemia1} (applied to colors $i=2,\hdots,q$), we have for some constants $K_1,K_2,K_3>0$ that 
\begin{equation}\label{eq:sumsquarescom}
K_1 n\leq Var\big[|C_1^{(1)}|\big]\leq K_2n, \quad E\Big[\sum_{j\geq 2}(|C_j^{(1)}|)^2+\sum_{q\geq i\geq 2}\sum_{j\geq 1}(|C_j^{(i)}|)^2\Big]\leq  K_3 n.
\end{equation}

Denote by $\mathcal{C}$ the random vector $\{\big|C_j^{(i)}\big|\}_{i\in[q],j\geq1}$. We first estimate the moments of $A_1$ conditioned on $\mathcal{C}$. We have
\begin{gather}
E\big[A_1\, \big|\, \mathcal{C}\big]=\big|C_1^{(1)}\big|+\frac{n-\big|C_1^{(1)}\big|}{q}=\frac{n}{q}+\Big(1-\frac{1}{q}\Big)\big|C_1^{(1)}\big|,\label{eq:c1xxpb}\\
Var\big[A_1\, \big|\, \mathcal{C}\big]=\frac{1}{q}\Big(1-\frac{1}{q}\Big)\bigg[\sum_{j\geq 2}(|C_j^{(1)}|)^2+\sum_{q\geq i\geq 2}\sum_{j\geq 2}(|C_j^{(i)}|)^2\bigg].
\end{gather}
It follows from \eqref{eq:c1xxpb} that $E[A_1]=\frac{n}{q}+(1-\frac{1}{q})E\big[\big|C_1^{(1)}\big|\big]$, so \eqref{eq:unconditionalexpec} follows from \eqref{eq:c1xxp}. Also, by the law of total variance we have  $Var[A_1]=Var[E\big[A_1\, \big|\, \mathcal{C}\big]]+E[Var\big[A_1\, \big|\, \mathcal{C}\big]]$, so from \eqref{eq:c1xxp},\eqref{eq:sumsquarescom},\eqref{eq:c1xxpb}, we obtain \eqref{eq:unconditionalvar}.

Finally, it remains to prove \eqref{eq:unconditionalmoments}. Let $\epsilon'':=\epsilon'/k>0$. By the triangle inequality and \eqref{eq:unconditionalexpec} (applied for the constant $\epsilon''$), we have that 
\[\big|A_1-E[A_1]\big|\leq \big|A_1-nF(\zeta/n)\big|+\big|E[A_1]-nF(\zeta/n)\big|\leq |A_1-nF(\zeta/n)|+n^{\epsilon''},\]
and hence by the AM-GM inequality we have 
\[\big|A_1-E[A_1]\big|^k\leq 2^{k-1}(|A_1-nF(\zeta/n)|^{k}+n^{\epsilon'}).\] 
By integrating the first inequality in \eqref{A1bound} (applied for the constant $\epsilon''$), we obtain that $E\big[|A_1-nF(\zeta/n)|^{k}\big]\leq 2^k n^{k/2+\epsilon'}+o(1)$. Combining these bounds yields \eqref{eq:unconditionalmoments} with $K=2^{3k}$ (to absorb the lower order terms $n^{\epsilon'}$ and $o(1)$).

This concludes the proof of Lemma~\ref{lestart}.
\end{proof}

\subsection{Upper bound on the mixing time at $B=\Bu$}\label{sec:upperboundBu}
In this section, we prove that the mixing time of the SW chain satisfies $\Tmix=O(n^{1/3})$ at the critical point $B=\Bu$. 

The most difficult part of our arguments is to argue that the SW chain escapes the vicinity of the majority phase in $O(n^{1/3})$ steps, i.e., when the size $S_t$ of the largest color class is in the window  $|S_t-na|\leq \delta n^{2/3}$ for some small constant $\delta>0$ (recall that $a$ is the marginal of the majority phase and satisfies $F(a)=a$, see also Lemma~\ref{lem:fixpoints}). Note that from \eqref{eko1} we have that  $E[S_{t+1}\mid S_t]\approx n F(S_t/n)$ and hence the drift of the process inside the window  is very weak;  for example,  when $S_t/n=a$, the expected value of $S_{t+1}/n$ remains very close to $a$. More generally, an expansion of $F$ around the point $a$ yields that  $F(z)\approx z-c(z-a)^2$ for all $z\in (a-\epsilon,a+\epsilon)$ for some constants $c,\epsilon>0$. Therefore, we obtain that $E[S_{t+1}\mid S_t]\approx S_t-c(S_t-an)^2/n$ for some  constant $c>0$, so the change (in expectation) of $S_{t+1}$  relative to $S_t$ is bounded above by roughly $\delta^2 n^{1/3}$. In particular, how does the process escape the window $|S_t-na|\leq \delta n^{2/3}$ in $O(n^{1/3})$ steps?

The rough intuition is that inside the window the variance of the process aggregates the right way and the process gets displaced (with constant probability) by the square root of the ``aggregate variance''. That is, after $\Omega(n^{1/3})$ steps, $S_t$ is displaced  by roughly $\Omega(\sqrt{n^{1/3} n})=\Omega(n^{2/3})$ from $na$. In the meantime, it holds that $F(z)\leq z$ for all $z\in [1/B,1]$ so $S_t$ is bound to escape from the lower end of the window. Once $S_t$ escapes the window, the drift $F(z)-z$ coming from the expectation of $S_t/n$ takes over and the trajectory of $S_t/n$ is close to a deterministic process $z(t)$ which satisfies the differential equation $dz=(F(z)-z)dt$.  Since $F(z)\approx z-c(z-a)^2$ for all $z\in(a-\epsilon,a+\epsilon)$, we obtain that the number of steps needed so that $S_t/n$ goes from $a-n^{-1/3}$ to $a- \epsilon$ is roughly $\int^{a- \epsilon}_{a-n^{-1/3}}\frac{1}{F(z)-z}dz\approx n^{1/3}$; from that point on,  the SW chain will get within constant distance from the uniform phase in $O(1)$ steps.\footnote{Heuristically, the exponent $1/3$ in our target mixing time bound $O(n^{1/3})$ is the value $\rho\geq 0$ obtained by balancing (i) the number of steps that the process needs to get out  from  the interval $(na-n^{1-\rho},na+n^{1-\rho})$ using its variance which we expect to happen in roughly $n^{1-2\rho}$ steps  (since $\sqrt{n^{1-2\rho}n}=n^{1-\rho}$), and (ii) the number of steps that the process needs to cross the intervals $(n(a-\epsilon),na-n^{1-\rho})$ and $(na+n^{1-\rho},n(a+\epsilon))$   using the drift $z-F(z)\approx c(z-a)^2$ (which requires roughly $n^\rho$ steps).} Rather than formalizing explicitly this intuition, we will capture the progress of the chain towards the uniform phase by a potential function argument.

The potential function is designed so that its maximum value is at most $O(n^{1/3})$ and, at each step of the SW chain, the expected decrease of the potential function is at least a constant. More precisely, we show the following lemma in Section~\ref{sec:potfun}.

\begin{lemma}\label{lem:simplifiedpotential}
Let $B=\Bu$. There exist constants $M_1,M_2,\tau>0$ such that for all sufficiently small  $\epsilon>0$, for all sufficiently large $n$ the following holds. There exists an increasing three-times differentiable potential function $G:[1/q,1]\rightarrow [0,M_1n^{1/3}]$ with $G(1/q)=0$ and $\max_{z\in [1/q,1/B]}G'(z)\leq M_2$ such that  for any $\zeta\geq (1+\epsilon) n/B$, if $X_t$ has $(q-1)$ colors which are $\epsilon$-light, then it holds that
\begin{equation}\label{eq:potentialcopy}
E[G(S_{t+1}/n)\,|\, S_t=\zeta] \leq G(\zeta/n) - \tau.
\end{equation}
\end{lemma}
The proof of Lemma~\ref{lem:simplifiedpotential} is quite technical, so let us briefly discuss the main ideas underlying the proof. The crucial ingredient is to specify the potential function $G$ so that \eqref{eq:potentialcopy} is satisfied. To motivate the choice of $G$, by taking expectations in the second order Taylor expansion of $G(S_{t+1}/n)$ around $E[S_{t+1}/n\,|\, S_t=\zeta]\approx F(\zeta/n)$ we obtain
\begin{equation}\label{eq:motivation}
E[G(S_{t+1}/n)\,|\, S_t=\zeta]\approx G(F(\zeta/n))+\frac{1}{2}Var[S_{t+1}/n\,|\, S_t=\zeta]\, G''(F(\zeta/n)).
\end{equation} 
(The precise conditions on the derivatives of $G$ such that the approximation in \eqref{eq:motivation} is sufficiently accurate are given in Lemma~\ref{lma2}.) From \eqref{eq:motivation}, in order to satisfy \eqref{eq:potentialcopy}, the function $G$ has to be carefully chosen to control the interplay between $G(F(x))-G(x)$ and $G''(F(x))$. The first derivative of $G$ should correspond to the drift $F(x)-x$ of the process coming from its expectation while the second derivative of $G$ to the variance of the process. More precisely, when $x$ is outside the critical window, the choice of the potential function is such that $G(F(x))-G(x)$ is bounded above by a negative constant (i.e., its derivative is $1/(x-F(x))$); by our earlier remarks this should be sufficient to establish progress outside the critical window. Indeed, with this choice it turns out that $|G''(x)|/n$ is bounded above by a small constant outside the critical window, so that \eqref{eq:potentialcopy} is satisfied. Inside the critical window, where $x\approx F(x)$ and hence $G(F(x))-G(x)\approx 0$,  we choose $G$ so that  $G''(x)$ is negative. More precisely, to satisfy \eqref{eq:potentialcopy},  since  $Var[S_{t+1}/n\,|\, S_t=\zeta]=\Theta(1/n)$ from Lemma~\ref{lestart}, we set $G''(x)=-Cn$ for some constant $C>0$. The remaining part  is then to interpolate between these two regimes keeping $G'(x)/G''(x)$ sufficiently large (so that \eqref{eq:potentialcopy} is satisfied) and $G(x)$ small (i.e., $O(n^{1/3})$); this is possible due to the quadratic behaviour of $F(z)-z$ around $z=a$. (See Lemma~\ref{lem:potentialchoice} and its proof for the explicit specification of $G$.)

We next combine Lemmas~\ref{lestart2},~\ref{lem:probdeviation} and~\ref{lem:simplifiedpotential} to show the following.
\begin{lemma}\label{combinecombine}
For $B=\Bu$, there exists $L>0$ such that the following is true. In $T=O(n^{1/3})$ steps, for any starting state $X_0$, with probability $\Theta(1)$
the SW algorithm ends up in a state $X_T$ such that $\|\alphab(X_T)-\u\|_{\infty}\leq L n^{-1/2}$.
\end{lemma}
\begin{proof}
Let $T:=\left\lceil 3M_1n^{1/3}/\tau\right\rceil$, where $M_1,\tau$ are the constants in Lemma~\ref{lem:simplifiedpotential}.

Let $\epsilon>0$ be a sufficiently small constant, to be picked later.  We will assume that the state $X_1$ has $q-1$ colors which are $\epsilon$-light since (by the first part of  Lemma~\ref{lestart2}) this event happens with probability $\Theta(1)$. Henceforth, we will condition on this event. 

Recall that $S_t$ is the size of the largest color class at time $t$.  We will show that with probability $\Theta(1)$ it holds that $S_T<(1+\epsilon) n/B$. Assuming this for the moment, then in the next step, i.e., at time $T+1$, by Lemma~\ref{lem:probdeviation} all color classes have size at most $(1+3 q \epsilon)n/q$ and (for all sufficiently small $\epsilon$) are thus subcritical in the percolation step of the SW dynamics. It follows that  the components' sizes after the percolation step satisfy, by Item~\ref{it:sub} in Lemma~\ref{lemia1}, $E\Big[\sum_{i}\, |C_i|^2\Big] = O(n)$. Hence, after the coloring step, using Azuma's inequality with constant probability 
we have color classes of size $(n+O(n^{1/2}))/q$ (see for example the derivation of \eqref{grg1} and \eqref{grg2} for details). 

It remains to argue that with probability $\Theta(1)$ it holds that $S_T<(1+\epsilon) n/B$.  Let $P_t$ be the probability that at time $t$ it holds that $S_t<(1+\epsilon) n/B$. We will show that $P_T\geq 1/10$. We will use Lemma~\ref{lem:simplifiedpotential} and the potential function $G$ therein to bound $P_T$. In particular, we will show that for all $n$ sufficiently large, for all $t=1,\hdots,T$, it holds that
\begin{equation}\label{eq:modifiedpotential234}
E[G(S_{t+1}/n)] \leq E[G(S_t/n)] - \tau(1-P_t)+\tau/2,
\end{equation}
where $\tau$ is the constant in Lemma~\ref{lem:simplifiedpotential}. Prior to that, let us conclude that $P_T\geq 1/10$ assuming~\eqref{eq:modifiedpotential234}. Note that if $S_t<(1+\epsilon) n/B$ then  $S_{t+1}<(1+\epsilon) n/B$ with probability at least $1-\exp(-n^{\Omega(1)})$ (by Lemma~\ref{lem:probdeviation}), so $P_{t}\leq P_{t+1}+O(1/n)$. Since $T=O(n^{1/3})$, we have $P_t\leq P_T+O(n^{-2/3})$ for all $t=1,\hdots,T$ and hence $\sum^{T}_{t=1}P_t\leq TP_T+o(1)$. By applying \eqref{eq:modifiedpotential234} recursively, it hence follows that
\[E[G(S_{T+1}/n)]\leq E[G(S_1/n)]-\tau T(1/2-P_T)+o(1).\]
Using that $0\leq G(z)\leq M_1n^{1/3}$ for all $z\in [1/q,1]$, we obtain that $P_T\geq 1/2 -M_1n^{1/3}/(\tau T)+o(1)$. For $T=\left\lceil 3M_1 n^{1/3}/\tau\right\rceil$ we thus have $P_T\geq 1/10$ as wanted.

Finally, we prove \eqref{eq:modifiedpotential234} for $t=1,\hdots,T$. Note that Lemmas~\ref{lem:simplifiedpotential} and~\ref{lem:probdeviation} apply whenever $X_t$ has $q-1$ $\eps$-light colors, so  we will need to account for the (small-probability) event that this fails. Namely, let $\mathcal{E}_t$ denote the event that $X_t$ has $q-1$ $\eps$-light colors. Since we condition on the event that $\mathcal{E}_1$ holds, we have that $\bigcap^T_{t=2}\mathcal{E}_t$  holds with probability  at least $1-\exp(-n^{\Omega(1)})$ (by the second part of Lemma~\ref{lestart2}).

Let $\mathcal{F}_t$ be the event that $S_t<(1+\epsilon)n/B$ and note that $P_t=\Pr(\mathcal{F}_t)$. By taking expectations in inequality \eqref{eq:potentialcopy} of Lemma~\ref{lem:simplifiedpotential}, we have 
\begin{equation}\label{Gtau1}
E\big[G(S_{t+1}/n)\mid \mathcal{E}_t, \overline{\mathcal{F}_t}\,\big]\leq E\big[G(S_t/n)\mid \mathcal{E}_t, \overline{\mathcal{F}_t}\,\big]-\tau.
\end{equation}
Note that if $S_{t}<(1+\epsilon) n/B$, then by Lemma~\ref{lem:probdeviation}, with probability $1-\exp(-n^{\Omega(1)})$ we have $S_{t+1}<(1+3q\epsilon)n/q$. From Lemma~\ref{lem:simplifiedpotential}, we have $G(1/q)=0$ and $\max_{z\in [1/q,1/B]}G'(z)\leq M_2$ where $M_2$ is an absolute constant independent of $n$. It follows that for all sufficiently small constant $\epsilon>0$, when $S_{t+1}<(1+3q\epsilon)n/q$, it holds that $G(S_{t+1}/n)\leq \tau/3$. It follows that 
\begin{equation}\label{Gtau2}
E\big[G(S_{t+1}/n)\mid \mathcal{E}_t, \mathcal{F}_t\big]\leq \tau/3.
\end{equation}
Note that $G$ is positive throughout the interval $[1/q,1]$ since $G(1/q)=0$ and $G$ is increasing. By the positivity of $G$, we thus obtain the crude inequality 
\begin{equation}\label{trfgtrfg}
\Pr\big(\overline{\mathcal{F}_t}\mid \mathcal{E}_t\big)\, E\big[G(S_t/n)\mid \mathcal{E}_t, \overline{\mathcal{F}_t}\,\big]\leq E[G(S_t/n)\mid \mathcal{E}_t].
\end{equation}
Let $P_t'$ be the probability that at time $t$ it holds that $S_t<(1+\epsilon) n/B$ conditioned on the event $\mathcal{E}_{t}$, i.e., $P_t':=\Pr(\mathcal{F}_t\mid \mathcal{E}_t)$. Note  that $P_t\geq P_t'(1-\exp(-n^{\Omega(1)}))\geq P_t'-\exp(-n^{\Omega(1)})$. Combining \eqref{Gtau1}, \eqref{Gtau2} and \eqref{trfgtrfg}, we obtain 
\begin{equation}\label{eq:modifiedpotential2}
E[G(S_{t+1}/n)\mid \mathcal{E}_t] \leq E[G(S_t/n) \mid \mathcal{E}_t] - \tau(1-P_t')+\tau/3.
\end{equation}
Since $G$ is bounded by a polynomial in $n$ and the probability of the event $\overline{\mathcal{E}_t}$ is $\exp(-n^{\Omega(1)})$, removing the conditioning in \eqref{eq:modifiedpotential2} only affects the inequality by an additive $o(1)$. Similarly, replacing $P_t'$ with $P_t$ in \eqref{eq:modifiedpotential2} only affects the inequality by an additive $o(1)$. This proves that \eqref{eq:modifiedpotential234} holds for all sufficiently large $n$, thus concluding the proof of Lemma~\ref{combinecombine}.
\end{proof}

Using Lemma~\ref{combinecombine}, it is not hard to obtain the following corollary.
\begin{corollary}\label{cor:mixingtimeBu}
Let $B=\Bu$. The mixing time of the Swendsen-Wang algorithm
on the complete graph on $n$ vertices is $O(n^{1/3})$.
\end{corollary}
\begin{proof}
Consider two copies $(X_t),(Y_t)$ of the SW chain. As in the proof of Corollary~\ref{cor:BgBh}, it suffices to  show that for $T=O(n^{1/3})$, there exists a coupling of $(X_t)$ and $(Y_t)$ such that $\Pr(X_T=Y_T)=\Omega(1)$.  

By Lemma~\ref{combinecombine}, for  $T_1=O(n^{1/3})$,  it holds that with probability $\Theta(1)$
\begin{equation}\label{eq:couplecouple2}
\|\alphab(X_{T_1}) - \u\|_\infty\leq Ln^{-1/2}\mbox{ and }\|\alphab(Y_{T_1}) - \u\|_\infty\leq Ln^{-1/2}.
\end{equation}
Conditioning on \eqref{eq:couplecouple2}, by Lemma~\ref{lem:phases-uniform}, there exists a coupling such  that with probability $\Theta(1)$ for $T_2=T_1+1$, it holds that $\alphab(X_{T_2})=\alphab(Y_{T_2})$. Conditioning on $\alphab(X_{T_2})=\alphab(Y_{T_2})$, by Lemma~\ref{lem:vertices-agree} there exists $T_3=O(\log n)$ and a coupling such that $\Pr(X_{T_2+T_3}=  Y_{T_2+T_3}\mid \alphab(X_{T_2})=\alphab(Y_{T_2}))=\Omega(1)$. It is now immediate to combine the couplings to obtain a coupling such that $\Pr(X_T=Y_T)=\Omega(1)$ with $T=T_2+T_3=O(n^{1/3})$, as desired.
\end{proof}

\subsection{Lower bound on the mixing time at $B=\Bu$}\label{sec:lowerboundBu}

In this section, we prove that the mixing time of the SW algorithm at $B=\Bu$ satisfies $\Tmix=\Omega(n^{1/3})$. 

As in the proof of the upper bound, the lower bound on the mixing time follows by carefully accounting for the number of steps that the SW algorithm needs to escape the window around the majority phase. In this section, our goal  is to show that it takes $\Omega(n^{1/3})$ steps to escape the window. The following lemma provides the ``reverse" direction of Lemma~\ref{lem:simplifiedpotential}. Recall that for a state $X_t$ of the SW algorithm, the size of the largest color class is denoted by $S_t$.
\begin{lemma}\label{lem:simplifiedpotentialb}
Let $B=\Bu$. There exist constants $M_1,M_2,\rho>0$ such that for all sufficiently small $\epsilon>0$, for all sufficiently large $n$ the following holds. There exists a three-times differentiable increasing function $G:[1/q,1]\rightarrow [0,M_1n^{1/3}]$ which satisfies $G(1/B)=O(1)$, $G(1)\geq M_2 n^{1/3}$ such that  for any $\zeta\geq n/q$, if $X_t$ has $(q-1)$ colors which are $\epsilon$-light, then it holds that
\begin{equation}\label{eq:potentialcopyb}
E[G(S_{t+1}/n)\,|\, S_t=\zeta] \geq G(\zeta/n) - \rho.
\end{equation}
\end{lemma}
We remark here that the potential function in Lemmas~\ref{lem:simplifiedpotential} and~\ref{lem:simplifiedpotentialb} will be chosen to be identical. We thus refer the reader to the discussion after Lemma~\ref{lem:simplifiedpotential} for an overview of the construction of $G$ and to Section~\ref{sec:potfun} for the actual construction and the proof of Lemma~\ref{lem:simplifiedpotentialb}.

Analogously to Section~\ref{sec:lowerboundBBh}, we will also need a (crude) bound on the probability mass of configurations which are far from the uniform phase in the Potts distribution. Recall from Section~\ref{sepp} that ${\cal B}(\vv,\delta)$ is the $\ell_\infty$-ball of configuration vectors of the $q$-state Potts model in $K_n$
around  $\vv$ of radius $\delta$, cf. equation \eqref{eq:ballball}. For a constant $\eta>0$, let
\[U(\eta):=\mathcal{B}(\u,\eta).\]
The following lemma is analogous to Lemma~\ref{lem:measureconc} and its proof hinges on the arguments used to derive the upper bound for the mixing time at $B=\Bu$. (Similarly to Lemma~\ref{lem:measureconc}, more precise bounds can be found in, e.g., \cite{HET}; the following estimate follows easily from our upper bound on the mixing time.)
\begin{lemma}\label{lem:measureconcBu}
Let $B=\Bu$ and $\eta>0$ be a constant. For all sufficiently large $n$, the Potts distribution $\mu$ (given in \eqref{eq:Gibbs}) satisfies $\mu\big(\overline{U(\eta)}\big)\leq 1/8$.
\end{lemma}
\begin{proof}
For convenience, denote $U:=U(\eta)$. By Lemma~\ref{combinecombine}, for all sufficiently large $n$ and any starting state $X_0$, we have that for $T=O(n^{1/3})$, it holds that
\[\Pr(X_T\in U)\geq \epsilon,\]
where $\epsilon>0$ is a constant independent of $n$. It follows that for all non-negative integers $j$ it also holds that 
\[\Pr(X_{(j+1)T}\in U\mid X_{jT}\notin U)\geq \epsilon.\]
 Further, by Lemma~\ref{l2}, for integer $t\geq 0$,  it holds that 
\[\Pr(X_{t+1}\in U\mid X_{t}\in U)\geq 1-\exp(-\Omega(n^{1/3})).\]
We thus obtain that for any starting state $X_0$, for some positive integer $j=j(\epsilon)$, for all sufficiently large $n$, for all  integer $t\geq j T$, it holds that 
\begin{equation}\label{eq:XtSBu}
\Pr(X_{t}\in U)\geq 15/16.
\end{equation}

Let $T^*=\max\{j T,\, 2\Tmix\}$. By Corollary~\ref{cor:mixingtimeBu}, we have $\Tmix=O(n^{1/3})$, so $T^*=O(n^{1/3})$ as well. The same arguments as in the proof of Lemma~\ref{lem:measureconc} (cf. equation \eqref{eq:vdvdvd}) yield 
\begin{equation}\label{eq:vdvdvdBu}
\mu\big(\overline{U}\big)-\Pr\big(X_{T^{*}}\in \overline{U}\big)\leq 1/16.
\end{equation}
Combining \eqref{eq:XtSBu} and \eqref{eq:vdvdvdBu} yields $\mu\big(\overline{U}\big)\leq 1/8$, as wanted.
\end{proof}

The following lemma can be derived from Lemma~\ref{lem:simplifiedpotentialb} by suitably adapting the proof of Lemma~\ref{combinecombine}.
\begin{lemma}\label{lem:implowerboundb}
For $B= \Bu$, there exists a constant $\eta>0$  such that the following is true for all $n$. Suppose that we start at a state $X_0$ where all the vertices are assigned the color 1.  Then, for some $T=\Omega(n^{1/3})$, with probability $\geq 1/2$, it holds that $X_T\notin U(\eta)$.
\end{lemma}
\begin{proof}
Let $M_1,M_2,\rho$ be the constants in Lemma~\ref{lem:simplifiedpotentialb} and let $T:=\left\lceil M_2n^{1/3}/(6\rho)\right\rceil$. 

Recall that $S_t$ is the size of the largest color class at time $t$. Let $\epsilon>0$ be a sufficiently small constant, to be picked later. We will prove that with probability $\geq 1/2$  it holds that 
\begin{equation}\label{eq:qazzaqwsxxsw}
\Pr\big(S_T> (1+\epsilon)n/q\big)\geq 1/2.
\end{equation}
Let $\eta:=\epsilon/q$ and note that $\eta$ is a constant. The lemma then follows by just observing that  $\Pr(X_T\notin U(\eta))\geq \Pr(S_T>(1+\epsilon)n/q)$.

We next argue that \eqref{eq:qazzaqwsxxsw} holds. To do this, we will show that for all $n$ sufficiently large, for all $t=0,\hdots,n$, it holds that
\begin{equation}\label{eq:modifiedpotential}
E[G(S_{t+1}/n)] \geq E[G(S_t/n)] - 2\rho,
\end{equation}
where $G,\rho$ are the potential function and the constant from Lemma~\ref{lem:simplifiedpotentialb}, respectively. Prior to proving \eqref{eq:modifiedpotential}, let us conclude the argument assuming~\eqref{eq:modifiedpotential}. Lemma~\ref{lem:simplifiedpotentialb} asserts that the constants $M_1,M_2$ are such that 
\begin{equation}\label{eq:G1a1a}
0\leq G(z)\leq G(1) \mbox{ for all } z\in [1/q,1],\mbox{ with } G(1)=Cn^{1/3} \mbox{ and $C$  satisfying $M_2\leq C\leq M_1$}.
\end{equation}
Applying \eqref{eq:modifiedpotential} for $t=0,\hdots,T-1$, we obtain that
\[E[G(S_{T}/n)]\geq G(S_0/n)-2\rho T.\]
Since $S_0=n$ and $G(1)= Cn^{1/3}$, it thus follows that $E[G(S_{T}/n)]\geq (2/3)Cn^{1/3}$. Let $\epsilon>0$ be such that $(1+\epsilon)/q<1/B$; such an $\epsilon$ exists since at $B=\Bu$ it holds that $1/q<1/B$. From $G(1/B)=O(1)$ and the fact that $G$ is increasing, we obtain that there exists a constant $\xi>0$ such that $G((1+\epsilon)/q)\leq \xi$. It is immediate now to conclude that with probability $\geq 1/2$ it holds that $S_T> (1+\epsilon)n/q$; otherwise, using \eqref{eq:G1a1a}, we would have that for sufficiently large $n$, it holds that $E[G(S_{T}/n)]\leq (3/5)Cn^{1/3}$, contradicting our lower bound for $E[G(S_{T}/n)]$.

Finally, we prove \eqref{eq:modifiedpotential} for $t=0,\hdots,n$. We will use Lemmas~\ref{lestart2} and~\ref{lem:simplifiedpotentialb}. Let $\epsilon>0$ be a small constant as in the statement of  Lemma~\ref{lestart2}.  Note that Lemma~\ref{lem:simplifiedpotentialb} applies whenever $X_t$ has $q-1$ $\eps$-light colors, so  we will need to account for the (small probability) event that this fails. Namely, let $\mathcal{E}_t$ denote the event that $X_t$ has $q-1$ $\eps$-light colors. Since the event $\mathcal{E}_0$ holds (by the choice of the starting state $X_0$), we have that $\bigcap^n_{t=0}\mathcal{E}_t$  holds with probability  at least $1-\exp(-n^{\Omega(1)})$ (by the second part of Lemma~\ref{lestart2}).

Let $t$ be an integer between 0 and $n$. By taking expectations in inequality \eqref{eq:potentialcopyb} of Lemma~\ref{lem:simplifiedpotentialb}, we have 
\begin{equation}\label{Gtau1b}
E\big[G(S_{t+1}/n)\mid \mathcal{E}_t\big]\geq E[G(S_t/n)\mid \mathcal{E}_t]-\rho.
\end{equation}
Since $G$ is bounded by a polynomial (cf. \eqref{eq:G1a1a}) and the probability of the event $\overline{\mathcal{E}_t}$ is exponentially small, removing the conditioning in \eqref{Gtau1b} only affects the inequality by an additive $o(1)$. This proves that \eqref{eq:modifiedpotential} holds for all sufficiently large $n$, thus concluding the proof of Lemma~\ref{combinecombine}.
\end{proof}

Using Lemmas~\ref{lem:measureconcBu} and~\ref{lem:implowerboundb}, we obtain the following corollary.
\begin{corollary}
Let $B= \Bu$. The mixing time $\Tmix$ of the Swendsen-Wang algorithm
on the complete graph on $n$ vertices satisfies $\Tmix=\Omega(n^{1/3})$.
\end{corollary}
\begin{proof}
Let $\eta$ be as in Lemma~\ref{lem:implowerboundb} and let $U:=U(\eta)$. Consider the starting state $X_0$ where all the vertices are assigned the color 1. Then, by Lemma~\ref{lem:implowerboundb}, for some $T=\Omega(n^{1/3})$ we have that 
\[\Pr\big(X_T\notin U\big)\geq 1/2.\]
On the other hand, by Lemma~\ref{lem:measureconcBu} we have that $\mu(\overline{U})\leq 1/8$. It follows that 
\[d_{TV}(X_T,\mu)=\max_{A\subseteq \Omega}|\mu(A)-\Pr(X_T\in A)|\geq \Pr(X_T\in \overline{U})-\mu(\overline{U})\geq 1/2-1/8>1/4.\]
It follows from the definition of mixing time that $\Tmix\geq T$, as claimed.
\end{proof}

\subsection{Constructing the potential function - Proof of Lemmas~\ref{lem:simplifiedpotential} and~\ref{lem:simplifiedpotentialb}}\label{sec:potfun}
In this section, we prove Lemmas~\ref{lem:simplifiedpotential} and~\ref{lem:simplifiedpotentialb}, i.e., construct the potential function $G$. We split the argument in several lemmas.

The first lemma achieves two goals: first, it quantifies the bounds that the function $G$ must satisfy so that the approximation
\begin{equation*}\tag{\ref{eq:motivation}}
E[G(S_{t+1}/n)\,|\, S_t=\zeta]\approx G(F(\zeta/n))+\frac{1}{2}Var[S_{t+1}/n\,|\, S_t=\zeta]\, G''(F(\zeta/n)),
\end{equation*}
which we described in Section~\ref{sec:upperboundBu} is valid; the bounds are given in \eqref{eq:Gders123}. Second, it gives an inequality that the function $G$ must satisfy (cf. equation \eqref{eq:QQ1}) which allows to deduce, using the approximation \eqref{eq:motivation}, the bounds on  $E[G(S_{t+1}/n)\,|\, S_t=\zeta]-G(\zeta/n)$ claimed in Lemmas~\ref{lem:simplifiedpotential} and~\ref{lem:simplifiedpotentialb} (see \eqref{eq:potential} below).
\begin{lemma}\label{lma2}
Let $\epsilon>0$. Suppose that, for all $n$ sufficiently large, $S_t$ and $S_{t+1}$ are random variables that satisfy  \eqref{eko1},\eqref{eko2},\eqref{eko3} when $\zeta\geq (1+\epsilon)n/B$.  

Let $G$ be a three-times differentiable potential function defined on the interval $[1/q,1]$ such that 
\begin{equation}\label{eq:Gders123}
\mbox{$\min_{x}G'(x)> 0$, $\max_{x} |G'(x)|=O(n^{2/3})$, $\max_{x} |G''(x)|=O(n)$, $\sup_{x} |G'''(x)| = O(n^{4/3})$}.
\end{equation} 
Further, assume that for each $x> 1/B$, it holds that
\begin{equation}\label{eq:QQ1}
\begin{aligned}
-\tau_2<G(F(x)) - G(x) + G''(F(x)) Q_1/(2n) &< -\tau_1,\\
-\tau_2<G(F(x)) - G(x) + G''(F(x)) Q_2/(2n) &< -\tau_1,
\end{aligned}
\end{equation}
where $\tau_1,\tau_2>0$ are constants (independent of $n$) and $Q_1,Q_2$ are as in \eqref{eko2}.

Then, for any $\zeta\geq (1+\epsilon) n/B$, it holds that
\begin{equation}\label{eq:potential}
G(\zeta/n) - 2\tau_2\leq E[G(S_{t+1}/n)\,|\, S_t=\zeta] \leq G(\zeta/n) - \tau_1/2.
\end{equation}
\end{lemma}

Recall that for $B=\Bu$, the function $F(z)$  has exactly one  fixpoint in the interval $(1/B,1]$ at $z=a$.  The following lemma specifies a potential function $G$ which will be used to verify the conditions \eqref{eq:Gders123} and \eqref{eq:QQ1} in Lemma~\ref{lma2}. We have already described in Section~\ref{sec:upperboundBu}, the high-level approach for the construction of $G$. The actual definition  of $G$ is quite technical due to the requirement that $G$ should be  three times differentiable. We pulled out the important bits in the construction of $G$ that will also be relevant in verifying \eqref{eq:QQ1}.

For positive real numbers $A,B$ we will use the notation $A\gg B$ to denote that for some (large) constant $C>1$, it holds that $A>BC$.
\begin{lemma}\label{lem:potentialchoice}
Let $L,L'$ be positive constants which satisfy $L\gg L'$. There exist positive constants $M,C_0,C_1,C_2$ such that the following holds. 

For all sufficiently large $n$, there exists a strictly  increasing three-times differentiable function $G:[1/q,1]\rightarrow [0,Mn^{1/3}]$ with $G(1/q)=0$ which satisfies \eqref{eq:Gders123} and
\begin{equation}\label{eq:Gspecify}
\begin{gathered}
|G'(z)|,|G''(z)|\leq C_0  \mbox{ for } z\in [1/q,1/B],\\
G'(z) = \frac{1}{z-F(z)} \mbox{ for } z\in [1/B,a-Ln^{-1/3}] \cup [a+Ln^{-1/3}, 1],\\
G'(z)\geq C_1n^{2/3},\ \ |G''(z)|\leq (10^2C_1/L) n \mbox{ for  } z\in [a-Ln^{-1/3},a-L'n^{-1/3}],\\
G''(z)\leq -C_2n \mbox{ for } z\in [a-L'n^{-1/3},a+Ln^{-1/3}].
\end{gathered}
\end{equation}  
\end{lemma}

\begin{lemma}\label{lem:satisfyineq}
Let $L,L'$ be positive constants which satisfy $L\gg L'\gg1$. Then, there exist constants $\tau_1,\tau_2>0$, such that, for any function $G$ satisfying \eqref{eq:Gders123} and \eqref{eq:Gspecify},   inequality \eqref{eq:QQ1} holds for every $x>1/B$. 
\end{lemma}

The following lemma will be useful throughout the rest of this section.
\begin{lemma}\label{lem:Flemma}
Let $B=\Bu$. Then  it holds that
\begin{enumerate}
\item $F'(z)=1$ iff $z=a$.
\item $F''(z)< 0$ for all $z\in(1/B,1]$.  
\item $F(z)\leq z$ for all $z\in[1/B,1]$ with equality iff $z=a$.
\end{enumerate}
\end{lemma}
\begin{proof}
The proofs for the first two parts are given in Lemmas~\ref{lem:notJacobianuniqueness} and~\ref{lem:Fshape}, respectively.  For the third part, note that the function $z-F(z)$ is convex in $[1/B,1]$ and has a unique critical point at $z=a$. Thus, $z-F(z)\geq a-F(a)=0$ with equality if $z=a$.
\end{proof}

We are now ready to prove Lemmas~\ref{lem:simplifiedpotential} and~\ref{lem:simplifiedpotentialb} (assuming Lemmas~\ref{lma2},~\ref{lem:potentialchoice} and~\ref{lem:satisfyineq}). 
\begin{proof}[Proof of Lemma~\ref{lem:simplifiedpotential}]
Let $L,L'$ be positive constants satisfying $L\gg L'\gg 1$. By Lemmas~\ref{lem:potentialchoice} and~\ref{lem:satisfyineq}, there exist constants $M,\tau_1,\tau_2>0$ such that for all sufficiently large $n$  there exists a three-times differentiable function $G:[1/q,1/B]\rightarrow [0,Mn^{1/3}]$ which satisfies both \eqref{eq:Gders123} and \eqref{eq:QQ1}. Note that \eqref{eq:Gders123} guarantees that $G$ is increasing. Further, by Lemma~\ref{lem:potentialchoice}, it holds that $G(1/q)=0$ and $\max_{z\in [1/q,1/B]}G'(z)\leq C_0$ where $C_0$ is a constant. We will use this function $G$ to prove Lemma~\ref{lem:simplifiedpotential} with $M_1=M$, $M_2=C_0$ and $\tau=\tau_1/2$.

Let $\epsilon>0$ be a sufficiently small constant and suppose that  $X_t$ has $(q-1)$ colors which are $\epsilon$-light. Recall that $S_t$ is the size of the largest color class in $X_t$. By Lemma~\ref{lestart}, we have that for all sufficiently large $n$, for all $\zeta\geq (1+\epsilon)n/B$, the random variables $S_t,S_{t+1}$ satisfy \eqref{eko1},\eqref{eko2},\eqref{eko3}. It follows by Lemma~\ref{lma2} that
\[E[G(S_{t+1}/n)\,|\, S_t=\zeta] \leq G(\zeta/n) - \tau_1/2.\]
This completes the verification of  all the conditions that  $G$ must satisfy, concluding the proof of the lemma.
\end{proof}

\begin{proof}[Proof of Lemma~\ref{lem:simplifiedpotentialb}]
We begin by specifying some constants. Let $\epsilon_0>0$ be a constant such that $(1+\epsilon_0)/B<a$ and let 
\begin{equation}\label{eq:defW0W0}
W_0:=\min_{z\in [1/B,(1+\epsilon_0)/B]} \{z-F(z)\}.
\end{equation}
By Lemma~\ref{lem:Flemma} and the choice of $\epsilon_0$, we have that $W_0>0$.

Consider positive constants $L,L'$ satisfying $L\gg L'\gg 1$. By Lemmas~\ref{lem:potentialchoice} and~\ref{lem:satisfyineq}, there exist constants $M,\tau_1,\tau_2>0$ such that for all sufficiently large $n$  there exists a function $G:[1/q,1/B]\rightarrow [0,Mn^{1/3}]$ which satisfies all of \eqref{eq:Gders123},  \eqref{eq:QQ1} and \eqref{eq:Gspecify}. Further, by Lemma~\ref{lem:potentialchoice}, there exist positive constants $C_0,C_1$ such that 
\[G(1/q)=0, \quad \max_{z\in [1/q,1/B]}G'(z)\leq C_0, \quad \min_{z\in [a-Ln^{-1/3},a-L'n^{-1/3}]}G'(z)\geq C_1 n^{2/3}.\]
Therefore, $G(1/B)=O(1)$ and $G(1)\geq C_3 (L-L') n^{1/3}$ (for the latter we also need that $G$ is increasing which is guaranteed from \eqref{eq:Gders123}). We will also need a bound on the variation of $G$ on the interval $[1/q,(1+\epsilon)/B]$.  Using \eqref{eq:Gspecify} and \eqref{eq:defW0W0}, we have that $\max_{z\in [1/B,(1+\epsilon_0)/B]}G'(z)\leq 1/W_0$. It follows that for $\eta_0:=\max\{C_0,1/W_0\}$, it holds that $G'(z)\leq \eta_0$  for all $z\in [1/q,(1+\epsilon_0)/B]$ and thus there exists a constant $\eta>0$ such that 
\begin{equation}\label{eq:Gz1z2lip}
|G(z_1)-G(z_2)|\leq \eta \mbox{ for all } z_1,z_2\in [1/q,(1+\epsilon_0)/B].
\end{equation}
We will use $G$ to prove Lemma~\ref{lem:simplifiedpotentialb} with $M_1=M$, $M_2=C_3(L-L')$ and $\rho=2\tau_2+2\eta$.

Let $\epsilon>0$ be a sufficiently small constant and $n$ be sufficiently large. Suppose that $X_t$ has $(q-1)$ colors which are $\epsilon$-light.  Recall that $S_t$ is the size of the largest color class in $X_t$ and suppose that $S_t=\zeta$ where $\zeta\geq n/q$. We will split the proof into cases depending on whether $\zeta\geq (1+\epsilon)n/B$.

Consider first the case where $\zeta\geq (1+\epsilon)n/B$. By Lemma~\ref{lestart}, we have that the random variables $S_t,S_{t+1}$ satisfy \eqref{eko1},\eqref{eko2},\eqref{eko3}. It follows by Lemma~\ref{lma2} that
\[E[G(S_{t+1}/n)\,|\, S_t=\zeta] \geq G(\zeta/n) - 2\tau_2.\]

Consider now the case where $\zeta\leq (1+\epsilon)n/B$ so that $S_t\leq (1+\epsilon)n/B$. Let $\mathcal{E}_t$ be the event that $S_{t+1}\leq (1+\epsilon)n/B$. By Lemma~\ref{lem:probdeviation}, we have that $\Pr(\mathcal{E}_t)=1-\exp(-n^{\Omega(1)})$. Also, using \eqref{eq:Gz1z2lip}, we have that 
\begin{equation}\label{eq:StponeSt}
E[G(S_{t+1}/n)\,|\, S_t=\zeta, \mathcal{E}_t ] \geq G(\zeta/n)-\eta.
\end{equation}
Recall that $G$ is non-negative with values that  are polynomially bounded. Since $\mathcal{E}_t$ holds with exponentially large probability, it follows that  removing the conditioning on the event $\mathcal{E}_t$ in \eqref{eq:StponeSt} only affects the inequality by $o(1)$. Hence, for all sufficiently large $n$, it holds that
\[E[G(S_{t+1}/n)\,|\, S_t=\zeta ] \geq G(\zeta/n)-2\eta.\]
This completes the verification of  all the conditions that  $G$ must satisfy, concluding the proof of the lemma.
\end{proof}

\begin{proof}[Proof of Lemma~\ref{lma2}]
Let $x=\zeta/n$ and $y=E[S_{t+1}/n\,|\, S_t = \zeta]$. Let
$Z = S_{t+1}/n - y$. Note that $Z$ is a random variable, $E[Z\, |\, S_t = \zeta]=0$,
and by Lemma~\ref{lestart},
\[Q_1/n \leq Var[Z\, |\, S_t=\zeta]=E[Z^2\, |\, S_t=\zeta]\leq Q_2/n.\]
By Taylor's expansion, we have
\begin{equation}\label{oooopp}
G(y+Z) = G(y) + G'(y)Z + G''(y)\frac{Z^2}{2} +
G'''(\rho)\frac{Z^3}{6},
\end{equation}
for some $\rho$ which lies between $y$ and $y+Z$ (note that $\rho$ is also a random variable).

From inequality~\eqref{eko3} of Lemma~\ref{lestart} we have for all sufficiently small $\epsilon'>0$
$$
 E[|Z|^3\, |\, S_t = \zeta] \leq K n^{-3/2+\epsilon'}.
$$
Taking expectations of~\eqref{oooopp} we obtain
\begin{equation}\label{eq:expansion}
\begin{split}
E[G(S_{t+1}/n)\,|\, S_t=\zeta] =
E[G(y+Z)\,|\, S_t=\zeta] =
G(y) + G''(y)\frac{E[Z^2\,|\, S_t=\zeta]}{2} +
C,
\end{split}
\end{equation}
where $|C|\leq K n^{-3/2+\epsilon'}\sup_x |G'''(x)|=o(1)$ since $\sup_x |G'''(x)|=O(n^{4/3})$. 
Using~\eqref{eko1} of Lemma~\ref{lestart} (for $\epsilon'=1/10$), we have
$$
|G(y)-G(F(x))|\leq \frac{n^{1/10}}{n}\sup_x |G'(x)|
\quad\mbox{and}\quad
|G''(y) - G''(F(x))| \leq \frac{n^{1/10}}{n}\sup_x |G'''(x)|.
$$
Plugging these estimates in \eqref{eq:expansion} we obtain
\begin{equation}\label{eq:expansiontwo}
\begin{split}
\left| E[G(S_{t+1}/n)\,|\, S_t=\zeta]-\left(G(F(x))+G''(F(x))\frac{E[Z^2\, |\, S_t=\zeta]}{2}\right)\right|\leq R,
\end{split}
\end{equation}
where $R$ is an error term satisfying 
\[|R|\leq \frac{n^{1/10}}{n}\sup_x |G'(x)|+\frac{n^{1/10}}{n}\sup_x |G'''(x)|\frac{E[Z^2\, |\, S_t=\zeta]}{2}+C.\] From $\sup_x |G'(x)|=O(n^{2/3})$,  $\sup_x |G'''(x)|=O(n^{4/3})$ and  $E[Z^2\, |\, X_t=\zeta]\leq Q_2/n$, we obtain that $|R|=o(1)$.

It thus follows from \eqref{eq:expansiontwo} that
\begin{equation}\label{eq:expansiontwob1}
E[G(S_{t+1}/n)\,|\, S_t=\zeta]-G(\zeta/n)=G(F(x))-G(x)+G''(F(x))\frac{E[Z^2\, |\, S_t=\zeta]}{2}+o(1).
\end{equation}
We also have that 
\begin{align}
G(F(x))-G(x)+G''(F(x))\frac{E[Z^2\, |\, S_t=\zeta]}{2}&\leq G(F(x))-G(x)+ \frac{\max\{Q_1G''(F(x)),Q_2G''(F(x))\}}{2n}\notag\\
&\leq -\tau_1\label{eq:tgbnm1}
\end{align}
where in the first inequality we used that $Q_1/n\leq E[Z^2\, |\, S_t=\zeta]\leq Q_2/n$ (note that both estimates are needed since we do not know the sign of $G''$) and in the second inequality we used \eqref{eq:QQ1}. Analogously, one has 
\begin{align}
G(F(x))-G(x)+G''(F(x))\frac{E[Z^2\, |\, S_t=\zeta]}{2}&\geq G(F(x))-G(x)+ \frac{\min\{Q_1G''(F(x)),Q_2G''(F(x))\}}{2n}\notag\\
&\geq -\tau_2\label{eq:tgbnm2}
\end{align}
Combining \eqref{eq:expansiontwob1}, \eqref{eq:tgbnm1} and \eqref{eq:tgbnm2}, it follows that for all sufficiently large $n$ it holds that 
\begin{equation*}
-2\tau_2\leq -\tau_2+o(1)\leq E[G(S_{t+1}/n)\,|\, S_t=\zeta]-G(\zeta/n)\leq -\tau_1+o(1)\leq -\tau_1/2
\end{equation*}
This proves that \eqref{eq:potential} holds, as wanted.
\end{proof}

We next prove Lemmas~\ref{lem:potentialchoice} and~\ref{lem:satisfyineq}. It is more instructive to use Lemma~\ref{lem:potentialchoice} as a black box for now and prove Lemma~\ref{lem:satisfyineq} first. 
\begin{proof}[Proof of Lemma~\ref{lem:satisfyineq}]
Let $L\gg L'\gg1$ be constants and $n$ be large. Let
\[z_-:=a-Ln^{-1/3},\quad z_{-}':=a-L'n^{-1/3}, \quad  z_+:=a+Ln^{-1/3},\]
and consider the intervals
\begin{equation*}
I_0=[1/q,1/B], \, \, \, I_1=[1/B,z_-],\, \, \, I_2=[z_-,z_{-}'], \, \, \, I_3=[z_{-}',z_+], \, \, \, I_4=[z_+,1].
\end{equation*}
Let $G$ be a function defined on the interval $[1/q,1]$ that satisfies  \eqref{eq:Gders123} and \eqref{eq:Gspecify}, i.e.,
\begin{equation*}\tag{\ref{eq:Gders123}}
\mbox{$\min_{x}G'(x)> 0$, $\max_{x} |G'(x)|=O(n^{2/3})$, $\max_{x} |G''(x)|=O(n)$, $\sup_{x} |G'''(x)| = O(n^{4/3})$},
\end{equation*} 
and 
\begin{equation*}\tag{\ref{eq:Gspecify}}
\begin{gathered}
|G'(z)|,|G''(z)|\leq C_0 \mbox{ for } z\in I_0,\\
G'(z) = \frac{1}{z-F(z)} \mbox{ for } z\in I_1 \cup I_4,\\
G'(z)\geq C_1n^{2/3},\ \ |G''(z)|\leq (10^2C_1/L) n \mbox{ for  } z\in I_2,\\
G''(z)\leq -C_2n \mbox{ for } z\in I_3,
\end{gathered}
\end{equation*}
where recall that $C_0,C_1,C_2$ are positive constants.

Our goal is to show that there exist constants $\tau_1,\tau_2>0$, so  that for all $x\in(1/B,1]$ it holds that 
\begin{equation*}\tag{\ref{eq:QQ1}}
\begin{aligned}
-\tau_2<G(F(x)) - G(x) + G''(F(x)) Q_1/(2n) &< -\tau_1,\\
-\tau_2<G(F(x)) - G(x) + G''(F(x)) Q_2/(2n) &< -\tau_1,
\end{aligned}
\end{equation*}
where $Q_1,Q_2$ are positive constants satisfying $Q_2\leq Q_1$.

We first show the inequality \eqref{eq:QQ1} in the easier regime where $x\in (1/B,1]$ and $x\notin (a-\epsilon,a+\epsilon)$, for any arbitrarily small constant $\epsilon>0$ when $n$ is sufficiently large. By Lemma~\ref{lem:Flemma}, for all $x\neq a$ such that $x\in I_1\cup I_4$ it holds that $F(x)<x$. Let $\rho:=F(\alpha+\epsilon)-\alpha$, so that $\rho\in(0,\epsilon)$. Set
\[W_1:=\min_{x\in (1/B,1],\ x\notin(a-\rho,a+\rho)}\{x-F(x)\}, \quad W_2:=\max_{x\in (1/B,1],\ x\notin(a-\rho,a+\rho)}\{x-F(x)\}.\] 
Since $x-F(x)$ is continuous, we obtain that $W_1,W_2>0$. By taking $n$ sufficiently large, we obtain that any $x\in (1/B,1]$ such that $x\notin (a-\rho,a+\rho)$ belongs to $I_1\cup I_4$ and therefore,  from \eqref{eq:Gspecify}, $G'(x)$ is upper and lower bounded by the absolute constants $1/W_1$ and $1/W_2$ for all $x\notin (a-\rho,a+\rho)$. Hence, there exist constants $W_1',W_2'>0$ such that for all $x\notin (a-\epsilon,a+\epsilon)$, it holds that
\[-W_2'\leq G(F(x))-G(x)\leq -W'_1.\] 
A similar argument shows that $|G''(x)|$ is bounded by a constant for all $x\in [1/q,1]$ with $x\notin(a-\rho,a+\rho)$. It follows that for all $x\in (1/B,1]$ and $x\notin (a-\epsilon,a+\epsilon)$ it holds that
\begin{equation}\label{eq:baza0}
\begin{aligned}
\max_{i\in \{1,2\}}G(F(x))-G(x)+G''(F(x))Q_i/(2n)&\leq -W'_1+o(1),\\
\min_{i\in \{1,2\}}G(F(x))-G(x)+G''(F(x))Q_i/(2n)&\geq -W_2'+o(1).
\end{aligned}
\end{equation} 
This proves \eqref{eq:QQ1} when $x\notin(a-\epsilon,a+\epsilon)$.

We next prove \eqref{eq:QQ1} when $x\in(a-\epsilon,a+\epsilon)$ for some appropriate constant $\epsilon>0$ to be specified next.  Let $c=-F''(a)/2$ and note that $c>0$ by Lemma~\ref{lem:Flemma}. Using again Lemma~\ref{lem:Flemma} and Taylor's Theorem, there exists $\epsilon''>0$ such that for all $z\in(a-\epsilon'',a+\epsilon'')$, it holds that
\begin{equation*}
F(z)=z-c(z-a)^2+O((z-a)^3).
\end{equation*}
Hence, there exists $\epsilon'>0$ so that for all $z\in (a-\epsilon',a+\epsilon')$ it holds that
\begin{equation}\label{eq:expansioninequality}
\begin{gathered}
\frac{1}{2}c (z-a)^2\leq z-F(z)\leq 2c (z-a)^2,\\
c|z-a|\leq |F'(z)-1|\leq  4c|z-a|.
\end{gathered}
\end{equation}
Let $\epsilon>0$ be a small constant such that $\epsilon+ 2c\epsilon^2<\epsilon'$ and $4c\epsilon, 4c\epsilon^2<1/8$. With this choice of $\epsilon$, we will be able to use the expansion of $F(z)$ around $z=a$. Before we proceed, we give a few intermediate inequalities that  will be later used to establish the desired inequalities in \eqref{eq:QQ1}.

For $x\in (a-\epsilon,a+\epsilon)$, we will use the parametrization $x=a+Kn^{-1/3}$ so that $|K|\leq \epsilon n^{1/3}$. From \eqref{eq:expansioninequality}, we have that  
\begin{equation}\label{eq:xFxxFx}
\frac{1}{2}cK^2n^{-2/3}\leq x-F(x)\leq 2cK^2n^{-2/3}.
\end{equation} 
By the Mean Value Theorem, we also have that there exists $\xi\in (F(x),x)$ such that 
\begin{equation}\label{eq:GFxGFx}
G(F(x))-G(x)=G'(\xi)(F(x)-x).
\end{equation}
Since $\xi\in (F(x),x)$, we have by \eqref{eq:xFxxFx} that $\xi=x-\kappa cK^2n^{-2/3}$ for some $1/2\leq \kappa\leq 2$. By the choice of $\epsilon$, it follows that $\xi\in (a-\epsilon',a+\epsilon')$ and hence, using \eqref{eq:expansioninequality}, we obtain $\frac{1}{2}c (\xi-a)^2\leq \xi-F(\xi)\leq 2c (\xi-a)^2$. Note that $\xi=a+Kn^{-1/3}-\kappa cK^2n^{-2/3}$, so using that $4c\epsilon, 4c\epsilon^2<1/8$, we obtain 
\begin{equation}\label{ineq:xi}
\frac{1}{4}c K^2n^{-2/3}\leq \xi-F(\xi)\leq 4c K^2n^{-2/3}.
\end{equation}
Finally, for the lower bounds we will use sometimes the following immediate consequences of \eqref{eq:Gders123}: there exist constants $C_1',C_2'>0$ such that for all $x\in[1/q,1]$ it holds that 
\begin{equation}\label{ineq:xi2}
\begin{gathered}
0\leq G'(x)\leq C_1'n^{2/3},\qquad |G''(x)|\leq C_2' n.
\end{gathered}
\end{equation} 

We are now ready to give the proof of \eqref{eq:QQ1} for $x\in (a-\epsilon,a+\epsilon)$. The proof splits into cases depending on the value of $K$ in the parametrization $x=a+Kn^{-1/3}$. \vskip 0.2cm

\noindent \textbf{Case I.} $K\leq -L$ or $K\geq L$. We  will do the case $K\leq -L$, the proof for $K\geq L$ is analogous. For $K\leq -L$, we have that $x\in I_1$. From \eqref{eq:xFxxFx}, we  also have $F(x)\in I_1$. In fact, our choice of $\epsilon$ guarantees that $F(x)\in (a-\epsilon',a+\epsilon')$, where recall that $\epsilon'$ is as in \eqref{eq:expansioninequality}.

For $z\in (a-\epsilon',a+\epsilon')$, we have $G'(z)=1/(z-F(z))$ and thus $G''(z)=\frac{F'(z)-1}{(z-F(z))^2}$. From \eqref{eq:expansioninequality}, we thus obtain that $|G''(z)|\leq \frac{16}{c|z-a|^3}$. Applying this for $z=F(x)$ and observing that $F(x)-a\leq -Ln^{-1/3}$, we obtain 
\[|G''(F(x))|\leq \frac{16n}{cL^3}, \]
so that 
\[\max_{i\in\{1,2\}}|G''(F(x))Q_i/(2n)|\leq \frac{8Q_1}{cL^3}.\] 
Let $\xi$ be as in \eqref{eq:GFxGFx}. Since $\xi\in(F(x),x)$, we have from \eqref{eq:Gspecify} that $G'(\xi)=1/(\xi-F(\xi))$. From \eqref{ineq:xi}, we have $1/(4cK^2 n^{-2/3})\leq G'(\xi)\leq 4/(cK^2 n^{-2/3})$. It follows from \eqref{eq:xFxxFx} and \eqref{eq:GFxGFx} that 
\[-8\leq G(F(x))-G(x)\leq -1/8.\]
Combining the above estimates, we can conclude that 
\begin{equation}\label{eq:baza4}
\begin{aligned}
\max_{i\in \{1,2\}}G(F(x))-G(x)+G''(F(x))Q_i/(2n)&\leq -\frac{1}{8}+\frac{8Q_1}{cL^3},\\
\min_{i\in \{1,2\}}G(F(x))-G(x)+G''(F(x))Q_i/(2n)&\geq -8-\frac{8Q_1}{cL^3}.
\end{aligned}
\end{equation}
Since we can choose $L$ to be an arbitrarily large constant, we can make the right-side quantities in \eqref{eq:baza4} to be negative constants, as needed. This completes the proof for  Case I. \vskip 0.2cm

\noindent \textbf{Case II. } $-L\leq K\leq -L'$. In this case, we have $x\in I_2$. It follows from \eqref{eq:xFxxFx} that 
\begin{equation}\label{eq:xinCaseI}
-\frac{1}{2}c(L')^2n^{-2/3}\geq F(x)-x\geq -2c L^2n^{-2/3}.
\end{equation} 
Since $x\in I_2$, from \eqref{eq:xinCaseI}, for all sufficiently large $n$, we clearly have that either $F(x)\in I_2$ or $F(x)\in I_1$.

Suppose first that $F(x)\in I_2$. From \eqref{eq:Gspecify} we have that $|G''(F(x))|\leq (10^2C_1/L) n$, so that
\[\max_{i\in\{1,2\}}|G''(F(x))Q_i/n|\leq 10^2C_1Q_1/L.\]
Since  $F(x)\in I_2$ and $\xi\in (F(x),x)$, we have $\xi\in I_2$ as well, so from \eqref{eq:Gspecify} we have $G'(\xi)\geq C_1n^{2/3}$. We also have from \eqref{ineq:xi2} that $G'(\xi)\leq C_1'n^{2/3}$. Thus, together with \eqref{eq:GFxGFx} and \eqref{eq:xinCaseI}, we obtain
\begin{equation}\label{eq:C1C1pF}
-2cL^2 C_1'\leq G(F(x))-G(x)\leq  -\frac{1}{2}c(L')^2C_1.
\end{equation}
It follows that 
\begin{equation}\label{eq:baza}
\begin{aligned}
\max_{i\in \{1,2\}}G(F(x))-G(x)+G''(F(x))Q_i/(2n)&\leq -C_1\Big(\frac{1}{2}c(L')^2-10^2Q_1/L\Big),\\
\min_{i\in \{1,2\}}G(F(x))-G(x)+G''(F(x))Q_i/(2n)&\geq -\Big(2cC_1' L^2+10^2C_1Q_1/L\Big).
\end{aligned}
\end{equation}

Suppose next that $F(x)\in I_1$ so that $a-Ln^{-1/3}\geq F(x)$. From the lower bound in \eqref{eq:xinCaseI}, we obtain $F(x)\geq a-Ln^{-1/3}-2cL^2n^{-2/3}$. Using that $\sup_{x} |G'''(x)| = O(n^{4/3})$ from \eqref{eq:Gders123}, it thus follows that $\big||G''(F(x))|- |G''(a-Ln^{-1/3})|\big|=O(n^{2/3})$. Moreover, using that $\max_{x} |G''(x)| = O(n)$ from \eqref{eq:Gders123} and $\xi\geq F(x)$, we see that $\big||G'(\xi)|-|G'(a-Ln^{-1/3})|\big|=O(n^{1/3})$. Combining these estimates yields again \eqref{eq:baza} (up to a $o(1)$ term which can be ignored for large $n$).

Since we can choose $L$ to be an arbitrarily large constant, we can make the right-side quantities in \eqref{eq:baza} to be negative constants, as needed. This completes the proof for  Case II.  \vskip 0.2cm

\noindent \textbf{Case III.} $-L'\leq K\leq L$. In this case, we have $x\in I_3$. Observe that \eqref{eq:xinCaseI} holds in this case as well, so for all sufficiently large $n$, we have that either $F(x)\in I_2$ or $F(x)\in I_3$. 

If $F(x)\in I_3$ then from \eqref{eq:Gspecify} and \eqref{ineq:xi2}, we have $-C_2'n\leq G''(F(x))\leq -C_2n$.  Since $G$ is increasing and $F(x)\leq x$, we trivially have $G(F(x))-G(x)\leq 0$. The lower bound on $G(F(x))-G(x)$ from \eqref{eq:C1C1pF} is valid in this case as well (since both \eqref{ineq:xi2} and \eqref{eq:xinCaseI}  hold), so we obtain
\[-2cL^2 C_1'\leq G(F(x))-G(x)\leq 0.\] 
It follows that
\begin{equation}\label{eq:baza3}
\begin{aligned}
\max_{i\in \{1,2\}}G(F(x))-G(x)+G''(F(x))Q_i/(2n)&\leq -C_2Q_2,\\
\min_{i\in \{1,2\}}G(F(x))-G(x)+G''(F(x))Q_i/(2n)&\geq -2cL^2 C_1'-C_2'Q_1.
\end{aligned}
\end{equation}

If $F(x)\in I_2$ then $a-L'n^{-1/3}\geq F(x)$. From \eqref{eq:xinCaseI} we obtain $F(x)\geq a-L'n^{-1/3}-2cL^2n^{-2/3}$. It follows that $|G(F(x))-G(a-L'n^{-1/3})|=o(1)$ and $\big||G''(F(x))|- |G''(a-L'n^{-1/3})|\big|=O(n^{2/3})$, yielding again \eqref{eq:baza3} (up to a $o(1)$ term which can be ignored for large $n$).

The right-side quantities in \eqref{eq:baza3} are negative constants, as needed. This completes the proof for  Case III.  \vskip 0.2cm

We have shown that \eqref{eq:QQ1} holds for all $x\in(1/B,1]$, thus finishing  the proof of Lemma~\ref{lem:satisfyineq}.
\end{proof}

We conclude by giving the  proof of Lemma~\ref{lem:potentialchoice}.
\begin{proof}[Proof of Lemma~\ref{lem:potentialchoice}]
Let $L,L'$ be positive constants satisfying $L\gg L'$. To keep better track of the various subintervals involved in the construction of the potential function $G$, define $L_-,L_+,L_m$ by setting $L_-=L_+=L$ and $L_m=L'$ and note that $L_+,L_-\gg L_m$. Further, set 
\[z_-:=a-L_-n^{-1/3},\quad z_{m}:=a-L_mn^{-1/3}, \quad  z_+:=a+L_+n^{-1/3}.\]
The function $G$ will be more complicated to construct in an interval around $a$. To help the reader keep track of the notation, we note that $z_-$ refers to the left-most point of the interval around $a$ that will be interesting, while $z_+,z_m$ to the right-most and ``middle" points of the interval, respectively.

We will define piecewise the function $G(z)$  in the intervals
\[I_0=[1/q,1/B],\, \, \, I_1=[1/B,z_-],\, \, \, I_2=[z_-,z_m], \, \, \, I_3=[z_m,z_+], \, \, \, I_4=[z_+,1].\] 
Specifically, for $j\in \{0,1,2,3,4\}$, let $G_j(z)$ be a strictly increasing three-times differentiable function defined on the interval $I_j$ which satisfies 
\begin{equation}\label{eq:iuyiuy}
\sup_{z\in I_j}|G_j(z)|=O(n^{1/3}),\quad \sup_{z\in I_j}|G_j'(z)|=O(n^{2/3}),\quad  \sup_{z\in I_j}|G_j''(z)|=O(n),\quad \sup_{z\in I_j}|G_j'''(z)|=O(n^{4/3}).
\end{equation}
For $z\in I_j$, we will set $G(z)=G_j(z)+w_j$, where the $w_j$'s  are such that  $G(1/q)=0$ and $G$ is well-defined on the interval (for example, $w_0=-G_0(1/q)$, $w_1=-G_1(1/B)+G_0(1/B)-G_0(1/q)$ and so on). Note, from \eqref{eq:iuyiuy}, the $w_j$'s satisfy $|w_j|=O(n^{1/3})$. 

The construction of $G$ so far ensures that $G(1/q)=0$, $G$ is continuous and strictly increasing in the interval $[1/q,1]$. From \eqref{eq:iuyiuy} and the fact that the $w_j$'s satisfy  $|w_j|=O(n^{1/3})$, we also obtain that there exists a constant $M>0$ such that $G(z)\leq Mn^{1/3}$ for all $z\in [1/q,1]$. 

The main part of the argument is to specify strictly increasing functions $G_j$ so that:
\renewcommand{\theenumi}{\roman{enumi}}
\begin{enumerate}
\item \label{it:cond1} The properties \eqref{eq:Gspecify} and  \eqref{eq:iuyiuy} hold. 
\item \label{it:cond2} $G$ is three-times differentiable. 
\end{enumerate} 
Provided that these conditions are met, we obtain that the function $G$ also satisfies \eqref{eq:Gders123} (which completes the proof of the lemma). The roadmap of the construction is as follows:
\renewcommand{\theenumi}{\arabic{enumi}}
\begin{enumerate}
\item \label{it:G1G4} We first specify the functions $G_1,G_4$. In particular, we will have 
\begin{equation}\label{eq:G1G4construction}
G_1'(z)=1/(z-F(z)) \mbox{ for } z\in I_1, \quad  G_4'(z)=1/(z-F(z)) \mbox{ for } z\in I_4.
\end{equation}
$G_1,G_4$ are strictly increasing three-differentiable functions which also satisfy \eqref{eq:iuyiuy}.
\item \label{it:G0} The derivatives of $G_0$ at $z=1/B$ need to match the derivatives of $G_1$ at $z=1/B$, i.e.,
\begin{equation}\label{eq:G0ders1}
G_0'(1/B)=G_1'(1/B),\quad G_0''(1/B)=G_1''(1/B), \quad G_0'''(1/B)=G_1'''(1/B).
\end{equation}
We will see that $G'_1(1/B)$, $G_1''(1/B)$, $G_1'''(1/B)$ are constants that do not depend $n$. Thus, $G_0$ can be chosen to be a function that does not depend on $n$ whatsoever; any strictly increasing three-times differentiable  function which satisfies \eqref{eq:G0ders1} will do. This yields that $G_0$ in fact satisfies the following bounds (which are stronger than those given in \eqref{eq:iuyiuy}):
\begin{equation}\label{eq:G0construction}
\max_{z\in I_0}|G_0(z)|=O(1), \quad \max_{z\in I_0}|G_0'(z)|=O(1),\quad \max_{z\in I_0}|G_0''(z)|=O(1), \quad  \max_{z\in I_0}|G_0'''(z)|=O(1).
\end{equation} 
\item \label{it:G3} The function $G_3(z)$ will be chosen to be quadratic. The requirement \eqref{eq:G3ders1} will thus completely specify $G_3$ (up to an additive constant). We will see that $G_4''(z_+)$ is negative, so the function $G_3$ will be concave. Our goal here is to ensure that for constants $C_2,C_3>0$ it holds that
\begin{gather}
G_3''(z)\leq -C_2n \mbox{ for } z\in I_3,\label{eq:G3construction}\\
G_3'(z_+)=G_4'(z_+),\quad G_3''(z_+)=G_4''(z_+).\label{eq:G3ders1}
\end{gather}
Note that $G_3$ is strictly increasing (since $G_3'(z_+)=G_4'(z_+)>0$ and $G_3'$ is decreasing from \eqref{eq:G3construction}) and three-times differentiable (since $G_3$ is quadratic). $G_3$ will also satisfy \eqref{eq:iuyiuy}.
\item \label{it:G2} The function $G_2$ will satisfy the following constraints (in addition to \eqref{eq:G2ders1}):
\begin{gather}
G_2'(z)\geq C_1n^{2/3},\quad |G_2''(z)|\leq (10^2C_1/L) n \mbox{ for  } z\in I_2,\label{eq:G2construction}\\
G_2'(z_-)=G_1'(z_-),\quad G_2''(z_-)=G_1''(z_-),\label{eq:G2ders1}\\
G_2'(z_m)=G_3'(z_m),\quad G_2''(z_m)=G_3''(z_m),\label{eq:G2ders2}
\end{gather}
where $C_1$ is a positive constant. Note that $G_2$ is clearly strictly increasing (from \eqref{eq:G2construction}). $G_2$ will also be three-times differentiable and it will satisfy \eqref{eq:iuyiuy}.

We will see that $G_1'(z_-)<G_3'(z_m)$ and $G_1''(z_-)>0>G_3''(z_m)$. Recall that we also need that the first derivative of $G_2$ is positive. Thus, the first derivative $G_2'$ will increase overall in the interval $I_2$, yet at the same time $G_2'$ should change monotonicity at some point inside the interval. 
\end{enumerate}
Let us assume for now that the functions $G_j$ satisfy all of the Items~\ref{it:G1G4}---\ref{it:G2} and conclude that the function $G$ satisfies Conditions~\ref{it:cond1} and~\ref{it:cond2}. For Condition~\ref{it:cond1}, first observe that \eqref{eq:iuyiuy} is satisfied for all $j\in\{0,1,2,3,4\}$ by Items~\ref{it:G1G4}---\ref{it:G2}. Also, equations \eqref{eq:G1G4construction}, \eqref{eq:G0construction}, \eqref{eq:G3construction} and \eqref{eq:G2construction} show that $G$ satisfies \eqref{eq:Gspecify}. This proves that $G$ satisfies Condition~\ref{it:cond1}. Relative to Condition~\ref{it:cond1}, using \eqref{eq:G0ders1},  \eqref{eq:G3ders1}, \eqref{eq:G2ders1}, \eqref{eq:G2ders2} and the three-times differentiability of the $G_j$'s, we have that $G$ is two-times continuously differentiable with a  third derivative which exists everywhere apart (possibly) from the points $z=z_-,z_m,z_+$. For each of these points, we interpolate $G'''$ in an (infinitesimally) small  neighborhood of the point using a steep linear function; the use of the linear function guarantees  that the order of $G'''$ is still $O(n^{4/3})$. The infinitesimally small length of the interpolation interval guarantees  that the effect on $G,G',G''$ by this modification of $G'''$ can safely be ignored. It follows that $G$ satisfies Conditions~\ref{it:cond1} and~\ref{it:cond2}, as wanted. 

It remains to obtain Items~\ref{it:G1G4}---\ref{it:G2}. We start with Item~\ref{it:G1G4}.

To specify the functions $G_1$ and $G_4$, first consider a function $h$ on the interval $I_1\cup I_4$ which satisfies $h(1/B)=h(z_+)=0$ and $h'(z)=1/(z-F(z))$ for $z\in I_1\cup I_4$. This well-defines $h$ on $I_1\cup I_4$. We then set $G_1(z)=h(z)$ for $z\in I_1$ and  $G_4(z)=h(z)$ for $z\in I_4$. For $z\in I_1\cup I_4$, note that $z> F(z)$ (using that $z\neq a$ and Lemma~\ref{lem:Flemma}) and thus $h'(z)>0$, so $G_1$ and $G_4$ are stricly increasing. It remains to show \eqref{eq:iuyiuy} for $j=1,4$. 

Note that 
\begin{equation}\label{eq:derh}
h''(z)=\frac{F'(z)-1}{(z-F(z))^2}, \quad h'''(z)=\frac{2(F'(z)-1)^2+F''(z)(z-F(z))}{(z-F(z))^3}.
\end{equation}
Let $c:=-F''(a)/2$. By Lemma~\ref{lem:Flemma}, we have that $c>0$. By Taylor's theorem, we have that for all sufficiently small  $\epsilon>0$, for all $z$ in the interval $I:=(a-\epsilon,a+\epsilon)$, it holds that 
\begin{equation}\label{eq:Fexpansion}
F(z)=z-c(z-a)^2+R_3(z)
\end{equation}
for a remainder function $R_3(z)$ which satisfies $\max_{z\in I}|R_3(z)|=O( |z-a|^3)$. From \eqref{eq:Fexpansion}, it also follows that 
\begin{equation*}
F'(z)=1-2c(z-a)+R_2(z), \qquad F''(z)=-2c+R_1(z),
\end{equation*}
for  remainder functions  $R_1(z),R_2(z)$ which satisfy $\max_{z\in I}|R_1(z)|= O(|z-a|)$ and $\max_{z\in I}|R_2(z)|= O(|z-a|^2)$. We thus obtain that there exist constants $U_1,U_2,U_3>0$ such that for $z\in I\backslash \{a\}$, it holds that 
\begin{equation}\label{eq:dersbounds}
\begin{gathered}
\Big|\frac{1}{z-F(z)}-\frac{1}{c}(z-a)^{-2}\Big|\leq U_1|z-a|^{-1}, \quad \Big|\frac{F'(z)-1}{(z-F(z))^2}+\frac{2}{c}(z-a)^{-3}\Big|\leq U_2|z-a|^{-2},\\
\Big|\frac{2(F'(z)-1)^2+F''(z)(z-F(z))}{(z-F(z))^3}-\frac{6}{c}(z-a)^{-4}\Big|\leq U_3|z-a|^{-3}.
\end{gathered}
\end{equation}
Using \eqref{eq:derh} and \eqref{eq:dersbounds}, it is immediate  to show that $\max_{z\in I_1\cup I_4} |h'(z)|=O(n^{2/3})$, $\max_{z\in I_1\cup I_4} |h''(z)|=O(n)$, $\max_{z\in I_1\cup I_4}|h'''(z)|=O(n^{4/3})$ and thus these bounds carry over to $G_1,G_4$ as well. We next show that  $\max_{z\in I_1} h(z)=O(n^{1/3})$, the proof for $\max_{z\in I_4} h(z)=O(n^{1/3})$ being completely analogous.

In the interval $z\in[1/B,a-\epsilon]$, we have that $h'$ is bounded above by an absolute constant throughout the interval, so we clearly have that $h(a-\epsilon)-h(1/B)=O(1)$. Consider next $z\in (a-\epsilon,z_-)$, and parameterize $z$ as $z=a-Kn^{-1/3}$ for some $K$ which satisfies $L_-<K<\epsilon n^{1/3}$. Using \eqref{eq:dersbounds}, we have the  bound
\[h'(z)\leq\frac{1+\epsilon U_1}{cK^2} n^{2/3}.\]
Thus
\begin{eqnarray*}
h(z_-)-h(a-\epsilon)
&=&
\int^{z_-}_{\alpha-\epsilon}h'(z)dz=n^{-1/3}\int^{\epsilon n^{1/3}}_{L_-}h'(a-Kn^{-1/3})dK
\\
&\leq &
 (1+\epsilon U_1)n^{1/3}\int^{\epsilon n^{1/3}}_{L_-}\frac{1}{cK^2}dK\leq M n^{1/3},
 \end{eqnarray*}
for some absolute constant $M$. This concludes the construction for  Item~\ref{it:G1G4}.

For Item~\ref{it:G0}, we only need to show that $G_1'(1/B),G_1''(1/B),G_1'''(1/B)$ are constants. This is clear for $G_1'(1/B)$ which is equal to $h'(1/B)=1/(1/B-1/q)$;  for  $G_1''(1/B)$ and $G_1'''(1/B)$, it follows from the expressions in \eqref{eq:derh} (note, using the method in Lemma~\ref{lem:jacobianuniform}, one can show that the right derivative of $F$ at $1/B$ is equal to $2(q-1)/q$, while the right second  derivative of $F$ at $1/B$ is equal to $-\frac{4B(q-1)}{3q}$). This yields  Item~\ref{it:G0}.

For Items~\ref{it:G3} and~\ref{it:G2}, we will need  the values of the derivatives of $G_1$ and $G_4$ at the points $z_-$ and $z_+$, respectively.  Set
\begin{equation*}
d_-':= G_1'(z_-),\quad d_-'':= G_1''(z_-),\quad d_+':= G_4'(z_+),\quad d_+'':= G_4''(z_+).
\end{equation*}
From the first two inequalities in \eqref{eq:dersbounds}, we obtain
\begin{equation}\label{eq:limits4545}
\lim_{n\rightarrow \infty} \frac{d_\pm'}{n^{2/3}}=\frac{1}{cL^2_\pm},\quad
 \lim_{n\rightarrow \infty} \frac{d_-''}{n}=\frac{2}{cL^3_-},\quad \lim_{n\rightarrow \infty} \frac{d_+''}{n}=-\frac{2}{cL^3_+}.
\end{equation}
From \eqref{eq:limits4545}, we obtain that for all sufficiently large $n$, there exist $D_\pm', D_\pm''>0$ such that
\[d_\pm'= D_\pm' n^{2/3},\quad d_-''= D_-'' n,\quad d_+''= -D_+'' n,\]
and
\begin{equation}\label{eq:Dpmaccent}
\frac{1}{cL^2_\pm}(1-10^{-5})\leq D_\pm'\leq (1+10^{-5})\frac{1}{cL^2_\pm},\quad \frac{2}{cL^3_\pm}(1-10^{-5})\leq D_\pm''\leq (1+10^{-5})\frac{2}{cL^3_\pm}.
\end{equation}
Note that $D_\pm',D_\pm''$ depend on $n$, but as \eqref{eq:Dpmaccent} shows they satisfy $D_\pm',D_\pm''=\Theta(1)$. 

We are now ready to show Item~\ref{it:G3}. For $z\in I_3$, we will set $G_3(z)=u_1 n^{2/3}(z-a)+u_2n(z-a)^2$ for $u_1,u_2$ which we next specify. To satisfy \eqref{eq:G3ders1}, we will choose
\begin{equation}\label{eq:u1u2choice}
2u_2=-D_+'',\quad  u_1 +2u_2L_+=D_+'. 
\end{equation}
Observe that $u_2<0$, so $G_3$ is not only a quadratic function but also concave. Note that $u_1,u_2$ satisfy $|u_1|,|u_2|=\Theta(1)$ from where it easily follows that \eqref{eq:iuyiuy} is satisfied (for $j=3$). For \eqref{eq:G3construction}, just observe that $G_3''(z)=2u_2=-D_+''$ and hence the bound on $G_3''$ follows from \eqref{eq:Dpmaccent}. This completes the construction for Item~\ref{it:G3}.

 For the construction in Item~\ref{it:G2}, we will need a handle of the derivatives of $G_3$ at the endpoint $z_m$ of the interval $I_3$ (we will also use these later in the construction for Item~\ref{it:G2}). Let 
\begin{equation*}
\begin{gathered}
D_m':=G_3'(z_m)/n^{2/3}\mbox{ and }D_m'':=-G_3''(z_m)/n.
\end{gathered}
\end{equation*}
We will show that 
\begin{equation}\label{eq:derivativesatzm}
\begin{gathered}
D_m'=(1\pm10^{-4})\frac{3}{cL^2_\pm},\quad D_m''= (1\pm 10^{-4})\frac{2}{cL^3_\pm}.\\
\end{gathered}
\end{equation} 
By the definition of $D_m'$, we have that  $D_m'=u_1-2u_2 L_m$ and hence, by the choice \eqref{eq:u1u2choice} of $u_1,u_2$, we have 
\[D_m'=D'_{+}+ D_+''(L_++L_m).\]
Also, we have $D_m''=D_+''$ since the function $G_3$ is quadratic. It is immediate thus to conclude \eqref{eq:derivativesatzm} using the bounds in \eqref{eq:Dpmaccent} and $L_\pm\gg L_m$.

We are now ready to give the construction for Item~\ref{it:G2}. To define the function $G_2(z)$ on the interval $I_2$, we will set
\[G_2(z)=n^{1/3}\, g\big(n^{1/3}(z-a)\big),\]
where $g$ is a three times differentiable  function on the interval $I:=[-L_-,-L_m]$ such that
\begin{gather}
g'(-L_-)=D_-',\quad g''(-L_-)=D_-'', \label{eq:diffwell1}\\
g'(-L_m)=D_m', \quad g''(-L_m)=-D_m'', \label{eq:diffwell2}\\
\min_{x\in I} g'(x)\geq  \frac{1}{2cL^2_-},\quad \max_{x\in I}|g''(x)|\leq \frac{25}{cL^3_+}.\label{eq:difflowersecond}
\end{gather}
Equations \eqref{eq:diffwell1}, \eqref{eq:diffwell2} and \eqref{eq:difflowersecond} ensure that the function $G_2$ satisfies \eqref{eq:G2construction}, \eqref{eq:G2ders1} and \eqref{eq:G2ders2}. Also it will be clear from the specification of $g$ that all of $g,g',g'',g'''$ are bounded by absolute constants, which thus implies that $G_2$ satisfies \eqref{eq:iuyiuy} (for $j=2$).

It remains to specify such a function $g$, we do this by specifying its second derivative. More precisely, for $z\in I$, we will set
\begin{equation}\label{eq:gdefviah}
g'(z):=D'_-+\int^{z}_{-L_-}h(x)dx, \mbox{ so that } g''(z)=h(z),
\end{equation}
where $h(z)$ is a differentiable function on $I$ satisfying 
\begin{gather}
h(-L_-)= D_-'', \quad h(-L_m)=-D_m'', \quad  \int^{-L_m}_{-L_-}h(x)dx=D'_m-D'_-,\label{eq:hbhbhb1} \\
\max_{x\in I} |h(x)|\leq \frac{25}{cL_+^3},\quad  \int^{z}_{-L_-}h(x)dx\geq 0\mbox{ for all } z\in I.\label{eq:hbhbhb2}
\end{gather}
Using \eqref{eq:hbhbhb1} and \eqref{eq:hbhbhb2}, it is immediate to verify that the function $g$, as specified in \eqref{eq:gdefviah}, satisfies \eqref{eq:diffwell1}, \eqref{eq:diffwell2} and \eqref{eq:difflowersecond} (for the first inequality in \eqref{eq:difflowersecond}, note that $g'(z)\geq D_-'$ for all $z\in I$ and then use the  bound for $D_-'$ from \eqref{eq:Dpmaccent}).

To specify the function $h$, we will need two parameters $K_1,K_2>0$ such that 
\begin{equation}\label{eq:K1K2}
-L_-< -K_1< -K_2< -L_m.
\end{equation} 
We will specify the parameters $K_1,K_2$ shortly but for now it will be more instructive to assume that $K_1,K_2$ just satisfy \eqref{eq:K1K2}; the freedom to specify $K_1,K_2$ will be helpful at a slightly later point. 

So, consider the function $h$ defined on $[-L_-,-L_m]$ by
\[h(z)=\begin{cases} \frac{ D''_-}{(L_--K_1)^2}(z+K_1)^2,&\mbox{if } -L_-\leq z\leq -K_1\\[0.2cm]
\frac{100(D'_m-D'_-)}{3(K_1-K_2)^5}(z+K_1)^2(z+K_2)^2,&\mbox{if } -K_1<z<-K_2\\[0.2cm]
\frac{- D''_m}{(K_2-L_m)^2}(z+K_2)^2, &\mbox{if } -K_2\leq z\leq -L_m\end{cases}\]
Note that 
\begin{equation}\label{eq:hLmLmDmDp}
h(-L_m)=-D''_m,\quad  h(-L_-)=D''_-,
\end{equation} and that $h$ is differentiable throughout the interval $[-L_-,-L_m]$ since at the points $z=-K_1,-K_2$ it holds that $h(-K_1)=h(-K_2)=h'(-K_1)=h'(-K_2)=0$. Further, by a direct calculation, the function $h$ satisfies the following:
\begin{equation}\label{eq:integralsofh}
\int^{-K_1}_{-L_-}h(z)dz=\frac{ D''_-}{3}(L_--K_1),\ \int^{-K_2}_{-K_1}h(z)dz=\frac{10}{9}(D_m'-D_-'), \ \int^{-L_m}_{-K_2}h(z)dz=-\frac{D_m''}{3}(K_2-L_m).
\end{equation}
Since $D_m'>D_-'>0$ and $D''_-,D''_m>0$ (cf. \eqref{eq:Dpmaccent} and  \eqref{eq:derivativesatzm}), we also have that 
\begin{equation}\label{eq:hmonotconvex}
\begin{gathered}
 0\leq h(z)\leq D''_- \mbox{ for } z\in [-L_-,-K_1],\\
 0\leq h(z)\leq \frac{100(D'_m-D'_-)}{48(K_1-K_2)} \mbox{ for } z\in (-K_1,-K_2),\\
-D''_m\leq h(z)\leq 0 \mbox{ for } z\in[-K_2,-L_m].
\end{gathered}
\end{equation}
It follows that 
\begin{equation}\label{eq:maxhboundM}
\max_{z\in I}|h(z)|\leq M, \mbox{ where } M:=\max\Big\{ D''_-,D''_m, \frac{3(D'_m-D'_-)}{K_1-K_2}\Big\}.
\end{equation}

It remains to choose $K_1,K_2$ satisfying \eqref{eq:K1K2} so that the specifications for $h$ in \eqref{eq:hbhbhb1} and \eqref{eq:hbhbhb2} are satisfied. We set
\begin{equation}\label{eq:K1K2tgecspec}
K_1=L_--\frac{D_-'}{3D_-''}, \quad K_2=L_m+\frac{D_m'}{3D_m''}.
\end{equation}
Using \eqref{eq:Dpmaccent} and \eqref{eq:derivativesatzm} and $L_\pm\gg L_m$, we have 
\begin{equation}\label{eq:K1K2approx}
K_1=\Big(\frac{5}{6}\pm 10^{-3}\Big)L_-,\quad K_2=\Big(\frac{1}{2}\pm10^{-3}\Big)L_+.
\end{equation}
Since $L_+=L_-\gg L_m$, we obtain that $K_1,K_2$ satisfy \eqref{eq:K1K2} as desired.

We next check that the specifications for $h$ in \eqref{eq:hbhbhb1} and \eqref{eq:hbhbhb2} are satisfied. First, combining \eqref{eq:integralsofh} and \eqref{eq:K1K2tgecspec}, we obtain that    
\begin{equation}\label{eq:integralsofh2}
\int^{-L_m}_{-L_-}h(z)dz=D'_m-D'_-.
\end{equation}
Equations \eqref{eq:hLmLmDmDp} and \eqref{eq:integralsofh2} show that $h$ does indeed satisfy \eqref{eq:hbhbhb1}.

We next show that $h$ satisfies the inequalities in \eqref{eq:hbhbhb2}. To show that $\max_{z\in I}|h(z)|\leq 25/(cL_+^3)$ it suffices to show that $M\leq 25/(cL_+^3)$, where $M$ is as in \eqref{eq:maxhboundM}. This is immediate to verify using \eqref{eq:Dpmaccent}, \eqref{eq:derivativesatzm} and \eqref{eq:K1K2approx}. For the second inequality in \eqref{eq:hbhbhb2}, note from \eqref{eq:hmonotconvex} that the function $h(z)$ is non-negative when $z< -K_2$ and negative when $z>-K_2$. Thus, it suffices to check the inequality in \eqref{eq:hbhbhb2} when $z=-L_-$ and $z=-L_m$. For $z=-L_-$, the inequality holds (trivially) at equality while for $z=-L_m$ the inequality follows from \eqref{eq:integralsofh2} and $D'_m>D'_-$. This completes the construction for Item~\ref{it:G2}.

We have thus shown how to do the construction of the functions $G_0,G_1,G_2,G_3,G_4$ so that Items~\ref{it:G1G4}---\ref{it:G2} hold, completing the proof of Lemma~\ref{lem:potentialchoice}.
\end{proof}

\bibliographystyle{plain}
\bibliography{references}
\end{document}